\documentclass[11pt]{article} 
\usepackage[letterpaper,hmargin=1in,vmargin=1in]{geometry}

\def\ShowAuthNotes{1}

\usepackage[utf8]{inputenc}

\usepackage{float,url,hyperref,epic,eepic}
\usepackage{xcolor, enumerate}
\usepackage{graphicx}
\usepackage[all]{xy}

\usepackage{listings}
\usepackage{tikz}

\usepackage{multicol}

\usepackage{amsmath, amsfonts, amssymb, amsthm, amscd}
\usepackage{calc}
\usepackage{algorithmic, algorithm}

\makeatletter
\newtheorem*{rep@theorem}{\rep@title}
\newcommand{\newreptheorem}[2]{%
\newenvironment{rep#1}[1]{%
 \def\rep@title{#2 \ref{##1}}%
 \begin{rep@theorem}}%
 {\end{rep@theorem}}}
\makeatother

\newcounter{hours}
\newcounter{minutes}
\newcommand{\printtime}{ %
        \setcounter{hours}{\time/60} %
        \setcounter{minutes}{\time-\value{hours}*60} %
        \ifthenelse{\value{hours}<10}{0}{}\thehours:%
        \ifthenelse{\value{minutes}<10}{0}{}\theminutes%
        } 

\theoremstyle{plain}
        \newtheorem{theorem}{Theorem}[section]
        \newreptheorem{theorem}{Theorem}
        \newtheorem{lemma}[theorem]{Lemma}
        \newreptheorem{lemma}{Lemma}

\theoremstyle{definition}
        \newtheorem{definition}[theorem]{Definition}

\theoremstyle{remark}
        \newtheorem{remark}[theorem]{Remark}

\newcommand{\URF}{\normalfont{\textbf{F}}}
\newcommand{\VarG}{G}
\newcommand{\VarH}{H}

\newcommand{\URP}{\normalfont{\textbf{P}}}
\newcommand{\VarP}{P} 
\newcommand{\TSRF}{\normalfont{\textbf{R}}}

\newcommand{\di}{\textnormal{\textbf{D}}}
\newcommand{\si}{\textnormal{\textbf{S}}}

\newcommand{\sitwo}{\ensuremath{\mathbf{S}}}

\newcommand{\Chain}{\ensuremath{\mathcal{C}}}
\newcommand{\hist}{\ensuremath{\mathcal{H}}}

\DeclareMathOperator{\val}{val}
\DeclareMathOperator{\Next}{next}
\DeclareMathOperator{\Prev}{prev}
\newcommand{\enqNewChains}{\normalfont\textsc{enqueueNewChains}}

\newcommand{\Qempty}{\normalfont\textsc{Empty}}
\newcommand{\enqueue}{\normalfont\textsc{Enqueue}}
\newcommand{\dequeue}{\normalfont\textsc{Dequeue}}
\newcommand{\CompletedChains}{\normalfont\textrm{CompletedChains}}
\newcommand{\Finner}{\ensuremath{\textsc{F}^{\textsc{inner}}}}
\newcommand{\F}{\normalfont\textsc{F}}
\renewcommand{\P}{\normalfont\textsc{P}}
\newcommand{\adapt}{\normalfont\textsc{Adapt}}
\newcommand{\forceVal}{\normalfont\textsc{ForceVal}}
\newcommand{\evalFwd}{\normalfont\textsc{EvaluateForward}}
\newcommand{\evalBwd}{\normalfont\textsc{EvaluateBackward}}
\renewcommand{\Check}{\normalfont\textsc{Check}}

\newcommand{\BadlyHit}{\mathsf{BadlyHit}}
\newcommand{\BadP}{\mathsf{BadP}}
\newcommand{\BadlyCollide}{\mathsf{BadlyCollide}}

\newcommand{\Feistel}{\Psi}

\newcommand{\Query}{\normalfont\textsc{Query}}
\newcommand{\ChainQuery}{\normalfont\textsc{ChainQuery}}
\newcommand{\CompleteChain}{\normalfont\textsc{CompleteChain}}
\newcommand{\XorQuery}{\normalfont\textsc{XorQuery}}
\newcommand{\inu}{\ensuremath{\leftarrow_\textsf{R}}}

\newcommand{\hmax}{\ensuremath{\textsf{h}_{\max}}}

\newcommand{\leftp}{\ensuremath{|_{\text{1}}}}
\newcommand{\rightp}{\ensuremath{|_{\text{2}}}}

\newcommand{\poly}{\ensuremath{\text{poly}}}

\definecolor{DSgray}{cmyk}{0,0,0,0.7}

\definecolor{DSred}{cmyk}{0,0.7,0,0.7}

\ifnum\ShowAuthNotes=1
\newcommand{\Authornote}[2]{\noindent{\small\textcolor{DSgray}{\sf{
\textcolor{red}{[#1: #2]\marginpar{\textcolor{red}{\fbox{\Large !}}}}}}}}
\newcommand{\Authormarginnote}[2]{\marginpar{\parbox{2.2cm}{\raggedright\tiny \textcolor{red}{#1: #2}}}}
\else
\newcommand{\Authornote}[2]{}
\newcommand{\Authormarginnote}[2]{}
\fi

\begin{document}

\title{The Equivalence of the Random Oracle Model\\and the Ideal Cipher Model, Revisited}
\author{Thomas Holenstein\thanks{ETH Zurich, Department of Computer Science, 8092 Zurich, Switzerland. E-mail: {\tt thomas.holenstein@inf.ethz.ch}} \and Robin K\"unzler\thanks{ETH Zurich, Department of Computer Science, 8092 Zurich, Switzerland. E-mail: {\tt robink@inf.ethz.ch}} \and Stefano Tessaro\thanks{University of California, San Diego, Department of Computer Science \& Engineering, La Jolla, CA 92093-0404. E-mail: {\tt stessaro@cs.ucsd.edu}}}
\maketitle
\thispagestyle{empty}

\begin{abstract}
  We consider the cryptographic problem of constructing an invertible
  random permutation from a {\em public} random function (i.e., which
  can be accessed by the adversary). This goal is formalized by
  the notion of {\em indifferentiability} of Maurer et al.\ (TCC
  2004). This is the natural extension to the public setting of the
  well-studied problem of building random permutations from random
  functions, which was first solved by Luby and Rackoff (Siam J.~Comput., '88)
  using the so-called {\em Feistel} construction.

  The most important implication of such a construction is the {\em
    equivalence} of the {\em random oracle model} (Bellare and
  Rogaway, CCS '93) and the {\em ideal cipher model}, which is
  typically used in the analysis of several constructions in symmetric
  cryptography.

  Coron et al.~(CRYPTO 2008) gave a rather involved proof that the
  six-round Feistel construction with independent random round
  functions is indifferentiable from an invertible random
  permutation. Also, it is known that fewer than six rounds do not
  suffice for indifferentiability. The first contribution (and
  starting point) of our paper is a concrete distinguishing attack
  which shows that the indifferentiability proof of Coron et al.\ is
  \emph{not} correct. In addition, we provide supporting evidence that
  an indifferentiability proof for the six-round Feistel construction
  may be very hard to find.

  To overcome this gap, our main contribution is a proof that the
  Feistel construction with fourteen rounds is indifferentiable from
  an invertible random permutation. The approach of our proof relies
  on assigning to each of the rounds in the construction a unique and
  specific role needed in the proof. This avoids many of the problems
  that appear in the six-round case.

  \bigskip

\end{abstract}

\nocite{CPS08CryptoVersion, CPS08v1, CPS08v2} 

\newpage

\tableofcontents

\newpage
\setcounter{page}{1}

\section{Introduction}

\subsection{Random Functions and Permutations: The Feistel
  Construction}

Many cryptographic security proofs rely on the assumption that a
concrete cryptographic function (e.g.\ a block cipher or a hash
function) behaves as a {\em random primitive}, i.e., an {\em ideal}
object which answers queries ``randomly''. A typical example is a
\textit{random function} $\URF: \{0,1\}^m \rightarrow \{0,1\}^n$,
which associates with each $m$-bit input $x$ a uniformly distributed
$n$-bit value $\URF(x)$. We speak of a {\em random oracle} if the
domain consists of all strings of finite length, rather than all
$m$-bit ones.
A \textit{random permutation} $\URP: \{0,1\}^n \rightarrow
\{0,1\}^n$ is another example: It behaves as a uniformly-chosen
permutation from the set of all permutations on $\{0,1\}^n$, allowing
both {\em forward queries} $\URP(x)$ and {\em backward queries}
$\URP^{-1}(y)$.

Many results in cryptography can be recast as finding an explicit
construction of a random primitive from another one in a purely
information-theoretic setting.  For instance, the core of Luby and
Rackoff's seminal result~\cite{LR88} on building \emph{pseudorandom
  permutations} from \emph{pseudorandom functions} (a computational
statement) is a construction of a random permutation from random
functions via the \emph{$r$-round Feistel construction} $\Feistel_{r}$:
It implements a permutation taking a $2n$-bit input $(L_0, R_0)$ (where
$L_0, R_0$ are $n$-bit values), and the output $(L_r, R_r)$ is computed
via $r$ rounds mapping $L_i, R_i$ to $L_{i+1}, R_{i+1}$ as
\begin{displaymath}
  L_{i+1} := R_{i}, \,\,\, R_{i+1} := L_i \oplus \URF_{i+1}(R_i),
\end{displaymath}
where $\URF_1, \ldots, \URF_r: \{0,1\}^n \to \{0,1\}^n$ are so-called
\emph{round functions}. The main statement of \cite{LR88} is that if
the round functions are independent random functions, then
$\Feistel_{3}$ is \emph{information-theoretically} indistinguishable
from a random permutation which does not allow backward queries,
whereas $\Feistel_{4}$ is indistinguishable from a full-fledged random
permutation.

\subsection{The Random Oracle and Ideal Cipher Models:
  Indifferentiability}

Random primitives are frequently employed to model an idealized
cryptographic function accessible by {\em all} parties in the scenario
at hand, including the adversary. The most prominent example is the
\emph{Random Oracle Model} \cite{BR93}, where a random oracle models
an ideal hash function. Although it is known that no concrete hash
function can achieve the functionality of a random
oracle~\cite{CaGoHa04} (see also \cite{MaReHo04}), security proofs in
the random oracle model provide a common heuristic as to which schemes
are expected to remain secure when the random oracle is instantiated
with a concrete hash function. In fact, to date, many widely employed
practical schemes, such as OAEP~\cite{BelRog94}\footnote{However, we note that
  standard model instantiations of OAEP for certain classes of
  trapdoor functions exist~\cite{KiOnSm10}, even though they only achieve a weaker security notion than what provable in the random oracle model.}
and FDH~\cite{BelRog96}, only enjoy security proofs in the random
oracle model.

The {\em ideal cipher model} is another widespread model in which all
parties are granted access to an ideal cipher $\mathbf{E}:
\{0,1\}^{\kappa} \times \{0,1\}^n \to \{0,1\}^n$, a random primitive
such that the restrictions $\mathbf{E}(k, \cdot)$ for $k \in
\{0,1\}^{\kappa}$ are $2^{\kappa}$ independent random
permutations. Application examples of the ideal cipher model range
from the analysis of block-cipher based hash function constructions
(see, for example \cite{BlRoSh02}) to disproving the existence of
generic attacks against constructions such as cascade
encryption~\cite{BelRog06,GazMau09} and to studying generic
related-key attacks~\cite{BelKoh03}.

\paragraph{Equivalence of models and indifferentiability.} This paper
addresses the fundamental question of determining whether the random
oracle model and the ideal cipher model are \emph{equivalent}, where
equivalence is to be understood within a simulation-based security
framework such as~\cite{Canett01}: In other words, we aim at answering
the following two questions:
\begin{enumerate}[(1)]
\item Can we find a construction $\mathbf{C}_1$, which uses an ideal
  cipher $\mathbf{E}$, such that $\mathbf{C}_1^{\mathbf{E}}$ is ``as
  good as'' a random oracle $\mathbf{R}$, meaning that any secure
  cryptographic scheme using $\mathbf{R}$ remains secure when using
  $\mathbf{C}_1^{\mathbf{E}}$ instead?
\item Conversely, is there $\mathbf{C}_2$ such that
  $\mathbf{C}_2^{\mathbf{R}}$ is ``as good as'' an ideal cipher
  $\mathbf{E}$?
\end{enumerate}
Indistinguishability is not sufficient to satisfy the above
requirement of being ``as good as'', as the adversary can exploit access
to the underlying primitive.  Instead, the stronger notion of {\em
  indifferentiability} due to Maurer et al.\ \cite{MaReHo04} is
needed: the system $\mathbf{C}_1^{\mathbf{E}}$ is indifferentiable
from $\mathbf{R}$ if there exists a \emph{simulator}\footnote{Usually
  required to be efficient, i.e., with running time polynomial in the
  number of queries it processes} $\si$ accessing $\mathbf{R}$ such
that $(\mathbf{C}_1^{\mathbf{E}}, \mathbf{E})$ and $(\mathbf{R},
\mathbf{S}^{\mathbf{R}})$ are information-theoretically
indistinguishable. This is equivalent to stating that the adversary is
able to locally simulate the ideal cipher consistently with
$\mathbf{R}$, given only access to the random oracle and without
knowledge of the queries to $\mathbf{R}$ of the honest users. Of
course, indifferentiability generalizes to arbitrary primitives: The
definition of $\mathbf{C}_2^{\mathbf{R}}$ being indifferentiable from
$\mathbf{E}$ is analogous.\footnote{Interestingly, we cannot construct
  a {\em non-invertible} random permutation from a random oracle. This
  follows from a well-known result by Rudich~\cite{Rudich89} and Kahn et al.~\cite{KSS00}.}

\paragraph{Prior work and applications.} Question (1) above is, to
date, well understood: Coron et al.~\cite{CDMP05}, and long series of
subsequent work, have presented several constructions of random oracles
from ideal ciphers based on hash-function constructions such as the
{\em Merkle-Damg\aa rd} construction~\cite{Merkle89,Damgaa89} with
block-cipher based compression functions. In particular,
indifferentiability has become a de-facto standard security
requirement for hash function constructions, generally interpreted as
the absence of generic attacks against the construction treating the
block cipher as a black box.

In a similar vein, answering question (2) could provide new approaches
to designing block ciphers from non-invertible primitives. But in
contrast to question (1), the problem is far less understood.  Dodis
and Puniya~\cite{DodPun06} considered constructions in the so-called
{\em honest-but-curious} model, where the adversary only gets to {\em
  see} queries made by the construction to the public random function,
but is {\em not} allowed to issue queries of her choice: They showed
that $\omega(\log n)$ rounds of the Feistel construction are
sufficient to build an ideal cipher.\footnote{The notion of
  honest-but-curious indifferentiability is very subtle, as in general
  it is not even implied by full indifferentiability.} In the same
work, it was first noted that four rounds are insufficient to achieve
indifferentiability of the Feistel construction.

Finally, at CRYPTO 2008, Coron et al.\ \cite{CPS08CryptoVersion} presented a
first proof that the six-round Feistel construction $\Feistel_6$ with
independent random round functions is indifferentiable from a
random permutation,\footnote{Note that this implies a construction of an ideal cipher
  from a random oracle, as we can construct the independent round
  functions from a random oracle. Moreover, they can be keyed to
  obtain an independent cipher for each value of the key.} hence
seemingly settling the equivalence of the ideal cipher model and the random oracle model.
They also showed that five rounds are insufficient for this task.
Also, a somewhat simpler proof that the ten-round Feistel
 construction $\Feistel_{10}$ with 
independent round functions is indifferentiable from a random permutation
was later presented in \cite{Seurin09}, the PhD thesis of the last 
author of \cite{CPS08CryptoVersion}.

Following the publication of this result, the equivalence of the random oracle and ideal cipher models has been used to infer security in the random oracle model using an ideal cipher (or random permutation) as an intermediate step~\cite{DzPiWi10} and to prove impossibility of black-box constructions from block ciphers~\cite{LinZar09}.

\subsection{Our Contributions}

The surprising starting point of our work is a distinguishing attack,
outlined and analyzed in Section~\ref{sec:attack}, which shows that
the proof of \cite{CPS08v2} (the full version of \cite{CPS08CryptoVersion})
is {\em not} correct: For the simulator
given in the proof, our attack distinguishes with overwhelming
advantage. Despite hopes, at first, that we could fix the proof of
Coron et al.\ by minor modifications, we were unable to do so. In
fact, we provide a stronger attack which appears to succeed against a
large class of simulators following the natural approach of
\cite{CPS08v2}. 
We also found similar problems in the proof given in \cite{Seurin09},
and so the question of settling the equivalence of the
ideal cipher model and of the random oracle model remains open.

In order to overcome this situation, the main contribution of this paper is
a proof (given in Section~\ref{sec:proof}) that the fourteen-round
Feistel construction $\Feistel_{14}$ is indifferentiable from a
random permutation. The round number is motivated by the goal of providing a simple to
understand proof, rather than by the goal of minimizing the number of
rounds. Our proof relies on techniques which are
significantly different than the ones used in \cite{CPS08v2}.

We discuss our results in more detail in the following section.

%

\subsection{Sketch of the Previous Problems and the New Proof}
First, we discuss the basic idea of building a random permutation from a random oracle
via the $r$-round Feistel construction $\Feistel_r$. Then we discuss the problems in the 
previous proofs and finally sketch our new proof.  Some readers might 
find it helpful to consider the illustration of the Feistel construction
on page~\pageref{fig:Feistel} in the following.

\paragraph{Simulation via chain-completion.}
Since we already fixed our construction to be $\Feistel_r$, 
the core of the proof is the construction of a simulator $\si$ that
uses a given random permutation $\URP: \{0,1\}^{2n} \to \{0,1\}^{2n}$
to consistently simulate $r$ independent functions $\URF_{1},\ldots,\URF_{r}$ from
$\{0,1\}^n \to \{0,1\}^n$.  In particular, suppose that a distinguisher queries
the round functions to evaluate $\Feistel_r$ on input $x \in
\{0,1\}^{2n}$.  Then, it is required that the result matches the output of
$\URP$ on input~$x$.\footnote{Of course, much more is needed, as the
  distribution of the output needs to be indistinguishable, but
  surely, the above requirement is necessary.}
To this end, the simulator needs to somehow {\em recognize} queries belonging to such a sequence
$x_1,\ldots,x_{r}$, and to set the values
$\URF_{i}(x_i)$ to enforce consistency with $\URP$.
In the following, such sequences $x_1,\ldots,x_r$ will be called \emph{chains}.

The natural idea used by Coron et al.~is to isolate so-called \emph{partial} chains among queries made to the round functions.  An example of a partial chain is a triple
$(x_1,x_{2},x_{3})$ such that $x_{3} = x_{1} \oplus \URF_{2}(x_{2})$,
and each of $x_1$, $x_2$, and $x_3$ has previously been queried to the corresponding round
function $\URF_{i}$.  In particular, upon each query to $\URF_i$, the simulator checks whether one (or more) partial chains are created. 
When such a partial chain is detected (and some additional conditions
are met), the simulator \emph{completes} it to a (full) chain $x_1, x_2, \ldots, x_r$ such that $x_{i+1} = \URF_i(x_i) \oplus x_{i-1}$ for all $i = 2, \ldots, r - 1$, and $\URP(x_0,x_1) = x_r,x_{r+1}$, where $x_0 := \URF_1(x_1) \oplus x_2$ and $x_{r+1} = \URF_r(x_r) \oplus x_{r-1}$.  In particular, the simulator defines two consecutive values $\URF_{\ell}(x_{\ell})$ and
$\URF_{\ell+1}(x_{\ell+1})$ adaptively to satisfy all constraints.   In our example, the simulator
could complete the partial chain by first finding $x_0$, computing $x_{r}$ and $x_{r+1}$ from $\URP(x_0,x_1)$, and finally evaluate the Feistel construction
backwards, by setting each undefined $\URF_i(x_i)$ to a fresh uniform random string,  until only $\URF_{4}(x_{4})$ and $\URF_{5}(x_{5})$ are
undefined.  These two values are then defined as $\URF_4(x_4) := x_3 \oplus x_5$, and $\URF_5(x_5) := x_4 \oplus x_6$. We refer to this step as {\em adapting} the output values of $x_4$ and $x_5$.

At this point, one faces (at least) two possible problems: 
\begin{enumerate}[(i)]
\item The simulator defines new values at chain completion, and
  may keep producing new partial chains while it completes chains,
  hence potentially running forever.  Coron et al.~solve this problem
  very elegantly by a smart decision of which partial chains are completed.
  Then, they are able to show that the recursion stops after
  at most $\poly(q)$ steps, where $q$ is the number of queries the
  distinguisher makes to the permutation.  We use their strategy in
  our proof, even though in the simplified version of \cite{Seurin09}, as we detect
  fewer chains.
\item The\label{page:probleminsimulator} simulator may try to adapt $\URF_{\ell}(x_\ell)$ to some
  value, even though
  $\URF_{\ell}(x_{\ell})$ has been fixed to a different value before.
  In this case, the simulator by Coron et al.~aborts, and it hence becomes
  necessary to show that no distinguisher can make the simulator abort
  except with negligible probability.
\end{enumerate}

\paragraph{Breaking previous simulators.}

Unfortunately, the proof given in \cite{CPS08v2} does not solve (ii)
above.  In fact, it is possible to find a
sequence of queries such that the simulator, with high probability,
attempts to change a value of a previously fixed $\URF_{i}(x)$,
and aborts.  We
provide an intuition of the attack in Section~\ref{sec:attack_cps08v2}.  
A full proof that that our attack breaks the simulator is
 contained in Appendix~\ref{app:analysisattack}.

We formally prove that our attack distinguishes with overwhelming advantage. However, in view of the complexity of the considered random experiments, we have also decided to gain extra confidence in the correctness of our proof by simulating the setting of the attack.  We therefore
implemented the simulator from \cite{CPS08v2} in Python, and then used
our distinguisher on it.  The results confirm our theoretical analysis.  The code
is included as an ancillary file in the full version of this paper~\cite{HKT10full}, and is available for download.

We also point out that the proceedings version of  \cite{CPS08v2}, as well as an
earlier version available on the eprint archive \cite{CPS08v1},
presented a significantly simpler simulator.  However, it suffered from the same problem, and a simpler attack was possible.  We assume that the authors were aware of this problem
when they modified the simulator, as some modifications appear to
specifically rule out some of the attacks we found against the simpler
simulator.  However, this is speculation: no explanation is given in
\cite{CPS08v2}.

In the 10-round case, Seurin gives a much simpler simulator in his PhD thesis~\cite{Seurin09}.  At present, we do not know whether this simulator can be attacked.  

\paragraph{Problems with the previous proofs.}

Given our attack, it is natural to ask where the proof given in~\cite{CPS08v2} fails.  We explore
this question in Section~\ref{sec:ProblemsInExistingProof}, but we can
give a short explanation here.  Consider the example above.  When the
simulator attempts to define the value of $\URF_{5}(x_5)$, the proof
assumes that it can do so, because earlier on, $\URF_{6}(x_6)$ was
chosen uniformly at random, and $x_5$ was set to be $x_5 := x_7 \oplus
\URF_{6}(x_6)$.  The hope is that this implies that in the meanwhile $\URF_5(x_5)$ has not been defined, except with very small probability.  Unfortunately, between the
moment where $\URF_{6}(x_6)$ was chosen uniformly, and the moment
where $\URF_{5}(x_5)$ needs to be defined, the simulator may have completed
a large number of other partial chains.  This can destroy our expectation
completely, and indeed, our attack does exploit this fact.  We cannot hope to
complete each detected chain immediately when it is detected: For example, the definition
of a single function value may cause that we detect many new chains at the same time. These chains have to be
completed in some order, and thus there exist chains such that between their detection
and their completion many other function values are defined. This means that it is not obvious how 
to solve this problem.

Furthermore, while we do not know whether the simulator in
\cite{Seurin09} can be attacked, the problems with the proof we
describe here are present in \cite{Seurin09} as well. 

\paragraph{Further problems with previous proofs.}
There are, in fact, further problems with the previous proofs
\cite{CPS08v2,Seurin09}.  All previous proofs reduced the task of
proving indifferentiability to the task of upper bounding the abort
probability of the given simulator.  Yet, it turns out that this
reduction is quite delicate as well. In fact, both proofs of
\cite{CPS08v2} and \cite{Seurin09} have several gaps in this part,
which we were not able to fill directly.\footnote{In very broad terms,
  both proofs present a step where an ideal permutation is replaced by
  the Feistel construction, and values of the round functions are set
  by the evaluation of the construction: While each of the proofs
  presents a different approach how this is done, neither of them
  presents a convincing argument of why this modification does not
  affect the input-output behavior.}  Thus, we give a completely new
proof for this part as well.

\paragraph{Ideas we use from the previous proofs.}
Since evidence points towards the fact that
simulating a 6-round Feistel construction is difficult, we consider
the simulator for the 10-round construction used in \cite{Seurin09}, which is
significantly simpler and much more elegant.
Even though our simulator is for $14$ instead of $10$ rounds,
it is similar to the one in \cite{Seurin09}: the zones where we detect and adapt chains are analogous.  
This allows us to reuse the elegant idea of \cite{Seurin09} for bounding the simulator's running time. 

\paragraph{Intuition of our proof.}
In order to explain the main new ideas we use, we first give a more complete 
sketch of our simulator and our proof.  Of course, many details are
omitted in this sketch.

As in previous ideas, our simulator detects chains, and completes them.
In order to detect chains, we follow \cite{Seurin09} and 
use special \emph{detect zones}, where the simulator detects new chains.
Also, as in \cite{Seurin09}, we have \emph{adapt zones}, in which
the simulator fixes the values of $\URF_i(x_i)$ such that the produced
chain matches the given permutation $\URP$.
Unlike before, we use \emph{buffer rounds} (namely rounds 3,6,9, and 12)
between the zones where chains are detected, and
the zones where values are adapted.  
The function values in the buffer rounds are always defined 
by setting them to uniform random values. 
The figure on page~\pageref{fig:Feistel} has these zones marked.

We now discuss what happens
when the simulator detects a new chain with values $(x_1,x_2,x_{13},x_{14})$,
and suppose that the simulator decides to adapt the resulting chain
at positions $4$ and $5$.
Because of the way the simulator chooses the adapt zone to use, 
it is not extremely hard to show that at the moment this chain is detected,
the values $\URF_3(x_3)$ and $\URF_{6}(x_6)$ in the buffer rounds
around this adapt zone have not been defined yet, where $x_3$ and
$x_6$ are the values corresponding to round $3$ and $6$ of the detected
chain (and similar statements are proven in \cite{Seurin09}).

The hope at this point is that $\URF_3(x_3)$ (and also $\URF_6(x_6)$, but
let us concentrate on $\URF_3(x_3)$) is still unset when the simulator
is ready to complete this chain.
Intuitively, this should hold at least in case the function
values $\URF_2$ and $\URF_4$ are set at random, because
then, the simulator should only run into trouble if
some kind of unlikely collision happens (it turns out later
that actually it is not necessary to always set $\URF_4$ at random, 
but the intuition why this holds is somewhat advanced).

In order to prove that indeed this hope holds, we first use a queue
to order the chains which the simulator detects and completes.
This ensures that when it detects the chain $C=(x_1,x_2,x_{13},x_{14})$
above, any chain during completion of 
which the simulator could possibly define
$\URF_3(x_3)$ is defined before the simulator detects $C$.

Next, we define a bad event $\BadlyCollide$, which we show to occur if 
our hope fails.  
To understand the main idea of event $\BadlyCollide$, consider the chain
$C$ again.  
Even though $\URF_3(x_3)$ is not yet defined when the chain is detected,
the value $x_3$ is already fixed at this point.
The event $\BadlyCollide$ occurs only if such a value appears
in some other chain due to some unlikely collision (the
``unlikely collision'' part is crucial: in
general a distinguisher can set up new chains which contain such values,
in which case $\BadlyCollide$ should not occur).

In total, the above shows that our simulator never aborts (in contrast to 
the one given in \cite{CPS08v2}).  
Of coures, it is still necessary to show that the result is indifferentiable
from a Feistel construction.

To see why this can be difficult, consider a distinguisher which 
first queries the given permutation~$\URP(x_0,x_1)$, giving
values $(x_{14},x_{15})$.  The distinguisher then checks (say) the
first bit of $x_{14}$, and depending on it, starts evaluating
the simulated Feistel construction from the top 
with the input values $(x_0,x_1)$, or from the bottom with 
values $(x_{14},x_{15})$.  Inspection of our simulator reveals that
the choice of the adapt zone of the simulator then depends on the
first bit of $x_{14}$.

The problem which now comes in is that the randomness inherent 
in $(x_{14},x_{15})$ is needed in order to show that the values of
$\URF$ in the adapt zones look random.
However, conditioned on using the upper adapt zone, one bit of 
$x_{14}$ is already fixed.

In order to solve this problem, we take the following, 
very explicit approach: we consider the
two experiments which we want to show to behave almost the same, and
define a map associating randomness in one experiment to randomness in
the other experiment.  We then study this map.  This leads to a more
fine-grained understanding and a much more formal treatment of the
indistinguishability proof.\footnote{For reference: this step can be
  found in Section~\ref{sec:theMappingTau}.}

\subsection{Model and Notational Conventions}

The results throughout this paper are information-theoretic, and
consider random experiments where a {\em distinguisher} $\mathbf{D}$
interacts with some given system $\mathbf{S}$, outputing a value
$\mathbf{D}(\mathbf{S})$. In the context of this paper, such systems
consist of the composition $\mathbf{S} = (\mathbf{S}_1, \mathbf{S}_2)$
of two (generally correlated) systems accessible in parallel, where
$\mathbf{S}_i$ is either a random primitive (such as a random function
$\URF$, a random permutation $\URP$ defined above), or a construction
$\mathbf{C}^{\mathbf{S}}$ accessing the random primitive
$\mathbf{S}$. The advantage
$\Delta^{\mathbf{D}}(\mathbf{S}, \mathbf{S}')$ of a distinguisher
$\mathbf{D}$ in distinguishing two systems $\mathbf{S}$ and
$\mathbf{S}'$ is defined as the absolute difference $\left|
  \Pr[\mathbf{D}(\mathbf{S}) = 1] - \Pr[\mathbf{D}(\mathbf{S}') =
  1]\right|$.

We dispense to the largest extent with a formal definition of such
systems (cf.\ e.g.\ the framework of Maurer~\cite{Maurer02} for a
formal treatmenet). Most systems we consider will be defined formally
using pseudocode in a RAM model of computation, following the approach
of~\cite{BelRog06,Shoup04}.  The time complexity of a
system/distinguisher is also measured with respect to such a model.

Defining indifferentiability is somewhat subtle, as different
definitions \cite{MaReHo04,CDMP05} are used in the literature.
Furthermore, it turns out that our simulator runs in polynomial time
only with overwhelming probability (which is a bit weaker than giving
a worst-case polynomial bound on its running time). In particular, we meet the following
definition, which implies the original definition in \cite{MaReHo04}:
\begin{definition}
  For a construction $\mathbf{C}$ accessing independent random
  functions $\URF= (\URF_1, \ldots, \URF_r)$,\footnote{Such a tuple
    can also be seen as a random primitive.}  we say that
  $\mathbf{C}^{\URF}$ is {\em indifferentiable} from a random
  permutation $\URP$ if there exists a simulator $\si$
  such that for all polynomially bounded $q$, the advantage
  $\Delta^{\mathbf{D}}((\mathbf{C}^{\URF}, \URF), (\URP,
  \mathbf{S}^{\URP}))$ is negligible for all distinguishers
  $\mathbf{D}$ issuing a total of at most $q$ queries to the
  two given systems, and furthermore there exists a fixed polynomial $p(q)$, 
  such that $\si$ runs in time $p(q)$ except with negligible probability.
\end{definition}

Finally, we warn the reader that the notation used in Section~\ref{sec:attack}
differs strongly from the one used in Section~\ref{sec:proof}.  The
reason is that first in Section~\ref{sec:attack} we aim to stay close
to the notation used in \cite{CPS08v2} (with some minor modifications
which we believe enhance readability).  Unfortunately this
notation has some obvious problems when as many as $14$ rounds are
used, which is why we cannot use it in Section~\ref{sec:proof}.

\section{The Six-Round Feistel Construction: An Attack}
\label{sec:attack}

This section presents several problems in the existing proof of Coron
et al.~\cite{CPS08v2}. While we cannot rule out the fact that the
six-round Feistel construction is indeed indifferentiable from a
random permutation, the contents of this section show that any such
proof would be significantly more involved than the one given
in~\cite{CPS08v2}.

\subsection{The Simulator of Coron et al.}
\newcommand{\threechain}[4]{((#1, #2, #3), #4)}

We first provide a high-level description of the simulator $\sitwo$
used in the indifferentiability proof of \cite{CPS08v2} for the
six-round Feistel construction. This description is sufficient to
convey the main ideas underlying our attack and the problems with the
existing proof, and a complete description is given in
Appendix~\ref{appendixSimulatorDefinition}. For ease of reference, we
use a similar notation to the one of \cite{CPS08v2} throughout this
section.

Recall that the simulator $\sitwo$ queries $\URP$ and $\URP^{-1}$ to
simulate the round functions $\URF_1, \ldots, \URF_6$ consistently,
where the given $\URP$ needs to have the same behaviour as the 
constructed six-round Feistel construction. 
For each $i \in \{1, \ldots, 6\}$, the simulator stores
the values $\URF_i(x)$ it has defined up to the current point of the
execution as a set $\hist(\URF_i)$ of pairs $(x, \URF_i(x))$, called
the {\em history} of $\URF_i$.  We will write $x \in \URF_i$ if
$\URF_i(x)$ is defined. At any point in time, the simulator
considers so-called $3$-chains, which are triples of values appearing
in the histories of three consecutive round functions and which are
consistent with the evaluation of the Feistel construction. In the
following, when we refer to round $i - 1$ or $i + 1$, addition and
subtraction are modulo $6$, and the elements $\{1,\ldots,6\}$
represent the equivalence classes.

\begin{definition}[$3$-chain]
  A {\em $3$-chain} is a triple $(x, y, z)$ (where the values are
  implicitly associated with three consecutive rounds $i-1$, $i$, and
  $i + 1$) which satisfies one of the following conditions with
  respect to the given histories:
\begin{enumerate}[(i)]
\item If $i \in \{2, 3, 4, 5\}$, $x \in \URF_{i-1}$, $y \in
  \URF_{i}$ and $z \in \URF_{i+1}$, and $\URF_{i}(y)
  = x \oplus z$;
\item If $i = 6$, $x \in \URF_5$, $y \in \URF_6$, $z
  \in \URF_1$, and $\exists x_0 \in \{0,1\}^n: \URP^{-1}(y\|\URF_6(y) \oplus x)= x_0\|z$;
  \item If $i = 1$, $x \in \URF_6$, $y \in \URF_1$, $z \in \URF_2$,
    and $\exists x_7 \in \{0,1\}^n:\URP(z \oplus
    \URF_1(y)\|y) = x\|x_7$.
\end{enumerate}
\end{definition}

We next describe the main points how the simulator attempts to
simulate the round functions. 
On a query $x \in
\URF_i$ for the $i$-th round function $\URF_i$, the simulator $\sitwo$ 
replies with $\URF_i(x)$.  If $x \notin \URF_i$, 
the simulator assigns to $\URF_i(x)$ a
uniform random value, and invokes a procedure called \ChainQuery\ with input
$(x, i)$. 

The procedure \ChainQuery\ operates as follows. Let $\Chain(+, x, i)$
and $\Chain(-, x, i)$ be the sets of all $3$-chains with $x$ in round
$i$ as their first value (these are so-called {\em positive chains})
and as their last value (so-called {\em negative chains}),
respectively. The procedure iterates over all 
$3$-chains $ \Chain(+, x, i) \cup \Chain(-, x, i)$, and for some subset 
of these chains it calls the procedure \CompleteChain. 
How the simulator chooses this subset is not important for this discussion.

The procedure \CompleteChain\  ensures consistency of the defined values with
respect to the six-round Feistel construction, and operates as follows
on input a positive $3$-chain $(x, y, z)$ (a negative chain is
processed analogously): It extends it to a $6$-tuple $(x_1, \ldots,
x_6)$ with $x_{i} = x$, $x_{i+1} = y$ and $x_{i+2} = z$ and the additional
property that $\URF_i(x_{i}) = x_{i-1} \oplus x_{i+1}$ for all $i \in \{2,3,4,5\}$ and also that $\URP(x_2 \oplus \URF_{1}(x_1)\|x_1) =
x_6\|\URF_6(x_6) \oplus x_5$. This is achieved by first computing
$x_j$ for some $j \in \{i-2, i+2\}$ and setting the value
$\URF_j(x_j)$ uniformly at random (if undefined), and then computing
the two remaining values $x_{\ell}$ and $x_{\ell + 1}$, and {\em
  adapting} the respective output values $\URF_{\ell}(x_{\ell})$ and
$\URF_{\ell+1}(x_{\ell+1})$ to satisfy the constraint imposed by the
permutation $\URP$. Note that it may be that setting these values is
not possible (since some $x_{\ell}' \in \URF_{\ell}$ or some $x'_{\ell
  + 1} \in \URF_{\ell + 1}$), in which case we say the
simulator {\em aborts}. We point out that in this situation, the simulator
is unable to define the remaining chain values consistently with $\URP$.
 The precise choice of how $j$
is chosen depends on the index $i$, and we describe it in detail in
Appendix~\ref{appendixSimulatorDefinition}.

As the completion of these $3$-chains defines new entries for the
function tables, new $3$-chains may appear.  In this case, \ChainQuery\ is
recursively called on input $(x', i')$ for each value $\URF_{i'}(x')$
defined within one of the \CompleteChain\ calls invoked by \ChainQuery.

In the above description, we omitted one more complication:
at the beginning of each invocation of \ChainQuery, some
special special procedures (called
$\XorQuery_1$, $\XorQuery_2$, and $\XorQuery_3$) are invoked.
Their purpose is to
avoid some distinguishing attacks which are possible against the
simulator as detailed above, but in our distinguishing
attack these procedures are not helpful.
 We refer the reader to the Appendix for a
complete treatment.

\subsection{A problem in the Proof of \cite{CPS08v2}}\label{sec:ProblemsInExistingProof}

The core of the proof of \cite{CPS08v2} considers a distinguisher
which interacts with $\URP$ and $\sitwo$.
The goal is to show that the 
probability that $\sitwo$ aborts is negligible.

The natural approach taken in \cite{CPS08v2} is, for any execution
of \ChainQuery\ and under the condition that no abort has occurred so
far, to upper bound the probability that $\sitwo$ aborts in one of
the recursive calls to $\ChainQuery$. To achieve this, the following
condition is introduced:
\begin{quote} 
  The distribution of $\URF_{i}(x)$ when $\ChainQuery(x,i)$ is
  called is uniform in $\{0,1\}^n$ given the histories of all round
  functions, but where the value $\URF_i(x)$ is removed from the
  history, and in the case where the \ChainQuery\ invocation results
  from the completion of a chain, also all other values defined in
  that completion are removed from the history.
\end{quote}
The proof idea is as follows: Suppose that the condition holds for 
a \ChainQuery\ invocation, and no abort has occurred so far. Then,
except with negligible probability, 
no \CompleteChain\ execution within that \ChainQuery\ leads to an abort, 
and the condition holds for every recursive call to \ChainQuery.\footnote{We 
note that in some cases it turns out that the
condition may not hold, but then a separate proof is given that no
abort occurs for such a \ChainQuery\ execution, and that the condition
holds for all subsequent calls to \ChainQuery, i.e., such bad
invocations do not propagate.} 

The proof of this statement is very subtle: Between the point in time 
where the value $\URF_i(x)$ is set and the
invocation of $\ChainQuery(x, i)$, potentially several other 
$\ChainQuery$ are executed and extend the history. These
additional values in the history could depend on $\URF_i(x)$ and could 
fully determine $\URF_i(x)$. 

In fact, for specific cases where \cite{CPS08v2} claims that the condition
described above can be established (such as Lemma 10), it is possible to 
show that the value $\URF_{i}(x)$ is not distributed uniformly, but is in fact 
fully determined by the given history values. Thus, the condition does clearly not hold. For details, we refer to~\cite{Kuenzler10}.
We do not see how the proof can be extended to fix this problem. 
These issues led to a concrete attack, which we describe below. Since this
is a much stronger statement, we dispense with a more detailed description
of the problems in \cite{CPS08v2}.

\paragraph{A Problem in the Proof of~\cite{Seurin09}.} An alternative,
and more elegant approach for the $10$-round Feistel construction is
given by Seurin~\cite{Seurin09} in his PhD thesis. The core of the
proof also consists in proving an upper bound on the abort probability
of the simulator.

However, the proof suffers from similar problems as in
\cite{CPS08v2}. As an example, the proof of Lemma 2.11 in
\cite{Seurin09} claims that the simulator aborts only with negligible
probability in $\CompleteChain_2(W,R,S,D)$ when adapting $\URF_3(X)$,
because $X = R \oplus \URF_2(W)$, and $\URF_2(W)$ is distributed
uniformly in $\{0,1\}^n$. Yet, the statement about $\URF_2(W)$ being
uniform is questionable, since, similarly to the above case, there are
function values defined {\em after} $\URF_2(W) \inu \{0,1\}^n$
occurred, but before $\URF_3(X)$ gets adapted. These values might well
depend on $\URF_2(W)$, and therefore it is not at all clear if
$\URF_2(W)$ is still distributed uniformly given the values in the
history, and how one could prove such a statement.

Orthogonally, conditioning on something different than the complete
history at the moment where $\URF_2(W)$ is used does also not appear a
viable option. This leads us to the conclusion that it is still open
if this simulator can be used to prove indifferentiability. However,
in contrast to the $6$-round case, we have been unable to find a
concrete distinguishing attack when this simulator is used.

\subsection{The Attack against the Simulator of \cite{CPS08v2}}
\label{sec:attack_cps08v2}

\newcommand{\set}[1]{\mathcal{#1}}
As formalized by the following theorem, we show that there exists a
strategy for $\di$ such that $\sitwo$ aborts with overwhelming
probability. This immediately implies that $\sitwo$ cannot be used to
prove indifferentiability, since using the given strategy, one can
distinguish the real setting from the ideal setting (where $\sitwo$
aborts).

\begin{theorem}
  There is a distinguisher $\di$ such that $\sitwo$ aborts with
  overwhelming probability when $\di$ interacts with $\URP$ and
  $\sitwo$.
\end{theorem}

\begin{figure}[!t]
  \begin{center}
    \includegraphics[scale=0.5]{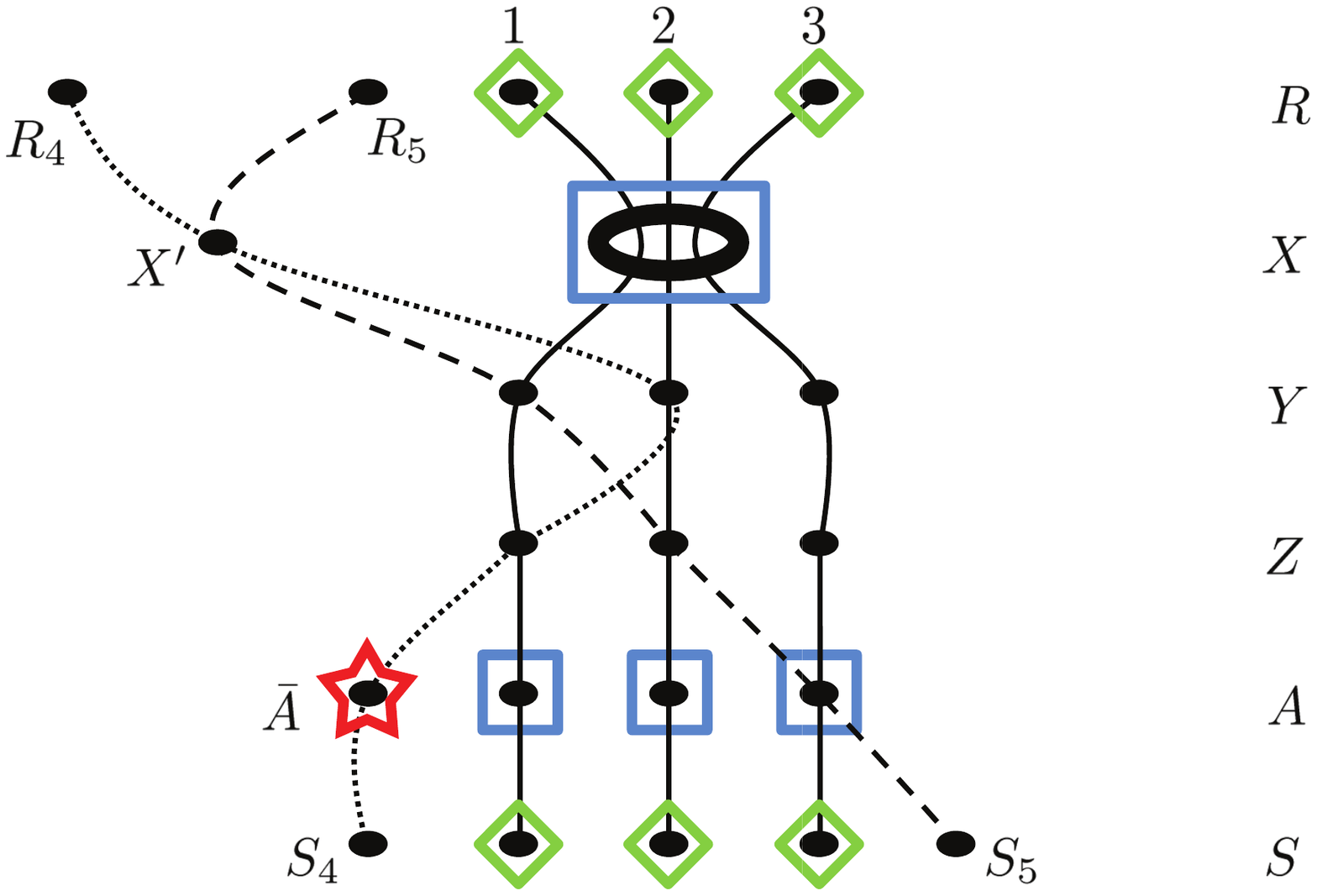}
  \end{center}
  \caption{Illustration of the attack provoking the simulator $\sitwo$ to abort.}
  \label{figAttackSitwo}
\end{figure}

The attack asks a very limited number of queries
(i.e., $7$ queries to the simulator, and three permutations queries).
When asking the last simulator query, the simulator is forced to
complete five different $3$-chains. The queries are chosen in a way that
after completing the first four chains, four values of the completion
of the remaining $3$-chain are defined {\em before} the associated
permutation query is issued by the simulator. At this point,
regardless of the strategy used, it is unlikely that the simulator can
set values so that this last chain is completed, and the simulator
aborts.  Figure~\ref{figAttackSitwo} illustrates the structure of the
(completed) $3$-chains. 

\paragraph{Outline of the attack and intuition.} 
The distinguisher $\di$ chooses $n$-bit values $X$,
$R_2, R_3$, and for $i \in \{2,3\}$, lets $L_i := \URF_1(R_i) \oplus
X$, $S_i\|T_i := \URP(L_i\|R_i)$, and $A_i := \URF_6(S_i) \oplus
T_i$. Then, it defines $R_1 := R_2 \oplus A_2 \oplus A_3$, $L_1 :=
\URF_2(X) \oplus R_1$, $S_1\|T_1 := \URP(L_1\|R_1)$, and $A_1 :=
\URF_6(S_1) \oplus T_1$. It is not hard to verify that $(S_i, R_i, X)$
are all $3$-chains for $i = 1, 2, 3$.  When completed to full
chains $(R_i, X, Y_i, Z_i, A_i, S_i)$, the values $(Y_1, Z_2, A_3)$ 
also constitute a $3$-chain, since
\begin{displaymath}
\URF_4(Z_2) = A_2 \oplus Y_2 = R_1 \oplus R_2 \oplus A_3 \oplus Y_2 = R_1 \oplus \URF_2(X) \oplus A_3 = Y_1 \oplus A_3
\end{displaymath}
under the assumption that the first three chains have been completed
correctly.  Finally, the distinguisher queries $\bar{A} := A_1 \oplus
R_1 \oplus R_2$ to $\URF_5$. Note that $(Y_2, Z_1, \bar{A})$ also
constitutes a $3$-chain under the assumptions that the first three
chains are completed correctly, since
\begin{displaymath}
    \URF_4(Z_1) = Y_1 \oplus A_1 = Y_1 \oplus \bar{A} \oplus R_1 \oplus R_2 = \URF_2(X) \oplus \bar{A} \oplus R_2 = \bar{A} \oplus Y_2.
\end{displaymath}
Finally, note that
\begin{displaymath}
  Z_1 \oplus \URF_3(Y_2) = Z_1 \oplus Z_2 \oplus X = Z_2 \oplus \URF_3(Y_1).
\end{displaymath}
This means in particular that when completed, the $3$-chains $(Y_1,
Z_2, A_3)$ and $(Y_2, Z_1, \bar{A})$ have a common second value, which
we denote as $X'$.

The core idea of the attack is the following: The simulator $\si$
first completes the $3$-chains $(S_i, R_i, X)$ for $i = 1, 2, 3$, and
only subsequently turns to completing the two remaining
$3$-chains. Say it completes the $3$-chain $(Y_2, Z_1, \bar{A})$ first:
Then, as this chain has a common value with the completion of the
$3$-chain $(Y_1, Z_2, A_3)$, in the end the simulator has only two
possible values (namely, those at both ends of the completed chain)
which are still free to be set to complete the $3$-chain $(Y_1, Z_2,
A_3)$, and this leads to an abort. (In fact, we prove the slightly
stronger statement that the simulator fails even if it adopts any
other strategy to complete this last chain, once the second-last chain
is completed.)

The main difficulty of the full analysis, given in
Appendix~\ref{app:analysisattack}, is that the actual simulator makes
calls to procedures (called $\XorQuery_1$, $\XorQuery_2$, and
$\XorQuery_3$) which are intended to prevent (other) attacks.  To 
show that no such call affects the intuition behind our attack is a
rather cumbersome task.

Note that our implementation of the simulator and our attack in Python
indeed shows that the simulator aborts (the code can be found
as ancilliary file in \cite{HKT10full}).

\paragraph{A stronger attack.} 
It is actually possible to come up with a simulator that defines all
function values consistently with $\URP$ under the attack present in
the previous section, and thus the attack falls short of proving that
the six-round Feistel construction cannot be indifferentiable from a
random permutation. Appendix~\ref{app:strongerattack} presents a
stronger distinguishing attack, for which, in fact, we were not able
to come up with a simulator which withstands it. We {\em conjecture}
that no simulator within a very large class of simulators is able to
withstand this distinguishing attack, but giving such a proof seems to
be quite difficult and to require a deeper and more general
understanding of the possible dependencies between
chains. Nonetheless, we consider this distinguisher to be a useful
testbed for any attempt to fix the indifferentiability proof for six
rounds.\footnote{ We have also implemented this more general attack in
  Python, and, not surprisingly, its execution also leads to an abort
  of the simulator $\sitwo$.}

\section{Indifferentiability of the Feistel Construction from a Random Permutation}\label{sec:proof}

\newcommand{\Exp}{\textsf{Exp}}
\newcommand{\Dist}{\mathbf{D}} 
\newcommand{\Syst}{\mathsf{S}}
\newcommand{\SystE}{\mathsf{E}}
\newcommand{\SIM}{\mathbf{S}} 
\newcommand{\SIMM}{\mathbf{T}}
\newcommand{\bin}[1]{\{0,1\}^{#1}} \newcommand{\totPqueries}{q_{\URP}}
\newcommand{\DistAdv}[3]{\Delta^{#1}\left(#2, #3\right)}
 \newcommand{\BadOC}{\textsf{BadOC}}
\newcommand{\Abort}{{\normalfont \textsf{Abort}}}
\newcommand{\evMismatch}{{\normalfont \textsf{Mismatch}}}
\newcommand{\evBadFill}{{\normalfont \textsf{BadFill}}}
\newcommand{\Hit}{{\normalfont \textsf{Hit}}}
\newcommand{\HitP}{{\normalfont \textsf{HitP}}}

We prove that the $14$-round Feistel construction is indifferentiable from 
a random permutation. 

\begin{theorem}\label{thm:mainTheorem}
  The $14$-round Feistel construction using $14$
  independent random functions is indifferentiable from a random
  permutation. 
\end{theorem}

The remainder of this section is devoted to the proof of
Theorem~\ref{thm:mainTheorem}. Our task is to provide a
simulator $\SIM$ with access to a random permutation $\URP$ such that $(\SIM^{\URP}, \URP)$ is indistinguishable from $(\URF,\Feistel^{\URF})$, where
$\URF$ denotes the random functions used in the Feistel construction.

We first define the simulator $\SIM$ in
Section~\ref{sec:defineSimulator}.  Then we transform $(\SIM^{\URP},
\URP)$ stepwise to $(\URF, \Feistel^{\URF})$. 
The random functions we consider in this section are always from $n$ bits to $n$ bits, and the random permutation $\URP$ is over $2n$ bits.

\subsection{Simulator Definition}\label{sec:defineSimulator}

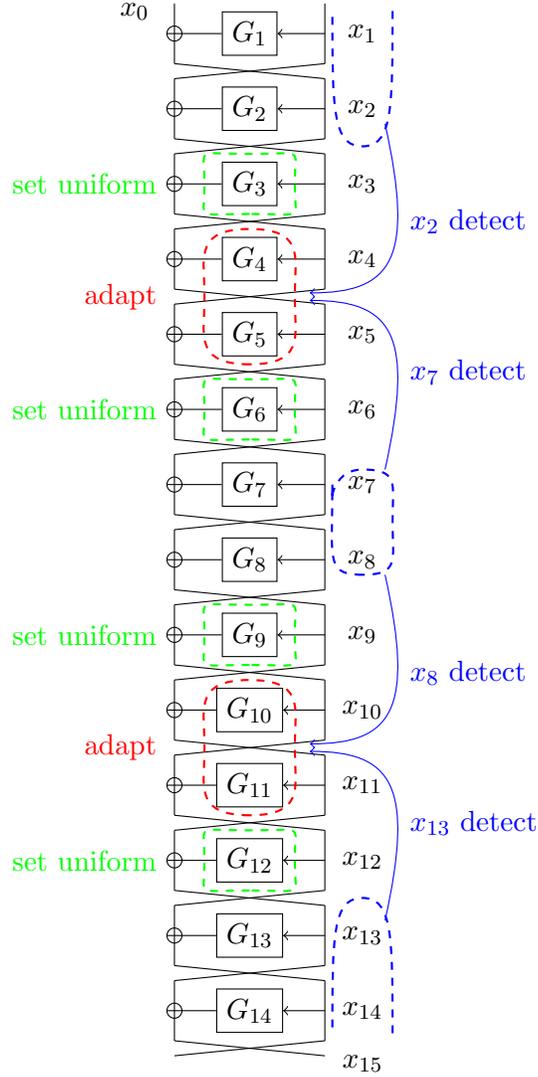
\begin{figure}
  \centering
  \tikzstyle{urf}=[rectangle, draw=black] 
  \pgfmathsetmacro{\xorsize}{0.1}

  \begin{tikzpicture}[xscale=1]
    \foreach \var in {1,2,...,14} {
      \draw (1.5,-\var) node {$x_{\var}$};
      \draw (0, -\var) node[urf](f\var) {$\VarG_{\var}$};
      \draw (1,-0.4-\var) -- (1,0.4-\var);
      \draw (-1,-0.4-\var) -- (-1,0.4-\var);
      \draw (-1,-0.4-\var) -- (1,-0.6-\var);
      \draw (1,-0.4-\var) -- (-1,-0.6-\var);
      \path[->] (1,-\var) edge (f\var);
      \draw (f\var) -- (-1 - \xorsize,-\var);
      \draw (-1,-\var) circle(\xorsize);
    }
    \draw (-1.2,-.7) node [left] {$x_0$};
    \draw (1.1,-14.7) node [right] {$x_{15}$};

    \draw[blue,dashed,thick] (1.9,-14.3) ..controls (1.9,-13) and (1.9,-12.5) .. (1.5,-12.5)
                  ..controls (1.1,-12.5) and (1.1,-13) .. (1.1, -14.3); 

    \draw[blue,dashed,thick] (1.9,-0.7) ..controls (1.9,-2) and (1.9,-2.5) .. (1.5,-2.5)
                  ..controls (1.1,-2.5) and (1.1,-2) .. (1.1, -0.7); 

    \draw[blue,dashed,thick] (1.5,-8.2) ..controls (1,-8.2) and (1.1,-8) .. (1.1,-7) 
                  ..controls (1.1,-07.5) and (1,-6.8) .. (1.5,-6.8) 
                  ..controls (2,-6.8) and (1.9,-7) .. (1.9,-7.5)
                  ..controls (1.9,-8) and (2,-8.2) .. (1.5,-8.2);


    \pgfmathsetmacro{\Y}{-5}
    \draw[dashed,red,thick] (0,\Y-.4) ..controls (-0.75,\Y-.4) and (-.6,\Y).. (-.6,\Y+0.5) 
                  ..controls (-.6,\Y+1) and (-0.75,\Y+1.4) .. (0,\Y+1.4) 
                  ..controls (0.75,\Y+1.4) and (.6,\Y+1) .. (.6,\Y+0.5)
                  ..controls (.6,\Y) and (0.75,\Y-.4) .. (0,\Y-.4);
    \draw[red] (-1.1,\Y+.5) node [left] {adapt};

    \pgfmathsetmacro{\Y}{-11}
    \draw[dashed,red,thick] (0,\Y-.4) ..controls (-0.75,\Y-.4) and (-.6,\Y).. (-.6,\Y+0.5) 
                  ..controls (-.6,\Y+1) and (-0.75,\Y+1.4) .. (0,\Y+1.4) 
                  ..controls (0.75,\Y+1.4) and (.6,\Y+1) .. (.6,\Y+0.5)
                  ..controls (.6,\Y) and (0.75,\Y-.4) .. (0,\Y-.4);
    \draw[red] (-1.1,\Y+.5) node [left] {adapt};

    \pgfmathsetmacro{\greenbox}{-3}
    \draw[dashed,green,thick] (0,\greenbox-0.4) 
    ..controls (-0.75,\greenbox-.4) and (-.6,\greenbox-.5).. (-.6,\greenbox) 
    ..controls (-.6,\greenbox+.5) and (-0.75,\greenbox+.4) .. (0,\greenbox+.4) 
    ..controls (0.75,\greenbox+.4) and (.6,\greenbox+.5) .. (.6,\greenbox)
    ..controls (.6,\greenbox-.5) and (0.75,\greenbox-.4) .. (0,\greenbox-.4);
    \draw[green] (-1.1,\greenbox) node [left] {set uniform};
    \pgfmathsetmacro{\greenbox}{-6}
    \draw[dashed,green,thick] (0,\greenbox-0.4) 
    ..controls (-0.75,\greenbox-.4) and (-.6,\greenbox-.5).. (-.6,\greenbox) 
    ..controls (-.6,\greenbox+.5) and (-0.75,\greenbox+.4) .. (0,\greenbox+.4) 
    ..controls (0.75,\greenbox+.4) and (.6,\greenbox+.5) .. (.6,\greenbox)
    ..controls (.6,\greenbox-.5) and (0.75,\greenbox-.4) .. (0,\greenbox-.4);
    \draw[green] (-1.1,\greenbox) node [left] {set uniform};
    \pgfmathsetmacro{\greenbox}{-9}
    \draw[dashed,green,thick] (0,\greenbox-0.4) 
    ..controls (-0.75,\greenbox-.4) and (-.6,\greenbox-.5).. (-.6,\greenbox) 
    ..controls (-.6,\greenbox+.5) and (-0.75,\greenbox+.4) .. (0,\greenbox+.4) 
    ..controls (0.75,\greenbox+.4) and (.6,\greenbox+.5) .. (.6,\greenbox)
    ..controls (.6,\greenbox-.5) and (0.75,\greenbox-.4) .. (0,\greenbox-.4);
    \draw[green] (-1.1,\greenbox) node [left] {set uniform};
    \pgfmathsetmacro{\greenbox}{-12}
    \draw[dashed,green,thick] (0,\greenbox-0.4) 
    ..controls (-0.75,\greenbox-.4) and (-.6,\greenbox-.5).. (-.6,\greenbox) 
    ..controls (-.6,\greenbox+.5) and (-0.75,\greenbox+.4) .. (0,\greenbox+.4) 
    ..controls (0.75,\greenbox+.4) and (.6,\greenbox+.5) .. (.6,\greenbox)
    ..controls (.6,\greenbox-.5) and (0.75,\greenbox-.4) .. (0,\greenbox-.4);
    \draw[green] (-1.1,\greenbox) node [left] {set uniform};
    
    \pgfmathsetmacro{\adaptpos}{-8}
    \draw [->,blue] (1.8,\adaptpos-.2) .. 
           controls (2.2,\adaptpos-2) and (1.9,\adaptpos-2.45) .. (0.8, \adaptpos-2.45);
    \draw [blue](2,\adaptpos-1.5) node [right] {$x_{8}$ detect};

    \pgfmathsetmacro{\adaptpos}{-2}
    \draw [->,blue] (1.8,\adaptpos-.2) .. 
           controls (2.2,\adaptpos-2) and (1.9,\adaptpos-2.45) .. (0.8, \adaptpos-2.45);
    \draw [blue](2,\adaptpos-1.5) node [right] {$x_2$ detect};

\begin{scope}[yscale=-1]
    \pgfmathsetmacro{\adaptpos}{7}
    \draw [->,blue] (1.8,\adaptpos-.2) .. 
           controls (2.2,\adaptpos-2) and (1.9,\adaptpos-2.45) .. (0.8, \adaptpos-2.45);
    \draw [blue](2,\adaptpos-1.5) node [right] {$x_7$ detect};

    \pgfmathsetmacro{\adaptpos}{13}
    \draw [->,blue] (1.8,\adaptpos-.2) .. 
           controls (2.2,\adaptpos-2) and (1.9,\adaptpos-2.45) .. (0.8, \adaptpos-2.45);
    \draw [blue](2,\adaptpos-1.5) node [right] {$x_{13}$ detect};
  \end{scope}

  \end{tikzpicture}
  \caption{The 14-round Feistel with the zones where our simulator
    detects chains and adapts them.  Whenever a function value
    $\VarG_2(x_2), \VarG_{7}(x_7), \VarG_{8}(x_{8})$, or $\VarG_{13}(x_{13})$
    is defined, the simulator checks whether the values in the blue
    dashed zones $x_7,x_{8}$ and $x_{1}, x_{2},x_{13}, x_{14}$ form a
    partial chain. In case a chain is detected, it is completed; the
    function values in the red dashed zones are adapted in order to
    ensure consistency of the chain.  
    }
  \label{fig:Feistel}
\end{figure}
We first give a somewhat informal, but detailed description of the
simulator.  We then use pseudocode to specify the simulator in a more
formal manner.

\subsubsection{Informal description}
The simulator provides an interface $\SIM.\F(k,x)$ to query the
simulated random function~$\URF_k$ on input~$x$.  For each~$k$, the
simulator internally maintains a table that has entries which are
pairs $(x,y)$.  They denote pairs of inputs and outputs of
$\SIM.\F(k,x)$.  We denote these tables by $\SIM.\VarG_k$ or just
$\VarG_k$ when the context is clear.  We write $x \in \VarG_k$ to
denote that~$x$ is a preimage in this table, often identifying
$\VarG_k$ with the set of \emph{preimages} stored.  When $x \in
\VarG_k$, $\VarG_k(x)$ denotes the corresponding image.

On a query $\SIM.\F(k,x)$, the simulator first checks whether $x \in
\VarG_k$.  If so, it answers with $\VarG_k(x)$.  Otherwise the
simulator picks a random value $y$ and inserts $(x, y)$ into
$\VarG_k(x)$.  After this, the simulator takes steps to ensure that in
the future it answers consistently with the random permutation $\URP$.

There are two cases in which the simulator performs a specific action
for this.  First, if $k \in \{2,13\}$, the simulator considers all
newly generated tuples $(x_{1},x_{2}, x_{13},x_{14}) \in
\VarG_{1}\times\VarG_{2}\times\VarG_{13}\times \VarG_{14}$, and
computes $x_0 := x_{2} \oplus \VarG_1(x_1)$ and $x_{15} := x_{13}
\oplus \VarG_{14}(x_{14})$.  It then checks whether $\URP(x_0,x_1) =
(x_{14},x_{15})$.  Whenever the answer to such a check is positive,
the simulator enqueues the detected values in a queue.  More
precisely, it enqueues a four-tuple $(x_1, x_2,1,\ell)$.  The value
$1$ ensures that later the simulator knows that the first value $x_1$
corresponds to $\VarG_1$. The value $\ell$ describes where to adapt
values of $\VarG_\ell$ to ensure consistency with the given
permutation.  If $k = 2$, then $\ell = 4$ and if $k = 13$ then $\ell =
10$.  

The second case is when $k \in \{7,8\}$.  Then, the simulator
enqueues all newly generated pairs $(x_7,x_8) \in
\VarG_{7}\times\VarG_{8}$.  It enqueues all these pairs into the
queue as $(x_7,x_{8},7,\ell)$, where $\ell = 4$ if $k=7$ and $\ell =
10$ if $k = 8$ (this is illustrated in Figure~\ref{fig:Feistel}).

After enqueuing this information, the simulator immediately takes the
partial chain out of the queue again, and starts completing it.  For
this, it evaluates the Feistel chain forward and backward (invoking
$\URP$ or $\URP^{-1}$ at one point in order to wrap around), until
$x_\ell$ and $x_{\ell+1}$ are computed, and only the two values
$\VarG_{\ell}(x_\ell)$ and $\VarG_{\ell+1}(x_{\ell+1})$ are possibly undefined. 
The simulator defines the remaining two values in such a way that consistency with $\URP$ is ensured, i.e.,
$\VarG_{\ell}(x_\ell) := x_{\ell-1} \oplus x_{\ell+1}$ and $\VarG_{\ell+1}(x_{\ell+1}) := x_{\ell} \oplus x_{\ell+2}$.
If a value for either of these is defined from a previous action
of the simulator, the simulator overwrites the value (possibly making 
earlier chains inconsistent). 

During the evaluation of the Feistel chain, the simulator usually defines new values
for the tables $\VarG$.  Whenever a value $\VarG_{k}(x_k)$ for $k \in
\{2,13\}$ is defined, the exact same checks as above are performed on the newly
generated tuples $(x_1,x_2,x_{13},x_{14})$.  Whenever a value
$\VarG_k(x_k)$ for $k \in \{7,8\}$ is defined, the simulator similarly enqueues all new
pairs $(x_7,x_{8})$.

When the simulator has finished completing a chain, it checks whether
the queue is now empty.  While it is not empty, it keeps dequeuing
entries and completing chains, otherwise, it returns the answer to the
initial query to the caller.

In order to make sure the simulator does not complete the same chains
twice, the simulator additionally keeps a set
$\CompletedChains$ that contains all triples $(x_k,x_{k+1},k)$ which
have been completed previously.  Whenever the simulator dequeues a chain,
it only completes the chain if it is not in the set $\CompletedChains$.

\subsubsection{The simulator in pseudocode}
We now provide pseudocode to describe the simulator as explained
above in full detail.  Later, during the analysis, we will consider a slightly
different simulator $\SIMM$.  For this, we replace whole lines; the
replacements are put into boxes next to these lines.  The reader can
ignore these replacements at the moment.

First, the simulator internally uses a queue and some
hashtables to store the function values, and a set $\CompletedChains$ 
to remember the chains that have been completed already.

\lstdefinelanguage{pseudocode}{numbers=left,numberstyle=\tiny,mathescape,escapechar=*,morekeywords={System,abort,procedure,if,else,do,done,for,to,step,endif,forall,new,assert,end,return,while,then,private,public, Simulator, Variables,true,false},flexiblecolumns}
\begin{lstlisting}[language=pseudocode,name=simulator]
System $\SIM$:*\hspace{5cm}\framebox{\textbf{System} $\SIMM(f)$:}*
Variables:
    Queue Q
    Hashtable $\VarG_1,\ldots,\VarG_{14}$
    Set $\CompletedChains := \emptyset$
\end{lstlisting}

The procedure $\F(i,x)$ provides the interface to a distinguisher.  It
first calls the corresponding internal procedure $\Finner$, which defines the
value and fills the queue if necessary. Then, the procedure $\F(i,x)$ completes the chains 
in the queue that were not completed previously, until the queue is empty.

\begin{lstlisting}[language=pseudocode,name=simulator]  
public procedure $\F(i,x)$
     $\Finner(i,x)$
     while $\lnot Q.\Qempty()$ do
          $(x_k,x_{k+1},k,\ell) := Q.\dequeue()$
          if $(x_k,x_{k+1},k) \notin \CompletedChains$ then*\quad/\!\!/ ignore previously completed chains*
               */\!\!/ complete the chain*
               $(x_{\ell-2},x_{\ell-1}) := \evalFwd(x_k,x_{k+1},k,\ell-2)$*\label{line:evalFwdChangeParam}*
               $(x_{\ell+2},x_{\ell+3}) := \evalBwd(x_k,x_{k+1},k,\ell+2)$*\label{line:evalBwdChangeParam}*
               $\adapt(x_{\ell-2},x_{\ell-1},x_{\ell+2},x_{\ell+3},\ell)$
               $(x_{1},x_2) := \evalBwd(x_k,x_{k+1},k,1)$
               $(x_{7},x_{8}) := \evalFwd(x_1,x_{2},1,7)$
               $\CompletedChains := \CompletedChains \cup \{ (x_{1},x_{2},1), (x_{7},x_8,7) \}$*\label{line:addToCompletedChains}*
     return $\VarG_i(x)$
\end{lstlisting}

The procedure $\adapt$ adapts the values.  It first sets the values
marked green in Figure~\ref{fig:Feistel} uniformly at random, and also
the next ones.  It then adapts the values of $\VarG_{\ell}(x_{\ell})$
and $\VarG_{\ell+1}(x_{\ell+1})$ such that the chain matches the
permutation.

It would be possible to simplify the code by removing
lines~\ref{line:adaptRemovableStart} to \ref{line:adaptRemovableEnd}
below, and changing the parameters in
lines~\ref{line:evalFwdChangeParam} and \ref{line:evalBwdChangeParam}
above.  The current notation simplifies notation in the proof.

\begin{lstlisting}[language=pseudocode,name=simulator]
private procedure $\adapt(x_{\ell-2},x_{\ell-1},x_{\ell+2},x_{\ell+3},\ell)$
     if $x_{\ell-1} \notin \VarG_{\ell-1}$ then *\label{line:adaptRemovableStart}*
          $\VarG_{\ell-1}(x_{\ell-1}) \leftarrow_R \{0,1\}^n$*\hspace{1.4cm}\framebox{$\VarG_{\ell-1}(x_{\ell-1}) := f(\ell-1,x_{\ell-1})$}*
     $x_{\ell} := x_{\ell-2} \oplus \VarG_{\ell-1}(x_{\ell-1})$
     if $x_{\ell+2} \notin \VarG_{\ell+2}$ then 
          $\VarG_{\ell+2}(x_{\ell+2}) \leftarrow_R \{0,1\}^n$*\hspace{1.4cm}\framebox{$\VarG_{\ell+2}(x_{\ell+2}) := f(\ell+2,x_{\ell+2})$}*
     $x_{\ell+1} := x_{\ell+3} \oplus \VarG_{\ell+2}(x_{\ell+2})$*\label{line:adaptRemovableEnd}*
     $\forceVal(x_\ell, x_{\ell+1} \oplus x_{\ell-1}, \ell)$
     $\forceVal(x_{\ell+1},x_{\ell} \oplus x_{\ell+2}, \ell+1)$

private procedure $\forceVal(x, y, \ell)$
     $\VarG_{\ell}(x) := y$
\end{lstlisting}

The procedure $\Finner$ provides the internal interface for
evaluations of the simulated function.  It only fills the queue, but
does not empty it.
\begin{lstlisting}[language=pseudocode,name=simulator]
private procedure $\Finner(i,x)$:
     if $x \notin \VarG_i$ then
          $\VarG_i(x) \leftarrow_R \{0,1\}^n$*\hspace{2.3cm}\framebox{$\VarG_{i}(x) := f(i,x)$}*
          if $i \in \{2,7,8,13\}$ then
               $\enqNewChains(i,x)$
     return $\VarG_i(x)$
\end{lstlisting}

The procedure $\enqNewChains$ detects newly created chains and
enqueues them.  Sometimes, chains may be detected which have been
completed before, but they are ignored when they are dequeued.
\begin{lstlisting}[language=pseudocode,name=simulator]
private procedure $\enqNewChains(i,x)$:
     if $i=2$ then
          forall $(x_{1},x_{2},x_{13},x_{14}) \in \VarG_1 \times \{x\} \times \VarG_{13} \times \VarG_{14}$ do
               if $\Check(x_2 \oplus \VarG_1(x_1),x_1,x_{14}, x_{13} \oplus \VarG_{14}(x_{14}))$ then *\label{line:URPQuery1}*
                    $Q.\enqueue(x_1,x_2,1,4)$ 
     else if $i = 13$ then
          forall $(x_{1},x_{2},x_{13},x_{14}) \in \VarG_1 \times \VarG_2 \times \{x\} \times \VarG_{14}$ do *\label{line:URPQuery2}*
               if $\Check(x_2 \oplus \VarG_1(x_1),x_1, x_{14}, x_{13} \oplus \VarG_{14}(x_{14}))$ then 
                    $Q.\enqueue(x_{1},x_{2},1,10)$ 
     else if $i = 7$ then
          forall $(x_7, x_{8}) \in \{x\} \times \VarG_{8}$ do 
               $Q.\enqueue(x_7,x_{8},7,4)$ 
     else if $i = 8$ then
          forall $(x_7, x_{8}) \in \VarG_{7} \times \{x\}  $ do
               $Q.\enqueue(x_7,x_{8},7,10)$

private procedure $\Check(x_0,x_1,x_{14},x_{15})$ *\label{line:check}*
     return $\URP(x_0,x_1) = (x_{14},x_{15})$*\hspace{1cm}\framebox{\textbf{return} $\TSRF.\Check(x_0,x_1,x_{14},x_{15})$}*
\end{lstlisting}

The helper procedures $\evalFwd$ and $\evalBwd$ take indices $k$ and $\ell$ and a pair
$(x_k,x_{k+1})$ of input values for $\VarG_k$ and $\VarG_{k+1}$, and either evaluate forward or backward in the Feistel to
obtain the pair $(x_{\ell},x_{\ell+1})$ of input values for $\VarG_\ell$ and $\VarG_{\ell+1}$.

\begin{lstlisting}[language=pseudocode,name=simulator]
private procedure $\evalFwd(x_{k},x_{k+1},k,\ell)$:
     while $k \neq \ell$ do
         if $k = 14$ then
             $(x_0,x_1) := \URP^{-1}(x_{14},x_{15})$*\hspace{1cm}\framebox{$(x_0,x_1) := \TSRF.\P^{-1}(x_{14},x_{15})$}*
             $k := 0$
         else
             $x_{k+2} := x_{k} \oplus \Finner(k+1,x_{k+1})$
             $k := k + 1$
     return $(x_\ell,x_{\ell+1})$

private procedure $\evalBwd(x_k,x_{k+1},k,\ell)$:
     while $k \neq \ell$ do
         if $k = 0$ then
             $(x_{14},x_{15}) := \URP(x_0,x_1)$*\hspace{1cm}\framebox{$(x_{14},x_{15}) := \TSRF.\P(x_0,x_1)$}*
             $k := 14$
         else
             $x_{k-1} := x_{k+1} \oplus \Finner(k,x_{k})$
             $k := k - 1$
     return $(x_{\ell},x_{\ell+1})$
\end{lstlisting}

\subsection{Proof of Indifferentiability}

In this section, we provide the indifferentiability analysis.  The
overall plan is that we first fix a deterministic distinguisher
$\Dist$, and suppose that it makes at most $q$ queries.\footnote{We may
assume that $\Dist$ is deterministic, since we are only interested in the
advantage of the optimal distinguisher, and for any probabilitstic distinguisher,
the advantage can be at most the advantage of the optimal deterministic 
distinguisher.}
We then show
that the probability that $\Dist$ outputs $1$ when interacting with
$(\URP, \SIM^{\URP})$ differs by at most $\frac{\poly(q)}{2^n}$ from
the probability it outputs $1$ when interacting with
$(\Feistel^{\URF}, \URF)$, where $\Feistel$ is a $14$-round Feistel
construction, and $\URF$ is a collection of $14$ uniform random
functions.

We denote the scenario where $\Dist$ interacts with
$(\URP,\SIM^{\URP})$ by $\Syst_1$, and the scenario where $\Dist$
interacts with $(\Feistel^{\URF},\URF)$ by $\Syst_4$.  The scenarios
$\Syst_2$ and $\Syst_3$ will be intermediate scenarios.

\subsubsection{Replacing the permutation with a random function}
Scenario $\Syst_2(f,p)$ is similar to $\Syst_1$.  However, instead of
the simulator $\SIM$ we use the simulator $\SIMM(f)$, and instead of a
random permutation $\URP$ we use a two-sided random function
$\TSRF(p)$. 
 The differences between these systems are as follows:
\begin{description}
\item[Explicit randomness:] We make the randomness used by the the
  simulator explicit.  Whenever $\SIM$ sets $\VarG_{i}(x_i)$ to a
  random value, $\SIMM(f)$ takes it from $f(i,x_i)$ instead, where $f$
  is a table which contains an independent uniform random bitstring of
  length $n$ for each $i\in \{1,2, \ldots, 14\}$ and $x_i \in
  \{0,1\}^n$.  This modification does not change the distribution of
  the simulation, because it is clear that the simulator considers
  each entry of $f$ at most once.

  As depicted in the pseudocode below, the randomness of the two-sided
  random function $\TSRF(p)$ is also explicit: it is taken from
  $p(\downarrow,x_0,x_1)$ or $p(\uparrow,x_{14},x_{15})$, a table in
  which each entry is an independent uniform random bitstring of
  length $2n$.
\item[Two-sided random function:] We replace the random permutation
  $\URP$ by a two-sided random function $\TSRF(p)$ (see below for
  pseudocode).  This function keeps a hashtable~$\VarP$ which contains
  elements $(\downarrow,x_0,x_1)$ and $(\uparrow,x_{14},x_{15})$.
  Whenever the procedure $\TSRF.\P(x_0,x_1)$ is queried, $\TSRF$
  checks whether $(\downarrow,x_0,x_1) \in \VarP$, and if so, answers
  accordingly.  Otherwise, an independent uniform random output
  $(x_{14},x_{15})$ is picked (by considering $p$), and
  $(\downarrow,x_{0},x_{1})$ as well as $(\uparrow,x_{14},x_{15})$ are
  added to $\VarP$, mapping to each other.
\item[$\Check$ procedure:] The two-sided random function
  $\TSRF$ has a procedure $\Check(x_0,x_1,x_{14},x_{15})$.  It checks
  whether $\VarP$ maps $(\downarrow,x_0,x_1)$ to $(x_{14},x_{15})$,
  and if so, returns true.  If not, it checks whether $\VarP$ maps
  $(\uparrow,x_{14},x_{15})$ to $(x_{0},x_{1})$, and if so, returns
  true.  Otherwise, it returns false. The simulator $\SIMM(f)$ also
  differs from~$\SIM$ in that $\SIMM(f).\Check$ simply calls
  $\TSRF.\Check$.
\end{description}
Pseudocode for $\SIMM(f)$ can be obtained by using the boxed contents
on the right hand side in the pseudocode of $\SIM$ instead of the
corresponding line.  For the two-sided random function $\TSRF$,
the pseudocode looks as follows:
\begin{lstlisting}[language=pseudocode]
System Two-sided random function $\TSRF(p)$:
Variables: 
     Hashtable $\VarP$

public procedure $\P(x_0,x_1)$
     if $(\downarrow,x_0, x_1) \notin \VarP$ then
          $(x_{14},x_{15}) := p(\downarrow,x_{0},x_{1})$*\label{line:tsrf_queryP1}*
          $\VarP(\downarrow,x_0,x_1) := (x_{14},x_{15})$
          $\VarP(\uparrow,x_{14},x_{15}) := (x_{0},x_{1})$ *\qquad/\!\!/ (May overwrite an entry)*
     return $\VarP(\downarrow,x_0,x_1)$
     
public procedure $\P^{-1}(x_{14}, x_{15})$
     if $(\uparrow,x_{14}, x_{15}) \notin \VarP$ then
          $(x_{0},x_{1}) := p(\uparrow,x_{14},x_{15})$*\label{line:tsrf_queryP2}*
          $\VarP(\downarrow,x_0,x_1) := (x_{14},x_{15})$ *\qquad/\!\!/ (May overwrite an entry)*
          $\VarP(\uparrow,x_{14},x_{15}) := (x_{0},x_{1})$ 
     return $\VarP(\uparrow,x_{14},x_{15})$

public procedure $\Check(x_0,x_1,x_{14},x_{15})$
     if $(\downarrow,x_0,x_1) \in \VarP$ then return $\VarP(\downarrow,x_0,x_1) = (x_{14},x_{15})$*\label{line:tsrf_check1}*
     if $(\uparrow,x_{14},x_{15}) \in \VarP$ then return $\VarP(\uparrow,x_{14},x_{15}) = (x_{0},x_{1})$*\label{line:tsrf_check2}*
     return false
\end{lstlisting}

We claim that for uniformly chosen $(f,p)$, the probability that
$\Dist$ outputs $1$ in scenario $\Syst_2(f,p)$ differs only by
$\frac{\poly(q)}{2^n}$ from the probability it outputs $1$ in
$\Syst_1$.  For this, we first note that clearly the simulator can
take the randomness from $f$ without any change.  Secondly, instead of
the procedure $\Check$ in the simulator $\SIM$, we can imagine that the
random permutation $\URP$ has a procedure $\URP.\Check$ which is
implemented exactly as in line \ref{line:check} of $\SIM$, and
$\SIM.\Check$ simply calls $\URP.\Check$.  The following lemma states
that such a system $\URP$ is indistinguishable from $\TSRF$ as above,
which then implies the claim.  We prove the lemma in
Section~\ref{sec:eqS1andS2}; the proof is neither surprising nor
particularly difficult.

\begin{lemma}\label{lem:PermAndTSRF}
  Consider a random permutation over $2n$ bits, to which we add the
  procedure $\Check$ as in line~\ref{line:check} of the simulator $\SIM$.
  Then, a distinguisher which issues at most $q'$ queries to either the
  random permutation or to the two-sided random function $\TSRF$ has
  advantage at most $\frac{6(q')^2}{2^{2n}}$ in distinguishing the two
  systems.
\end{lemma}

Additionally, we will need that in $\Syst_2$ the number of queries
made by the simulator is $\poly(q)$.  This is given in the following
lemma, which we prove in Section~\ref{sec:SimulatorIsEfficient}.
\begin{lemma}\label{lem:SimulatorIsEfficient}
  In $\Syst_2$, at any point in the execution we have $|\VarG_i| \leq
  6q^2$ for all $i$.  Furthermore, there are at most $6q^2$ queries to
  both $\TSRF.\P$, and $\TSRF.\P^{-1}$, and at most $1296q^8$ queries
  to $\TSRF.\Check$.
\end{lemma}
In order to prove Lemma~\ref{lem:SimulatorIsEfficient}, we can follow
the elegant idea from \cite{CPS08v2}.  Thus, while we give the proof
for completeness, it should not be considered a contribution of this
paper.  It is in fact very similar to the corresponding proof in
\cite{Seurin09}.

\subsubsection{Introducing a Feistel-construction}

In $\Syst_3(h)$, we replace the above two-sided random function
$\TSRF(p)$ by a Feistel construction $\Psi(h)$.  For this, we use the
following system:

\begin{lstlisting}[language=pseudocode]
System $\Psi(h)$:

Variables: 
     Hashtable $\VarH_1,\ldots,\VarH_{14}$
     Hashtable $\VarP$

private procedure $\F(i,x_i)$:
     if $x_i \notin \VarH_i$ then
          $\VarH_i(x_i) := h(i,x_i)$
     return $\VarH_i(x_i)$

public procedure $\P(x_0,x_1)$
     for $i := 2$ to $15$ do
          $x_{i} := x_{i-2} \oplus \F(i-1,x_{i-1})$
     $\VarP(\downarrow,x_{0},x_{1}) := (x_{14},x_{15})$
     $\VarP(\uparrow,x_{14},x_{15}) := (x_0,x_1)$
     return $(x_{14},x_{15})$
     
public procedure $\P^{-1}(x_{14}, x_{15})$
     for $i := 13$ to $0$ step $-1$ do
          $x_{i} := x_{i+2} \oplus \F(i+1,x_{i+1})$
     $\VarP(\downarrow,x_{0},x_{1}) := (x_{14},x_{15})$
     $\VarP(\uparrow,x_{14},x_{15}) := (x_0,x_1)$
     return $(x_{0},x_{1})$

public procedure $\Check(x_0,x_1,x_{14},x_{15})$
     if $(\downarrow,x_0,x_1) \in \VarP$ then return $\VarP(\downarrow,x_0,x_1) = (x_{14},x_{15})$
     if $(\uparrow,x_{14},x_{15}) \in \VarP$ then return $\VarP(\uparrow,x_{14},x_{15}) = (x_{0},x_{1})$
     return false
\end{lstlisting}

We define $\Syst_3(h)$ to be the system where the distinguisher
interacts with $(\Psi(h), \SIMM(h)^{\Psi(h)})$.  Note that the
randomness used by $\Psi$ and $\SIMM$ is the same, and we call it $h$.
We show the following lemma in Section~\ref{sec:eqS2andS3}.  

\begin{lemma}\label{lem:eqS2andS3}
  The probability that a fixed distinguisher answers $1$ in
  $\Syst_2(f,p)$ for uniform random $(f,p)$ differs at most by
  $\frac{8\cdot 10^{19} \cdot q^{10}}{2^n}$ from the probability that it answers $1$ in
  $\Syst_3(h)$ for uniform random~$h$.
\end{lemma}

The proof of this lemma is the main contribution of this paper.
A large part of the proof consists in showing that the simulator
does not overwrite a value in calls to $\forceVal$; this part
is very different than the corresponding step in \cite{CPS08v2}.
An interesting feature of the proof is that in a second part
it directly maps pairs
$(f,p)$ to elements $h = \tau(f,p)$ such that $\Syst_2(f,p)$ and
$\Syst_3(h)$ behave the same for most pairs $(f,p)$, and the
distribution induced by $\tau$ is close to uniform.  This part
is also very different than \cite{CPS08v2}.

\subsubsection{Removing the simulator} 
In $\Syst_3$, the distinguisher accesses the random functions through the
simulator.  We want to show that the distinguisher can instead access the random
functions directly.

\begin{lemma}\label{lem:s3EqualsS4}
  Suppose that in $\Syst_3(h)$ the simulator $\SIMM(h)$ eventually
  answers a query $\F(i,x)$. Then, it is answered with $h(i,x)$.
\end{lemma}
\begin{proof}
  The simulator $\SIMM(h)$ either sets $\VarG_{i}(x) := h(i,x)$ or
  $\VarG_{i}(x_i) := x_{i-1} \oplus x_{i+1}$ in a call to $\adapt$.  For
  pairs $(i,x)$ which are set by the first call the lemma is clear.
  Otherwise, consider the $\adapt$ call: just before the call, the
  Feistel construction was evaluated either forward or backward in a
  call to $\Psi.\P(x_0,x_1)$ or $\Psi.\P^{-1}(x_{14},x_{15})$.  Since
  $\Psi$ evaluates $\P$ and $\P^{-1}$ with calls to $h$, the value 
  $\VarG_{i}(x_i)$ must be $h(i,x)$ as well.
\end{proof}

\subsubsection{Indifferentiability}
We can now prove Theorem~\ref{thm:mainTheorem}, which we restate here
for convenience.
\begin{reptheorem}{thm:mainTheorem}
  The $14$-round Feistel construction using $14$
  independent random functions is indifferentiable from a random
  permutation. 
\end{reptheorem}
\begin{proof}
  Fix a distinguisher $\Dist$ which makes at most $q$ queries.  We
  want to show that $\Dist$ fails to distinguish $\Syst_1 = (\URP,
  \SIM^{\URP})$ from $\Syst_4 = (\Psi^{\URF},\URF)$, and furthermore,
  that the simulator runs in polynomial time, except with negligible
  probability.

  Consider the system $\Syst_1$, where the distinguisher
  interacts with $(\URP,\SIM^{\URP})$.  If we replace $\URP$ with the
  two-sided random function $\TSRF$ as described above and the simulator $\SIM$ by $\SIMM(f)$, 
  then we obtain $\Syst_2$.  According to
  Lemma~\ref{lem:SimulatorIsEfficient} the number of queries the
  simulator makes in $\Syst_2$ to $\TSRF$ is at most  
  $ 6q^2 + 6q^2 + 1296q^8 \leq 1400 q^8.$
  Since Lemma~\ref{lem:PermAndTSRF} gives that the permutation is
  indistinguishable from a two-sided random function, we get for $q'=1400 q^8$ that the
  probability that $\Dist$ outputs $1$ differs by at most
  $\frac{10{q'}^2}{2^{2n}} \leq \frac{2\cdot 10^7\cdot q^{16}}{2^{2n}}$ 
  in $\Syst_1$ and $\Syst_2$.  
  Furthermore, also by 
  Lemma~\ref{lem:PermAndTSRF}, with probability at least $1-\frac{2\cdot 10^7\cdot q^{16}}{2^{2n}}$, 
  the first $1400 q^8$ queries and answers to $\URP$ in $\Syst_1$ and $\TSRF$ in $\Syst_2$ are equivalent, so that the simulator is
  efficient (that is, it makes at most $1400 q^8$ queries) in $\Syst_1$ with probability at least $1-\frac{2\cdot 10^7\cdot q^{16}}{2^{2n}}$. 

  Now, the probability that $\Dist$ outputs $1$ does not differ by
  more than $\frac{8\cdot 10^{19} \cdot q^{10}}{2^n}$ in $\Syst_2$ and $\Syst_3$, by
  Lemma~\ref{lem:eqS2andS3}.  Finally, since this implies that with
  probability $1-\frac{8\cdot 10^{19} \cdot q^{10}}{2^n}$ the distinguisher must give an answer
  in $\Syst_3$, we can also use Lemma~\ref{lem:s3EqualsS4} and get
  that the probability that the distinguisher answers $1$ differs in
  $\Syst_3$ and $\Syst_4$ by at most $\frac{8\cdot 10^{19} \cdot q^{10}}{2^n}$.
  
  Finally, from the above results, the probability that $\Dist$ outputs $1$ in $\Syst_1$ and $\Syst_4$ differs
  by at most
  \begin{align*}
 \frac{2\cdot 10^7\cdot q^{16}}{2^{2n}} + 2 \cdot \frac{8\cdot 10^{19} \cdot q^{10}}{2^n} < \frac{10^{22}\cdot q^{16}}{2^n},
 \end{align*}
 which implies the lemma.
\end{proof}

\subsection{Equivalence of the First and the Second Experiment}\label{sec:eqS1andS2}

We now show that $\Syst_1$ and $\Syst_2$ behave the same way, more
concretely, that our two-sided random function $\TSRF$ behaves as a
uniform random permutation $\URP$.
\begin{replemma}{lem:PermAndTSRF}
  Consider a random permutation over $2n$ bits, to which we add the
  procedure $\Check$ as in line~\ref{line:check} of the simulator $\SIM$.
  Then, a distinguisher which issues at most $q'$ queries to either the
  random permutation or to the two-sided random function $\TSRF$ has
  advantage at most $\frac{6(q')^2}{2^{2n}}$ in distinguishing the two
  systems.
\end{replemma}
	
	To prove the lemma, fix any deterministic distinguisher $\Dist$ that issues at most 
	$q'$ queries. 
	Consider the following random experiment $\SystE_0$:
	\begin{quote}
		\textbf{Experiment $\SystE_0$:} 
		$\Dist$ interacts with $\URP'(p)$, which is defined as follows: The procedures 
		$\URP'.\P$ and $\URP'.\P^{-1}$ are the same as $\TSRF.\P$ and $\TSRF.\P^{-1}$. 
		The $\Check$ procedure is defined as
\begin{lstlisting}[language=pseudocode]
public procedure $\URP'(p).\Check(x_0,x_1,x_{14},x_{15})$
     if $(\downarrow,x_0,x_1) \in \VarP$ then return $\VarP(\downarrow,x_0,x_1) = (x_{14},x_{15})$
     if $(\uparrow,x_{14},x_{15}) \in \VarP$ then return $\VarP(\uparrow,x_{14},x_{15}) = (x_{0},x_{1})$
     return $\P(x_{0},x_{1}) = (x_{14},x_{15})$ *\qquad/\!\!/ Note that the procedure $\URP'.\P$ is called!*
\end{lstlisting}	
		Finally, $p$ is the table of a uniform random permutation 
               (i.e., $p(\downarrow,x_{0},x_{1}) = (x_{14},x_{15})$ if and only if
               $p(\uparrow,x_{14},x_{15}) = (x_0,x_1)$).
	\end{quote}
	If we let $\Dist$ interact with $\URP$ (adding a $\Check$-procedure to $\URP$ in the most standard way,
       i.e., as in the simulator $\SIM$), then we get an experiment which behaves exactly as $\SystE_0$.

 	We next replace the table $p$ of the random permutation by a table that has uniform random
	entries: 
	\begin{quote}
		\textbf{Experiment $\SystE_2$:} $\Dist$ interacts with $\URP'(p)$, where the entries 
		of $p$ are chosen uniformly at random from $\{0,1\}^{2n}$.
	\end{quote}
	We will show that
	\begin{lemma}\label{lem:rfrpswitch}
	  The probability that $\Dist$ outputs $1$ in $\SystE_0$ differs by at most
	  $\frac{(q')^2}{2^{2n}}$ from the probability that $\Dist$ outputs $1$ in $\SystE_2$.
	\end{lemma}
	Finally we consider the experiment where $\Dist$ interacts with our two-sided random function.
	\begin{quote}
		\textbf{Experiment $\SystE_3$:} $\Dist$ interacts with $\TSRF(p)$, where the entries 
			of $p$ are chosen uniformly at random from $\{0,1\}^{2n}$.
	\end{quote}
         The only difference between $\SystE_2$ and $\SystE_3$
        is the change in the last line in the procedure $\Check$.
	We show that $\SystE_2$ behaves almost as $\SystE_3$:
	\begin{lemma}\label{lem:badChecks}
		The probability that $\Dist$ outputs $1$ in $\SystE_2$ differs by at most
	  $\frac{5(q')^2}{2^{2n}}$ from the probability that $\Dist$ outputs 
	  $1$ in $\SystE_3$.
	\end{lemma}
Lemma~\ref{lem:PermAndTSRF} then follows immediately, as 
	\begin{align}
		\Bigl| \Pr[\Dist \text{ outputs $1$ in }\SystE_0] - \Pr[\Dist \text{ outputs $1$ in }\SystE_3]\Bigr| \leq \frac{(q')^2}{2^{2n}} + \frac{5(q')^2}{2^{2n}} \leq \frac{6(q')^2}{2^{2n}}. \nonumber
	\end{align}	
	We proceed to prove Lemmas~\ref{lem:rfrpswitch} and \ref{lem:badChecks}.

\begin{proof}[Proof of Lemma~\ref{lem:rfrpswitch}]
	This proof is very similar to the proof that a (one-sided) random permutation can be
	replaced by a (one-sided) random function. 
	
	We introduce the following intermediate experiment: 
	\begin{quote}
		\textbf{Experiment $\SystE_1$:} 
		$\Dist$ interacts with $\URP''(p)$.  In $\URP''$, the procedure $\URP''.\P$ is defined as follows:
\begin{lstlisting}[language=pseudocode]
public procedure $\URP''.\P(x_0,x_1)$
     if $(\downarrow,x_0, x_1) \notin \VarP$ then
          $(x_{14},x_{15}) := p(\downarrow,x_{0},x_{1})$
          if $(\uparrow, x_{14},x_{15}) \in \VarP$ then
               $(x_{14},x_{15}) \inu \{0,1\}^{2n} \setminus \{(x'_{14},x'_{15}) | (\uparrow, x'_{14},x'_{15}) \in \VarP \}$
          $\VarP(\downarrow,x_0,x_1) := (x_{14},x_{15})$
          $\VarP(\uparrow,x_{14},x_{15}) := (x_{0},x_{1})$ 
     return $\VarP(\downarrow,x_0,x_1)$
\end{lstlisting}
		The procedure $\URP''.\P^{-1}$ is defined analogously, i.e., picks $(x_0,x_1)$ from $p$, and
               replaces it in case $(\downarrow,x_0,x_1) \in \VarP$.  The procedure $\Check$ is defined as in $\URP'.\Check$
               above.
		Finally, the entries of $p$ are chosen uniformly at random from $\{0,1\}^{2n}$.
	\end{quote}		
	
First consider the transition from $\SystE_0$ to $\SystE_1$.  
The procedure $\Check$ is the same in both experiments.  
Furthermore, a distinguisher can keep track of the table $\VarP$ and
it is also the same in both experiments, and so we only need to
consider the procedures $\P$ and $\P^{-1}$:  the 
procedure $\Check$ could be a part of the distinguisher.

Now, in both experiments, the values chosen in the procedures $\P$
and $\P^{-1}$ are chosen uniformly at random from the set of values
that do not correspond to an earlier query. Thus, $\SystE_0$ and
$\SystE_1$ behave identically.

	Now consider the transition from $\SystE_1$ to $\SystE_2$.
	Consider $\SystE_{1}$ and let $\mathsf{BadQuery}$ be the event that in 
	$\P$ we have $(\uparrow, x_{14},x_{15}) \in \VarP$, or in $\P^{-1}$ we have
	$(\downarrow, x_{0},x_{1}) \in \VarP$. 
	We show that this event is unlikely, and that the two experiments behave
	identically if $\mathsf{BadQuery}$ does not occur in $\SystE_1$.
	
	There are at most $q'$ queries to $\P$ or $\P^{-1}$ in an execution of 
	$\SystE_{1}$, since each $\Check$ query issues at most one query to 
	$\P$. 
	Observe that each table entry in $p$ is accessed at most once and thus
	each time $p$ is accessed it returns a fresh uniform random value.
	Since for each query there are at most $q'$ values in $\VarP$, 
	and $p$ contains uniform random entries, we have $
		\Pr[\mathsf{BadQuery} \text{ occurs in }\SystE_{1}] 
		\leq \frac{(q')^2}{2^{2n}}$.
	The systems $\SystE_{1}$ and $\SystE_2$ behave identically if
	$\mathsf{BadQuery}$ does not occur.  Thus,$
		\Bigl| \Pr[\Dist \text{ outputs $1$ in }\SystE_{1}] - 
		\Pr[\Dist \text{ outputs $1$ in }\SystE_2]\Bigr| \leq 
		\Pr_p[\mathsf{BadQuery} \text{ occurs in }\SystE_{1}] 
		\leq \frac{(q')^2}{2^{2n}}$.
                \nocite{Maurer02} 
\end{proof}

\begin{proof}[Proof of Lemma~\ref{lem:badChecks}]
	The event $\mathsf{BadCheck}$ 
	occurs for some $p$ if $\URP'.\Check$ returns true in the last line
        in $\SystE_2$ in an execution using $p$.
	The event $\mathsf{BadOverwrite}$ occurs for some $p$ if
        either in $\SystE_2$ or in $\SystE_3$, in any call to $\P$ or
	$\P^{-1}$, an entry of $\VarP$ is overwritten.\footnote{It
        would actually be sufficient to consider the system $\SystE_2$
        here, but we can save a little bit of work by considering 
        both $\SystE_2$ and $\SystE_3$.}
	The event $\mathsf{BadBackwardQuery}$ occurs if in $\SystE_2$ 
        there exist
	$(x_0, x_1), (x^*_{14}, x^*_{15})$ such that all of 
	the following	hold:
	\begin{enumerate}[(i)]
		\item	The query $\P(x_0,x_1)$ is issued in the last line of a $\Check$
				 	query, and $\VarP(\downarrow,x_0,x_1)$ is set to
				 	$(x^*_{14},x^*_{15}) = p(\downarrow, x_0, x_1)$. 
		\item After (i), the query $\P^{-1}(x^*_{14}, x^*_{15})$, or the query
					$\Check(x_0, x_1, x^*_{14}, x^*_{15})$ is issued. 
		\item The query $\P(x_0, x_1)$ is not issued by the distinguisher between point (i) 
					and point (ii).  
	\end{enumerate}
	We show that these events are unlikely, 
	and that $\SystE_2$ and $\SystE_3$ behave
	identically if the events do not occur for a given $p$.
	
	For $\mathsf{BadCheck}$ to occur in a fixed call 
	$\URP'.\Check(x_0, x_1, x_{14}, x_{15})$, it must be that 	
	$(\downarrow,x_0,x_1) \notin \VarP$ and $(\uparrow,x_{14},x_{15}) 
	\notin \VarP$. Thus, in the call 
	$\P(x_{0},x_{1})$ in the last line of $\Check$, 
	$\VarP(\downarrow,x_0,x_1)$ will be set to a fresh uniform random value
	$p(\downarrow, x_{0}, x_{1})$, and	this value is returned by $\P$. 
	Therefore, the probability over the choice of $p$ that
	$\P(x_{0},x_{1}) = (x_{14},x_{15})$ is at most $\frac{1}{2^{2n}}$.
	Since $\Check$ is called at most $q'$ times, we see that
        $
		\Pr_p[\mathsf{BadCheck}] 
		\leq \frac{q'}{2^{2n}}$.

	We now bound the probability that $\mathsf{BadOverwrite}$ occurs 
        in $\SystE_2$.  
        This only happens if a fresh uniform random entry 
	read from $p$ collides with an entry in $\VarP$. Since there are 
	at most $q'$ queries to $\P$ and $\P^{-1}$ and at most
	$q'$ entries in $\VarP$, we get $
		\Pr_p[\mathsf{BadOverwrite} \text{ occurs in }\SystE_2] 
		\leq \frac{(q')^2}{2^{2n}}$.
        The same argument gives a bound on $\mathsf{BadOverwrite}$ in
        $\SystE_3$, and so $\Pr_p[\mathsf{BadOverwrite}] \leq 
        \frac{2(q')^2}{2^{2n}}$.
	
We next estimate the 
probability of $(\mathsf{BadBackwardQuery} \land \lnot 
\mathsf{BadCheck})$.
Consider any pairs $(x_0,x_1) , (x^*_{14}, x^*_{15})$ such that (i) holds. Clearly, 
since $\mathsf{BadCheck}$ does not occur, the $\Check$ query returns false.
Now, as long as none of the queries $\P(x_0,x_1)$, $\P^{-1}(x^*_{14},x^*_{15})$
or $\Check(x_0,x_1,x^*_{14},x^*_{15})$ is done by the distinguisher, the 
value of $(x^*_{14},x^*_{15})$ is independently chosen at random from
all the pairs $(x'_{14},x'_{15})$ for which $\Check(x_0,x_{1},x'_{14},x'_{15})$ was
not queried.  Thus, the probability that in a single query, the distinguisher queries one of 
$\P^{-1}(x^*_{14},x^*_{15})$ or $\Check(x_0,x_1,x^*_{14},x^*_{15})$ is
at most $\frac{q'}{2^{2n}-q'} \leq \frac{2q'}{2^{2n}}$ 
(assuming $q' < \frac{2^{2n}}{2}$).  
	Since there are at most $q'$ $\Check$ queries, we find
        $
		\Pr_p[(\mathsf{BadBackwardQuery} \land \lnot \mathsf{BadCheck})] 
		\leq \frac{2(q')^2}{2^{2n}}$.

We proceed to argue that if the bad events do not occur, 
the two experiments behave identically.   
Thus, let $p$ be a table such that none of
$\mathsf{BadCheck}$, $\mathsf{BadBackwardQuery}$, and $\mathsf{BadOverwrite}$ 
occurs.

We first observe that the following invariant
holds in both $\SystE_2$ and $\SystE_3$: after any call to $\P$,
$\P^{-1}$ or $\Check$, if $\VarP(\downarrow,x_0,x_{1}) =
(x_{14},x_{15})$ for some values $(x_0,x_1,x_{14},x_{15})$,
 then $\VarP(\uparrow,x_{14},x_{15}) = 
(x_0,x_1)$, and vice versa.
The reason is simply that no value is ever overwritten in the tables, 
and whenever $\VarP(\uparrow,\cdot,\cdot)$ is set, then $\VarP(\downarrow,
\cdot,\cdot)$ is also set.

Next, we argue inductively that for a $p$ for which none of our bad events
occur, all queries and answers in $\SystE_2$ and $\SystE_3$ are the same. 

For this, we start by showing that (under the induction hypothesis), 
if a triple $(\downarrow,x_0,x_1)$ is in $\VarP$ in
the experiment $\SystE_3$, 
then the triple is in $\VarP$ in $\SystE_2$ as well, and both have the
same image $(x_{14},x_{15})$.  
This holds because of two reasons: first, in $\SystE_3$, each such entry
corresponds to an answer to a previously issued query to $\P$ or $\P^{-1}$.
This query was also issued in $\SystE_2$, and at that point the answer 
was identical, so that the entry $\VarP(\downarrow,x_0,x_{1})$ was identical 
(this also holds if the response in $\SystE_2$ 
is due to the entry $\VarP(\uparrow,x_{14},x_{15})$, 
because we saw above that this implies 
$\VarP(\downarrow,x_0,x_1) = (x_{14},x_{15})$).
Since the event $\mathsf{BadOverwrite}$ does not occur, the 
property will still hold later.  
(We remark that entries in the table $\VarP$ in $\SystE_2$ may
exist which are not in $\SystE_3$.)

We now establish our claim that all queries and answers of the distinguisher
in $\SystE_2$ and $\SystE_3$ are the same. 

Consider first a $\P$-query $\P(x_0,x_1)$.  
If $(\downarrow,x_0,x_1) \in \VarP$ in $\SystE_3$, the previous paragraph
gives the result for this query.
If $(\downarrow,x_0,x_1) \notin \VarP$ in both $\SystE_2$ and $\SystE_3$,
the same code is executed.
The only remaining case is $(\downarrow,x_0,x_1) \in \VarP$ in $\SystE_2$
and $(\downarrow,x_0,x_1) \notin \VarP$ in $\SystE_3$.
The only way this can happen is if the query $\P(x_0,x_1)$ was invoked
previously from a query to $\Check$ in $\SystE_2$, 
in which case the same entry 
$p(\downarrow, x_0,x_1)$ was used to set $\VarP$, and we get the result.

Next, we consider a $\P^{-1}$-query $\P^{-1}(x^*_{14},x^*_{15})$.
Again, the only non-trivial case is if $(\uparrow,x^*_{14},x^*_{15}) 
\in \VarP$ in
$\SystE_2$ and $(\uparrow,x^*_{14},x^*_{15}) \notin \VarP$ in $\SystE_3$.
This is only possible if during some query to $\Check(x_0,x_1,\cdot,\cdot)$
in $\SystE_2$, the last line invoked $\P(x_{0},x_1)$, and 
$(x^*_{14},x^*_{15}) = p(\downarrow,x_0,x_1)$.
Since it also must be that until now the distinguisher never invoked 
$\P(x_0,x_1)$ (otherwise, $\VarP(\uparrow,x^*_{14},x^*_{15})=(x_0,x_1)$ 
in $\SystE_3$), this implies that the event $\mathsf{BadBackwardQuery}$ 
must have happened.

Finally, consider a call $\Check(x_0,x_1,x_{14},x_{15})$ to $\Check$.
In case $(\downarrow,x_0,x_1) \in \VarP$ in $\SystE_3$ and in case
$(\downarrow,x_0,x_1) \notin \VarP$ in $\SystE_2$, line \ref{line:tsrf_check1}
behaves the same in both $\SystE_2$ and $\SystE_3$.
If $(\downarrow,x_0,x_1) \in \VarP$ in $\SystE_2$ and 
$(\downarrow,x_0,x_1) \notin \VarP$ in $\SystE_3$, then first in $\SystE_3$,
$\Check$ returns $\mathsf{false}$.  In $\SystE_2$, $\Check$ can only
return $\mathsf{true}$ if the event $\mathsf{BadBackwardQuery}$ occurs.

The second if statement in $\Check$ (in $\SystE_3$ this is line \ref{line:tsrf_check2} of $\TSRF$) 
can only return $\mathsf{false}$ in both $\SystE_2$ and $\SystE_3$: otherwise, the 
first if statement in $\Check$ (in $\SystE_3$ this is line \ref{line:tsrf_check1} of $\TSRF$) 
would already have returned $\mathsf{true}$.
This is sufficient, because the event $\mathsf{BadCheck}$ does not occur,
and so the last line of $\Check$ in both systems also returns $\mathsf{false}$.
	
	Thus,
	\begin{align*}
		\Bigl| \Pr_p[\Dist \text{ outputs $1$ in }&\SystE_2] - \Pr_p[\Dist \text{ outputs $1$ in }\SystE_3]\Bigr| \\
		& \leq \Pr_p[(\mathsf{BadCheck} \lor \mathsf{BadOverwrite} \lor 
			\mathsf{BadBackwardQuery})] \\
		& \leq \frac{q'}{2^{2n}} +\frac{2(q')^2}{2^{2n}}+  
			\frac{2(q')^2}{2^{2n}} \leq 	
		  \frac{5(q')^2}{2^{2n}}. \qedhere
	\end{align*}	
\end{proof}

\subsection{Complexity of the Simulator}\label{sec:SimulatorIsEfficient}
In this section we show that the simulator is efficient in scenario
$\Syst_2$.  
\begin{lemma}\label{lem:nOfOuterChains}
  Consider $\Syst_2$, and suppose that the distinguisher
  makes at most $q$ queries.  Then, the simulator dequeues at most $q$
  times a partial chain of the form $(x_1,x_2,1,\ell)$ for which
  $(x_1,x_2,1) \notin \CompletedChains$.
\end{lemma}
\begin{proof}
  Consider such a dequeue call and let $(x_1,x_2,1,\ell)$ be the
  partial chain dequeued for which $(x_1,x_2,1) \notin
  \CompletedChains$.  A chain $(x_1,x_2,1,\ell)$ is only enqueued when
  $(x_1,x_2,x_{13},x_{14})$ is detected, and since neither
  $\VarG_{1}(x_1)$ nor $\VarG_{14}(x_{14})$ are ever overwritten, this
  means that we can find a unique $4$-tuple $(x_0,x_1,x_{14},x_{15})$
  associated with $(x_1,x_2,1)$ for which
  $\Check(x_0,x_1,x_{14},x_{15})$ was true at the moment
  $(x_1,x_2,1,\ell)$ was enqueued.  We can now find a unique query to
  $p$ which corresponds to $(x_0,x_{1},x_{14},x_{15})$: pick
  $p(\uparrow,x_{14},x_{15}) = (x_0,x_{1})$ if this was a query and
  its answer, or otherwise $p(\downarrow,x_0,x_1) = (x_{14},x_{15})$.
  This table entry of $p$ was accessed during a call to $\P$ or
  $\P^{-1}$, and this call was made either by the distinguisher or the
  simulator.  We argue that this call cannot have been by the
  simulator.  The simulator issues such calls only
  when it completes a chain, and after this completion, it adds
  $(x_1,x_0 \oplus \VarG_{1}(x_1),1)$ to $\CompletedChains$, and so it cannot have been that
  $(x_1,x_2,1) \notin \CompletedChains$ when it was dequeued.  Thus, we
  found a unique query of the distinguisher associated with this
  dequeue call. Finally, note that after $(x_1, x_2, 1)$ is completed by the
  simulator, $(x_1, x_2, 1)$ is added to $\CompletedChains$. 
  Thus, there are at most $q$ such dequeue calls.
\end{proof}

\begin{replemma}{lem:SimulatorIsEfficient}
  Consider $\Syst_2$, and suppose that the distinguisher makes at most
  $q$ queries.  Then, at any point in the execution we have $|\VarG_i|
  \leq 6q^2$ for all $i$.  Furthermore, there are at most $6q^2$
  queries to both $\TSRF.\P$, and $\TSRF.\P^{-1}$, and at most $1296q^8$
  queries to $\TSRF.\Check$.
\end{replemma}
\begin{proof}
  We first show that $|\VarG_7| \leq 2q$ and $|\VarG_8| \leq 2q$.
  Assignments $\VarG_{7}(x_7) := f(7,x_7)$ and $\VarG_8(x_8) :=
  f(8,x_8)$ only happen in two cases: either when the
  distinguisher directly queries the corresponding value using $\F$, or
  when the simulator completes a chain $(x_1,x_2,1,\ell)$ which it
  dequeued.  There can be at most $q$ queries to $\F$, and according to
  Lemma~\ref{lem:nOfOuterChains} there are at most $q$ such chains
  which are completed, which implies the bound.

  The set $\VarG_{i}$ can only be enlarged by $1$ in the following
  cases: if the distinguisher queries $\F(i,\cdot)$, if a chain of the
  form $(x_1,x_2,1,\ell)$ is dequeued and not in $\CompletedChains$,
  or if a chain $(x_7,x_8,7,\ell)$ is dequeued and not in $\CompletedChains$.  
  There are at most
  $q$ events of the first kind, at most $q$ events of the second kind
  (using Lemma~\ref{lem:nOfOuterChains}), and at most $|\VarG_7|\cdot
  |\VarG_8| \leq 4q^2$ events of the last kind, giving a total of
  $4q^2 + 2q \leq 6q^2$.  
 
 	A query to $\TSRF.\P$ or $\TSRF.\P^{-1}$ can be made either by the distinguisher,
 	or by the simulator when it completes a chain. At most $q$ events of the
 	first kind, and at most $q + 4q^2$ events of the second kind are possible. 
 	Thus, at most $6q^2$ of these queries occur.
 	The number of $\Check$ queries by the simulator is bounded by 
 	$|\VarG_1 \times \VarG_2 \times \VarG_{13}\times \VarG_{14}| \leq (6q^2)^4$.
\end{proof}

\subsection{Equivalence of the Second and the Third Experiment}\label{sec:eqS2andS3}
This section contains the core of our argument: We prove
Lemma~\ref{lem:eqS2andS3}, which states that $\Syst_2(f,p)$ and
$\Syst_3(h)$ have the same behaviour for uniformly chosen $(f,p)$ and
$h$.  For most part of the analysis, we consider the scenario
$\Syst_2(f,p)$.  We let $\VarG = (\VarG_1,\ldots,\VarG_{14})$ be the
tuple of tables of the simulator $\SIMM(f)$ in the execution.

\subsubsection{Partial chains}

\paragraph{Evaluating partial chains.}
A \emph{partial chain} is a triple $(x_{k},x_{k+1},k) \in \{0,1\}^{n}
\times \{0,1\}^n \times \{0,\ldots,14\}$.  Given such a partial chain
$C$, and a set of tables $\SIMM.\VarG$ and $\TSRF.\VarP$, it can be
that we can move ``forward'' or ``backward'' one step in the Feistel
construction.  This is captured by the functions $\Next$ and $\Prev$.
Additionally, the functions $\val^+$ and $\val^{-}$ allow us to access
additional values of the chain indexed by $C$, $\val^+$ by invoking
$\Next$, and $\val^{-}$ by invoking $\Prev$.  The function $\val$
finally gives us the same information in case we do not want to bother
about the direction.

\begin{definition}
  Fix a set of tables $\VarG=\SIMM.\VarG$ and $\VarP = \TSRF.\VarP$ in
  an execution of $\Syst_2(f,p)$.  Let $C = (x_k,x_{k+1},k)$ be a
  partial chain.  We define the functions
  $\Next$, $\Prev$, $\val^+$, $\val^-$, and $\val$ with the following
  procedures (for a chain $C = (x_k,x_{k+1},k)$, we let $C[1] = x_k$,
  $C[2] = x_{k+1}$ and $C[3] =k$):
  
\begin{lstlisting}[language=pseudocode]
procedure $\Next(x_k,x_{k+1},k)$:
     if $k < 14$ then
          if $x_{k+1} \notin \VarG_{k+1}$ then return $\bot$
          $x_{k+2} := x_{k} \oplus \VarG_{k+1}(x_{k+1})$
          return $(x_{k+1},x_{k+2},k+1)$
     else if $k = 14$ then
          if $(\uparrow,x_{14},x_{15}) \notin \VarP$ then return $\bot$
          $(x_{0},x_{1}) := \VarP(\uparrow,x_{14},x_{15})$
          return $(x_0,x_1,0)$

procedure $\Prev(x_k,x_{k+1},k)$:
     if $k > 0$ then
          if $x_{k} \notin \VarG_{k}$ then return $\bot$
          $x_{k-1} := x_{k+1} \oplus \VarG_{k}(x_k)$
          return $(x_{k-1},x_{k},k-1)$
     else if $k = 0$ then
          if $(\downarrow,x_{0},x_{1}) \notin \VarP$ then return $\bot$
          $(x_{14},x_{15}) := \VarP(\downarrow,x_{0},x_{1})$
          return $(x_{14},x_{15},14)$

procedure $\val^+_i(C)$
     while $(C \neq \bot )\land (C[3]\notin \{ i-1, i\})$ do
          $C := \Next(C)$
     if $C = \bot$ then return $\bot$
     if $C[3] = i$ then return $C[1]$ else return $C[2]$

procedure $\val^-_i(C)$
     while $(C \neq \bot )\land (C[3]\notin \{ i-1, i\})$ do
          $C := \Prev(C)$
     if $C = \bot$ then return $\bot$
     if $C[3] = i$ then return $C[1]$ else return $C[2]$

procedure $\val_i(C)$
     if $\val_{i}^{+}(C) \neq \bot$ return $\val_{i}^{+}(C)$ else return $\val_{i}^{-}(C)$
\end{lstlisting}
\end{definition}

We use the convention that $\bot \notin \VarG_i$ for any $i \in
\{1,\ldots,14\}$.  Thus, the expression $\val_{i}(C) \notin \VarG_i$
means that $\val_{i}(C) = \bot$ or that $\val_{i}(C) \neq
\bot$ and $\val_{i}(C) \notin \VarG_i$.
Furthermore, even though $\Next$ and $\Prev$ may return $\bot$, according to our 
definition of partial chains, $\bot$ is not a partial chain.

\paragraph{Equivalent partial chains}

We use the concept of equivalent partial chains:
\begin{definition}\label{def:equivalence}
  For a given set of tables $G$ and $P$, two partial chains $C$ and $D$ are
  \emph{equivalent} (denoted $C \equiv D$) if they are in the reflexive transitive closure of
  the relations given by $\Next$ and $\Prev$.
\end{definition}
In other words, two chains $C$ and $D$ are equivalent if $C=D$, or
if $D$ can be obtained by applying $\Next$ and $\Prev$ finitely many times on
$C$.

Note that this relation is not an equivalence relation, since it is not necessarily 
symmetric.\footnote{The symmetry can be violated if in the two-sided random function 
$\TSRF$ an entry of the table $\VarP$ is overwritten.} 
However, we will prove that for most executions of $\Syst_2(f,p)$ it actually is
symmetric and thus an equivalence relation. Furthermore, it is possible that two different chains $(x_0,x_1,0)$ and
$(y_0,y_1,0)$ are equivalent (e.g., by applying $\Next$ 15 times).  
While we eventually show that for most
executions of $\Syst_2(f,p)$ this does not happen, this is not easy to
show, and we cannot assume it for most of the following proof. 

\subsubsection{Bad events and good executions}

As usual in indistinguishability proofs, for some pairs $(f,p)$ the
system $\Syst_2(f,p)$ does not behave as ``it should''.  In this
section we collect events which we show later to occur with low
probability.  We later study $\Syst_2(f,p)$ for pairs $(f,p)$ for
which these events do not occur.

All events occur if some unexpected collision happens to one of the
partial chains which can be defined with elements of
$\VarG_1,\ldots,\VarG_{14}$ and $\VarP$.
\begin{definition}
  The set of \textit{table-defined partial chains} contains all chains
  $C$ for which $\Next(C) \neq \bot$ and $\Prev(C) \neq \bot$.
\end{definition}
If $C = (x_k,x_{k+1},k)$ for $k \in \{1,\ldots,13\}$, then $C$ is
table-defined if and only if $x_{k} \in \VarG_{k}$ and $x_{k+1} \in
\VarG_{k+1}$.  For $k \in \{0,14\}$, $C$ is table-defined if the
``inner'' value is in $\VarG_{1}$ or $\VarG_{14}$, respectively, and
a corresponding triple is in $\VarP$.

\paragraph{Hitting permutations.}
Whenever we call the two-sided random function, a query to the table
$p$ may occur.  If such a query has unexpected effects, the event
$\BadP$ occurs.

\begin{definition}
  The event $\BadP$ occurs in an execution of $\Syst_2(f,p)$ if 
  immediately after
  a call $(x_{14},x_{15}) := p({\downarrow, x_0,x_1})$ in line \ref{line:tsrf_queryP1}
  of $\TSRF$ we have one of
  \begin{itemize}
  \item $(\uparrow,x_{14},x_{15}) \in\VarP$,
  \item $x_{14} \in \VarG_{14}$.
  \end{itemize}
  Also, it occurs if immediately after a call $(x_0,x_1) :=
  p(\uparrow,x_{14},x_{15})$ in line \ref{line:tsrf_queryP2} of $\TSRF$ we have one of
  \begin{itemize}
  \item $(\downarrow,x_{0},x_{1}) \in\VarP$,
  \item $x_{1} \in \VarG_{1}$.
  \end{itemize}
\end{definition}
If $\BadP$ does not occur, then we will be able to show that evaluating
$\P$ and $\P^{-1}$ is a bijection, since no value is overwritten. 

\paragraph{Chains hitting tables.}
Consider an assignment $\VarG_{i}(x_i) := f(i,x_i)$.  Unless
something unexpected happens, such an assignment allows evaluating
$\Next(C)$ at most \emph{once} more.

\begin{definition}\label{def:badlyhit}
  The event $\BadlyHit$ occurs if one of the following happens
  in an execution of $\Syst_2(f,p)$:
  \begin{itemize}
  \item After an assignment $\VarG_k(x_k) := f(k,x_k)$ there is a table-defined chain
  $(x_{k},x_{k+1},k)$ such that $\Prev(\Prev(x_{k},x_{k+1},k)) \neq \bot$.
  \item After an assignment $\VarG_k(x_k) := f(k,x_k)$ there is a table-defined chain
  $(x_{k-1},x_k,k-1)$ such that $\Next(\Next(x_{k-1},x_{k},k-1))\neq \bot$.
  \end{itemize}
\end{definition}
Furthermore, if the above happens for some chain $C$, and $C'$ is a
chain equivalent to $C$ before the assignment, we say that $C'$ badly
hits the tables.

We will later argue that the event $\BadlyHit$ is unlikely, because a chain
only badly hits the tables if $f(k,x_k)$ takes a very particular value.
For this (and similar statements), 
it is useful to note that the set of table-defined chains
after an assignment $\VarG_k(x_k) := f(k,x_k)$ 
\emph{does not depend} on the value of $f(k,x_k)$, as the reader can verify.

\paragraph{Colliding chains.}
Two chains $C$ and $D$ collide if after an assignment suddenly
$\val_i(C) = \val_i(D)$, even though this was not expected.  More
exactly:
\begin{definition}\label{def:badlycollide}
  Let $\VarG$ and $\VarP$ be a set of tables, let $x_k \notin
  \VarG_k$, and consider two partial chains $C$ and $D$.  An
  assignment $\VarG_{k}(x_k) := y$ \emph{badly collides} $C$ and $D$
  if for some $\ell \in \{0,\ldots,15\}$ and $\sigma, \rho \in
  \{+,-\}$ all of the following happen:
  \begin{itemize}
  \item Before the assignment, $C$ and $D$ are not equivalent.
  \item Before the assignment, $\val^{\sigma}_{\ell}(C) = \bot$ or
    $\val^{\rho}_{\ell}(D) = \bot$.
  \item After the assignment, $\val^{\sigma}_{\ell}(C) =
    \val^{\rho}_{\ell}(D) \neq \bot$.
  \end{itemize}

  We say that the event $\BadlyCollide$ occurs in an execution
  $\Syst_2(f,p)$, if an assignment of the form $\VarG_{i}(x_{i}) :=
  f(i,x_i)$ makes two partial chains badly collide, and the 
  two chains are table-defined after the assignment.
\end{definition}

Finally, we say that a pair $(f,p)$ is \textit{good} if none of the above three
events happen in an execution of $\Syst_2(f,p)$.

\subsubsection{Bad events are unlikely}
In this subsection we show that all the bad events we have introduced
are unlikely.  

\paragraph{Hitting permutations} 
\begin{lemma}\label{lem:noBadP}
  Suppose that $\Syst_2(f,p)$ is such that for any $(f,p)$ the tables
  satisfy $|\VarG_{i}| \leq T$ for all $i$ and $|\VarP|\leq T$ at any
  point in the execution.  Then, the probability over the choice of
  $(f,p)$ of the event $\BadP$ is at most $\frac{2T^2}{2^n}$.
\end{lemma}
\begin{proof}
  For any query to $p$, only $2$ events are possible.  In both
  cases, these events have probability at most $\frac{T}{2^n}$.  Since
  at most $T$ positions of $p$ can be accessed without violating
  $|\VarP|\leq T$ we get the claim.
\end{proof}

\paragraph{Chains hitting tables.}
We now show that the event $\BadlyHit$ is unlikely.
\begin{lemma}\label{lem:noBadHits}
  Suppose that $\Syst_2(f,p)$ is such that for any $(f,p)$ the tables
  satisfy $|G_i| \leq T$ for all $i$ and $|\VarP|\leq T$ at any point
  in the execution.  Then, the probability over the choice of $(f,p)$
  of the event $\BadlyHit$ is at most $30 \frac{T^3}{2^n}$.
\end{lemma}
\begin{proof}
  We first bound the probability of the first event, i.e., 
  that after the assignment $\VarG_{k}(x_k) := f(k,x_k)$ there is 
  a table-defined chain $C=(x_{k},x_{k+1},k)$ such that $\Prev(\Prev(C)) \neq \bot$. This can only happen 
  if $x_{k+1} \oplus \VarG_{k}(x_k)$ has one of at most $T$ different values (namely, it has to be in $\VarG_{k-1}$ in case
  $14 \geq k \geq 2$ or in $\VarP$ together with $x_{1}$ in case $k = 1$). Thus, for fixed $x_{k+1} \in \VarG_{k+1}$ the probability
  that $\Prev(\Prev(C)) \neq \bot$ is at most $T/2^n$. Since there are at most $T$ possible choices for $x_{k+1}$ 
  (this also holds if $k=14$) the total probability is at most $T^2/2^n$.  
  
  The analogous probability for $\Next$ is exactly the same and thus the probability of $\BadlyHit$ for one assignment is at most
  $2\cdot T^2/2^n$. In total, there are at most $14\cdot T$ assignments of the form 
  $\VarG_{k}(x_k) := f(k,x_k)$, and thus the probability of $\BadlyHit$ is at most $28 T^3/2^n$.
\end{proof}

\paragraph{Colliding chains}
We next show that it is unlikely that chains badly collide.  First, we give a useful lemma which explains how the
chains behave when they do not badly hit $\VarG$: only one value
$\val_i(C)$ can change from $\bot$ to a different value.
\begin{lemma}\label{lem:resultOfNoBadHits}
  Consider a set of tables $\VarG$ and $\VarP$, $x_{k}
  \notin \VarG_k$, fix a partial chain $C$, and suppose that $C$ does not
  badly hit the tables due to the assignment $\VarG_{k}(x_k) := f(k,x_k)$.  
  Then, for each chain $C$ and each $\sigma \in \{+,-\}$ there
  is at most one value~$i$ such that~$\val^{\sigma}_{i}(C)$ 
  differs after the assignment from before the assignment.
  Futhermore, if some value changes, then it changes from $\bot$ to
  a different value, and
  \begin{equation*}
	i = \left\{
	\begin{array}{rl}
	k+1 & \text{if } \sigma = +\\
	k-1 & \text{if } \sigma = -,
	\end{array} \right.
	\end{equation*}
	and $\val^\sigma_k(C) = x_k$ before the assignment.
\end{lemma}
\begin{proof}
We give the proof for $\sigma = +$, the other case is symmetric.
First, we see that if $\val_{i}^+(C) \neq \bot$ before the assignment,
then it does not change due to the assignment.  This follows by 
induction on the number of calls to $\Next$ in the evaluation of $\val^+$,
and by noting that $\VarG_k(x_k) := f(k,x_k)$ is not called
if $x_k \in \VarG_k$ in the simulator.

  Thus, suppose that $\val_{i}^+(C) = \bot$. 
   This means that during the evaluation of
  $\val_{i}^{+}(C)$ at some point the evaluation stopped.  
  This was either because a queried triple was not in
  $\VarP$, or because a value $x_j$ was not in $\VarG_{j}$ during the
  evaluation.  In the first case, the evaluation of $\val_{i}^{+}(C)$ 
  will not change due to an assignment to $\VarG_{k}(x_k)$.
  In the second case, the evaluation can only change if it stopped
  because $\val^{+}_{k}(C) = x_{k}$. Then after the assignment,
  $\val_{k+1}^{+}(C)$ will change from $\bot$ to a different value.
  Since $C$ does not badly hit the tables under the assignment,
  $\val_{k+1}^{+}(C) \notin \VarG_{k+1}$ after this assignment (in
  case $k+1 < 15$), and $(\uparrow,\val_{14}^+(C),\val_{15}^+(C))
  \notin \VarP$ (in case $k+1 = 15$).  Thus, there is only one change in the
  evaluation. 
\end{proof}

Instead of showing that $\BadlyCollide$ is unlikely, it is
slightly simpler to consider the event $(\BadlyCollide \land
\lnot \BadlyHit \land \lnot\BadP)$.
\begin{lemma}\label{lem:noBadCollides}
  Suppose that $\Syst_2(f,p)$ is such that for any $(f,p)$ the tables
  satisfy $|\VarG_{i}| \leq T$ for all $i$ and $|\VarP|\leq T$ at any
  point in the execution.  Then, the probability of the event
  $(\BadlyCollide \land \lnot \BadlyHit \land
  \lnot\BadP)$ is at most $15\,000 \frac{T^5}{2^n}$.  
\end{lemma}
\begin{proof}
If the event $(\BadlyCollide \land \lnot \BadlyHit \land \lnot\BadP)$
happens for a pair $(f,p)$, then there is some point in the
execution where some assignment $\VarG_k(x_k) := f(k,x_k)$
makes a  pair $(C,D)$ of partial
chains collide as in Definition~\ref{def:badlycollide}.
After this assignment, both $(C,D)$ are table defined, and 
$\val_{\ell}^{\sigma}(C) = \val_{\ell}^{\rho}(D)$.

We distinguish some cases: first suppose that
$\val^{-}_{\ell}(C) = \val^{-}_{\ell}(D) = \bot$ before the assignment,
and $\val^{-}_{\ell}(C) = \val^{-}_{\ell}(D) \neq \bot$ after the assignment.
Since $\BadlyHit$ does not happen, Lemma~\ref{lem:resultOfNoBadHits}
implies that before the assignment, $\val^{-}_{\ell+1}(C) = 
\val^{-}_{\ell+1}(D)$, and furthermore $\ell+1 \in \{1,\ldots,14\}$.  
Also, since $C \not\equiv D$ before the 
assignment, it must be that before the assignment $\val^{-}_{\ell+2}(C) \neq 
\val^{-}_{\ell+2}(D)$.  However, this implies that $\val^{-}_{\ell}(C) \neq
\val^{-}_{\ell}(D)$ after the assignment.  Therefore, this case is 
impossible and has probability~$0$.

Next, we consider the case $\val^{-}_{\ell}(C) = \bot$,
$\val^{-}_{\ell}(D) \neq \bot$ before the assignment, and 
$\val^{-}_\ell(C) = \val^{-}_\ell(D)$ after the assignment.
Since $D$ is table defined after the assignment, and 
we assume $\BadlyHit$ does not occur, 
by Lemma~\ref{lem:resultOfNoBadHits} the value $\val^{-}_{\ell}(D)$ does not
change due to the assignment.
Since $\val^{-}_\ell(C) = \val^{-}_{\ell+2}(C) \oplus \VarG_{\ell+1}(x_{\ell+1})$,
and $\VarG_{\ell+1}(x_{\ell+1})$ is chosen uniformly at random, the 
probability that it exactly matches $\val^{-}_\ell(D)$ is $2^{-n}$.

The next two cases are similar to the previous ones, we give them for
completeness.  The first of these two is that
$\val^{+}_{\ell}(C) = \val^{-}_{\ell}(D) = \bot$ before the assignment,
and $\val^{+}_{\ell}(C) = \val^{-}_{\ell}(D) \neq \bot$ after the assignment.
However, due to Lemma~\ref{lem:resultOfNoBadHits} this is impossible:
we would need both $k = \ell+1$ and $k = \ell-1$ for both values to change 
as needed.

Then, we have the case that $\bot = \val^{+}_{\ell}(C) \neq \val^{-}_{\ell}(D)$
before the assignment, and $\val^{+}_{\ell}(C) = \val^{-}_{\ell}(D)$ after the
assignment.  
Again, $\val^{-}_\ell(D)$ does not change by the assignment due to 
Lemma~\ref{lem:resultOfNoBadHits}, and also similarly to before,
the probability that $\val^+_{\ell-2}(C) \oplus f(\ell-1,\val_{\ell-1}^{-}(C))
= \val^{+}_\ell(D)$ is $2^{-n}$.

Bounds on the probability of the the
 4 remaining cases follow by symmetry of the construction.
   
  There are $4$ possibilities for the values of $\sigma$ and $\rho$. 
  As previously, there can be at most $14 \cdot T$ assignments of the form $\VarG_k(x_k):=f(k,x_k)$.
  For each assignment, there are at most $15 \cdot T^2$ possibilities for a chain
  to be table-defined before the assignment. Since the chains that are table-defined after the assignment, but 
  not before must involve $x_k$, there are at most $2 \cdot T$ possibilities for a fixed assignment.
  Thus the probability of the event $(\BadlyCollide \land \lnot \BadlyHit \land
  \lnot\BadP)$ is at most $\frac{4\cdot 14 \cdot T \cdot (15 \cdot T^2 + 2 \cdot T)^2}{2^n} \leq \frac{4\cdot 14 \cdot 16^2 \cdot T^5}{2^n}$.
\end{proof}

\paragraph{Most executions are good}

We collect our findings in the following lemma:
\begin{lemma}\label{lem:goodPairsAreLikely}
  Suppose that $\Syst_2(f,p)$ is such that for any $(f,p)$ the tables
  satisfy $|\VarG_{i}| \leq T$ for all $i$ and $|\VarP|\leq T$ at any
  point in the execution.  Then, the probability that a uniform
  randomly chosen $(f,p)$ is not good is at most
  $16\,000\cdot\frac{T^5}{2^n}$.
\end{lemma}
\begin{proof}
  This follows immediately from Lemmata~\ref{lem:noBadP},
  \ref{lem:noBadHits}, and \ref{lem:noBadCollides}.
\end{proof}
\subsubsection{Properties of good executions}
\label{sec:propertiesOfGoodExecutions}

We now study executions of $\Syst_2(f,p)$ with good pairs $(f,p)$.
One of the main goals of this section is to prove
Lemma~\ref{lem:noSetFafterF}, which states that no call to $\forceVal$
overwrites a previous entry.  However, we later also use
Lemma~\ref{lem:atEndFeistelGivesTheRightThing} (in good executions,
evaluating the Feistel construction for a pair $(x_0,x_1)$
leads to $P(x_0,x_1)$ --- if not, it would be silly to hope that our simulator 
emulates a Feistel construction), and Lemma~\ref{lem:nAdaptEqualsNofPqueries} (the
number of times $\adapt$ is called in $\SIMM(f)$ is exactly the same as
the number of times the table~$p$ is queried in $\TSRF(p)$).

We first state two basic lemmas about good executions: 

\begin{lemma} \label{lem:basicsAboutGoodExecutions}
	Consider an execution of $\Syst_2(f,p)$ with a good pair $(f,p)$. Then, we have
	\begin{enumerate}[(a)]
		\item	For any partial chain $C$, if $\Next(C) = \bot$ before an assignment $\VarG_{i}(x_{i}) :=
  				f(i,x_i)$ or a pair of assignments to $\VarP$ in~$\TSRF$, then if $C$ is table-defined after
  				the assignment(s), $\Next(\Next(C)) = \bot$. 
  				
  				For any partial chain $C$, if $\Prev(C) = \bot$ before an assignment $\VarG_{i}(x_{i}) :=
  				f(i,x_i)$ or a pair of assignments to $\VarP$ in~$\TSRF$, then if $C$ is table-defined after
  				the assignment(s), $\Prev(\Prev(C)) = \bot$. 
  	\item For all partial chains $C$ and $D$, we have $\Next(C) = D \iff \Prev(D) = C$. 
  	\item The relation $\equiv$ between partial chains is an equivalence relation.
	\end{enumerate}
\end{lemma}
\begin{proof}
	For assignments of the form $\VarG_{i}(x_{i}) :=f(i,x_i)$, (a) follows directly since $\BadlyHit$ does not occur. For the assignments to $\P$, 
        it follows because $\BadP$ does not occur. 

	The statement (b) is trivial for chains $C=(x_k, x_{k+1},k)$ with $k\in \{0, \ldots, 13\}$, since evaluating the Feistel construction one step forward 
	or backward is bijective. For $k=14$ we get (b) because $\BadP$ does not occur: no value is ever overwritten in a call to $\P$ or $\P^{-1}$, and thus
	evaluating $\P$ and $\P^{-1}$ is always bijective. 

	To see (c), observe that the relation $\equiv$ is symmetric because of (b), and it is reflexive and transitive by definition.
\end{proof}

\begin{lemma}\label{lem:structureOfEquivalentChains}
  Consider an execution of $\Syst_2(f,p)$ with a good pair $(f,p)$. 
  Suppose that at any point in the execution, two table-defined 
  chains $C$ and $D$ are equivalent.
  Then, there exists a sequence of partial chains $C_1,\ldots,C_r$, $r \geq 1$,
  such that
  \begin{itemize}
   \item $C = C_1$ and $D=C_r$, or else $D=C_1$ and $C=C_r$, 
   \item $C_i = \Next(C_{i-1})$ and $C_{i-1} = \Prev(C_{i})$,
   \item and each $C_i$ is table-defined.
  \end{itemize}
\end{lemma}
\begin{proof}
  Since $C \equiv D$, $D$ can be obtained from $C$ by applying $\Next$ and
  $\Prev$ finitely many times.  A shortest such sequence can only apply
  either $\Next$ or $\Prev$, due to 
  Lemma~\ref{lem:basicsAboutGoodExecutions} (b).
  The resulting sequence of chains is the sequence we are looking for 
  (possibly backwards) --
  note that the last bullet point also follows by  
  Lemma~\ref{lem:basicsAboutGoodExecutions} (b).
\end{proof}

We first show that assignments $\VarG_i(x_i) := f(i,x_i)$ and also
assignments to $\VarP$ in $\TSRF$ do not change the
equivalence relation for chains which were defined before.

\begin{lemma}\label{lem:equivalentChainsStayEquivalent}
  Consider an execution of $\Syst_2(f,p)$ with a good pair $(f,p)$.
  Let $C$ and $D$ be two table-defined partial chains at some point
  in the execution.  Suppose that after this point, there is an
  assignment $\VarG_i(x_i) := f(i,x_i)$ or a pair of assignments to
  $\VarP$ in~$\TSRF$.  Then $C \equiv D$ before the assignment(s) if and only if $C\equiv D$
  after the assignment(s).
\end{lemma}
\begin{proof}
  Suppose that $C \equiv D$ before the assignment.  
  We apply Lemma~\ref{lem:structureOfEquivalentChains} to get a sequence 
  $C_1,\ldots,C_r$ of table-defined chains.  This sequence still implies
  equivalence after the assignment, since no value in $\VarP$ or
  $\VarG$ can be overwritten by one of the assignments considered
  (recall that $\BadP$ does not occur), i.e.~the conditions of
  Definition~\ref{def:equivalence} still hold if they held previously, thus $C \equiv D$ 
  after the assignment(s).
  
  Now suppose that $C$ and $D$ are equivalent after the assignment.  
  Again consider the sequence $C_1,\ldots,C_r$ as given by
  Lemma~\ref{lem:structureOfEquivalentChains}.  Suppose first that
  the assignment was $\VarG_{i}(x_i) := f(i,x_i)$.  If $x_i$
  was not part of any chain, then $C_1,\ldots,C_r$ are a sequence
  which show the equivalence of $C$ and $D$ before the assignment.
  Otherwise, there is $j$ such that the chains $C_{j-1}$ and $C_j$ 
  have the form
  $C_{j-1} = (x_{i-1},x_{i},i-1)$ and $C_j = (x_{i},x_{i+1},i)$.
  It is not possible that $C_{j} = C_{r}$, as $C_{j}$ is not table-defined
  before the assignment.  
  After the assignment $\Next(\Next(C_{j-1})) \neq \bot$ which
  is impossible by Lemma~\ref{lem:basicsAboutGoodExecutions} (a).
  Suppose now we have a pair of assignments to $\VarP$, mapping $(x_0,x_1)$ to 
  $(x_{14},x_{15})$.  
  If $(x_{14},x_{15},14)$ is not part of the sequence connecting $C$ and $D$
  after the assignment, the same sequence shows equivalence before
  the assignment.  Otherwise, $\Next(\Next(x_{14},x_{15},14)) = \bot$
  by Lemma~\ref{lem:basicsAboutGoodExecutions} (a), as before.
\end{proof}

Next, we show that calls to $\forceVal$ \emph{also} do not change 
the equivalence relation for previously defined chains.  Also, they
never overwrite a previously defined value. However, 
we only show this under the assumption $x_{\ell-1} \notin \VarG_{\ell-1}$
and $x_{\ell+2} \notin \VarG_{\ell+2}$.  Later, we will see that this
assumption is safe.
\begin{lemma}\label{lem:noBadAdapt}
  Consider an execution of $\Syst_2(f,p)$ with a good pair $(f,p)$.
  Let $\ell \in \{4,10\}$ and suppose that for a call
  $\adapt(x_{\ell-2},x_{\ell-1},x_{\ell+2},x_{\ell+3},\ell)$ it holds
  that $x_{\ell-1} \notin \VarG_{\ell-1}$ and $x_{\ell+2}\notin
  \VarG_{\ell+2}$ before the call.

  Then, the following properties hold:
  \begin{enumerate}[(a)]
  \item For both calls $\forceVal(x,\cdot,j)$ we have $x \notin
    \VarG_j$ before the call.
  \item Let $C$ be a table-defined chain before the call to $\adapt$,
    $i \in \{1,\ldots,14\}$.  Then, $\val_i(C)$ stays constant during
    both calls to $\forceVal$.
  \item If the chains $C$ and $D$ are table-defined before the
    call to $\adapt$, then $C \equiv D$ before the calls to $\forceVal$
    if and only if $C \equiv D$ after the calls to $\forceVal$.
  \end{enumerate}
\end{lemma}
\begin{proof}
  Before $\adapt$ is called, $\evalFwd$ and $\evalBwd$ make sure that
  all the values $x_{\ell-1}, x_{\ell-2}, \ldots, x_0, \allowbreak
  x_{15},\ldots, x_{\ell+3}, x_{\ell+2}$ corresponding to 
  $(x_{\ell-2},x_{\ell-1},\ell-2)$ are
  defined in $\VarP$ and $\VarG$. By Lemma~\ref{lem:basicsAboutGoodExecutions} (b) and (d), 
  all partial chains defined by these values are equivalent
  to $(x_{\ell-2},x_{\ell-1},\ell-2)$.
  
  By our assumption, $x_{\ell-1} \notin \VarG_{\ell-1}$ and
  $x_{\ell+2}\notin \VarG_{\ell+2}$, and thus the procedure $\adapt$
  defines $\VarG_{\ell-1}(x_{\ell-1}):=f(\ell-1, x_{\ell-1})$ and
  $\VarG_{\ell+2}(x_{\ell+2}):=f(\ell+2, x_{\ell+2})$. These
  assignments lead to $x_{\ell} \notin \VarG_{\ell}$ and
  $x_{\ell+1} \notin \VarG_{\ell+1}$, as otherwise the event
  $\BadlyHit$ would occur.  This shows (a).
  
  We next show (b), i.e., for any $C$ the values $\val_i(C)$ stay
  constant.  For this, note first that this is true for table-defined chains $C$ that are
  equivalent to $(x_{\ell-2},x_{\ell-1},\ell-2)$ before the call to $\adapt$: 
  $\val_i$ gives exactly $x_i$ both before and after the calls to $\forceVal$.
  
  Now consider the table-defined chains that are not equivalent to 
  $(x_{\ell-2},x_{\ell-1},\ell-2)$ before the call to $\adapt$. 
  We show that for such a chain $C$, even $\val_i^+(C)$ and
  $\val_i^-(C)$ stay constant, as otherwise $\BadlyCollide$ would
  occur.  A value $\val^\sigma_{i}(C)$ can only change during the
  execution of $\forceVal(x_\ell,\cdot,\ell)$ if
  $\val^{\sigma}_\ell(C) = x_{\ell}$. But this implies that the
  assignment $\VarG(x_{\ell-1}) := f(\ell-1, x_{\ell-1})$ in $\adapt$
  made the two partial chains $C$ and $(x_{\ell-2},x_{\ell-1},\ell-2)$
  badly collide. For this, note that $C$ is table-defined even before the assignment, since
  it was table-defined before the call to $\adapt$. 
  Moreover, $(x_{\ell-2},x_{\ell-1},\ell-2)$ 
  is table-defined after the assignment. 
  The argument for $\forceVal(x_{\ell+1},\cdot,{\ell+1})$ is the same.  Thus, this
  establishes (b).

  We now show (c). First suppose that $C \equiv D$ before the calls
  to $\forceVal$. The sequence of chains given by Lemma~\ref{lem:structureOfEquivalentChains}
  is not changed during the calls to $\forceVal$, since by (a), no value is overwritten. Thus,
  the chains are still equivalent after the calls.  
  
  Now suppose that $C \equiv D$ after the calls to $\forceVal$. 
  Let $C_1,\ldots,C_r$ be the sequence given by
  Lemma~\ref{lem:structureOfEquivalentChains}.  If $C$ and $D$ were
  not equivalent before the calls to $\forceVal$, there is $i$
  such that before the call, $C_i$ was table defined, 
  but $C_{i+1}$ was not.
  Then, $\val^+(C_{i})$ changes during a
  call to $\forceVal$, contradicting the proof of (b). 
  Thus, the chains must have been
  equivalent before the calls.
\end{proof}

Equivalent chains are put into $\CompletedChains$ simultaneously:
\begin{lemma}\label{lem:equivalentChainsInCompletedChains}
  Suppose that $(f,p)$ is good.  Fix a point in the execution of
  $\Syst_2(f,p)$, and suppose that until this point, for no call to
  $\forceVal$ of the form $\forceVal(x,\cdot,\ell)$ we had $x \in
  \VarG_\ell$ before the call.  Suppose that at this point $C =
  (x_{k},x_{k+1},k)$ with $k \in \{1,7\}$ and $D = (y_m,y_{m+1},m)$
  with $m \in \{1,7\}$ are equivalent.  Then, $C \in \CompletedChains$
  if and only if $D \in \CompletedChains$.
\end{lemma}

\begin{proof}
  We may assume $k = 1$. We first show that the lemma holds right 
  after~$C$ was added to
  $\CompletedChains$. Since the chain was just adapted, and using
  Lemma~\ref{lem:basicsAboutGoodExecutions} (b) and (d), the only 
  chains which are equivalent to $C$ are those of the form
  $(\val_{i}(C),\val_{i+1}(C),i)$.
  Thus both $C$ and $D$ are added to $\CompletedChains$, and $D$ is the
  only chain with index $m=7$ that is equivalent to $C$.

  Now, the above property can only be lost if the event
  $\BadP$ occurs or else if a value is overwritten by
  $\forceVal$. Thus, we get the lemma.
\end{proof}

If the simulator detects a chain $(x_9,x_{10},9)$ for which
$\val^+$ is defined for sufficiently many values, a chain equivalent to it was
previously enqueued:
\begin{lemma}\label{lem:nonFreshChainsHaveEarlierOnes}
  Consider an execution of $\Syst_2(f,p)$ with a good pair $(f,p)$.
  Suppose that at some point, a chain $C = (x_7,x_8,7)$ 
  is enqueued for which $\val^+_2(C) \in \VarG_2$ or $\val^-_{13}(C) \in \VarG_{13}(C)$.
  Then, there is a chain equivalent to $C$ which was previously enqueued.
\end{lemma}
\begin{proof}
  We only consider the case
  $\val^+_2(C) \in \VarG_{2}$, the other case is symmetric.  Define
  $(x_0,x_1,x_2,x_{13},\allowbreak x_{14},x_{15}) :=
  (\val^+_{0}(C),\val^+_{1}(C),\val^+_2(C),\val^+_{13}(C),\val^+_{14}(C),
  \val^+_{15}(C))$.  All these must be different from $\bot$,
  since otherwise $\val^+_2(C) = \bot$.
  
  At some point in the execution, all the following entries are set in
  their respective hashtables: $\VarG_{1}(x_1), \VarG_{2}(x_{2}),
  \VarG_{13}(x_{13}), \VarG_{14}(x_{14})$, and
  $\VarP(\uparrow,x_{14},x_{15})$.  The last one of these must have
  been $\VarG_{2}(x_{2})$ or $\VarG_{13}(x_{13})$: if it was
  $\VarP(\uparrow,x_{14},x_{15})$, then the event $\BadP$ must
  have happened.  If it was $\VarG_{1}(x_{1})$, then the event
  $\BadlyHit$ must have happened (as $(x_0,x_1,0)$ is table-defined
  after the assignment).  Analogously, $\VarG_{14}(x_{14})$ cannot have
  been the last one.  Thus, since $\VarG_{2}(x_2)$ or
  $\VarG_{13}(x_{13})$ was defined last among those, the simulator
  will detect the chain and enqueue it.  
\end{proof}

If a chain $C$ is enqueued for which previously no equivalent 
chain has been enqueued, then the assumptions of Lemma~\ref{lem:noBadAdapt}
actually \emph{do} hold in good executions.  We first show that
they hold at the moment when the chains are enqueued (Lemma~\ref{lem:theGreenRoundsAreNew1}), and then
that they still hold when the chains are dequeued (Lemma~\ref{lem:theGreenRoundsAreNew2}).
\begin{lemma}\label{lem:theGreenRoundsAreNew1}
  Consider an execution of $\Syst_2(f,p)$ with a good pair $(f,p)$.
  Let $C$ be a partial chain which is enqueued in the execution at
  some time and to be adapted at position $\ell$.  Suppose that at the moment the chain is enqueued, no
  equivalent chain has been previously enqueued.

  Then, before the assignment $\VarG_{k}(x_k) := f(k,x_k)$ happens which 
  just preceds $C$ being 
  enqueued, $\val_{\ell-1}(C) = \bot$ and $\val_{\ell+2}(C) = \bot$.
\end{lemma}
\begin{proof}
  We have $\ell \in \{4,10\}$.  We will assume $\ell = 4$, 
  and due to symmetry of the
  construction, this also implies the lemma in case $\ell = 10$ for the
  corresponding rounds.

  The assignment sets either the value of $\VarG_7(x_7)$ or
  $\VarG_2(x_2)$ uniformly at random (otherwise, $\enqNewChains$ is
  not called in the simulator).  Consider first the case that 
  $\VarG_2(x_2)$ was just set.  Then, before this happened, 
  $\val_3^{+}(C)=\bot$, since $x_2 \notin \VarG_2$.  
  Furthermore, $\val_6^{-}(C) = \bot$, since
  otherwise, $\val^{-}_{7}(C) \in \VarG_{7}$, and then
  $(\val^{-}_{7}(C),\val^{-}_{8}(C),7)$ would be an equivalent,
  previously enqueued chain.  This implies the statement in case $\VarG_2(x_2)$ 
  is just set.
  The second case is if $\VarG_7(x_7)$ was
  just set.  Then, before the assignment, $\val_6^{-}(C) = \bot$, as
  $x_7 \notin \VarG_7$, and $\val^{+}_3(C) = \bot$, since otherwise
  $\val^+_2(C) \in \VarG_2$ and so an equivalent chain would have been
  previously enqueued, according to Lemma~\ref{lem:nonFreshChainsHaveEarlierOnes}.  
\end{proof}

\begin{lemma}\label{lem:theGreenRoundsAreNew2}
  Consider an execution of $\Syst_2(f,p)$ with a good pair $(f,p)$.
  Let $C$ be a partial chain which is enqueued in the execution at
  some time and to be adapted at position $\ell$.

  Then, at the moment $C$ is dequeued, it holds that $C \in
  \CompletedChains$, or that $(\val_{\ell-1}(C) \notin \VarG_{\ell-1})
  \land (\val_{\ell+2}(C) \notin \VarG_{\ell+2})$.
\end{lemma}
\begin{proof}
  Suppose that the lemma is wrong, and let $C$ be the first chain for
  which it fails. Because this is the first chain for which it fails,
  Lemma~\ref{lem:noBadAdapt}(a) implies that until the moment $C$ is
  dequeued, no call to $\forceVal$ overwrote a value. Now,
  consider the set $\mathfrak{C}$ of table-defined chains at some
  point in the execution that is not in an $\adapt$ call, 
	and before $C$ is dequeued.  Because of
  Lemmas~\ref{lem:equivalentChainsStayEquivalent}
  and~\ref{lem:noBadAdapt}(c), the equivalence relation among chains in
  $\mathfrak{C}$ stays constant from this point until the moment $C$
  is dequeued.

  We distinguish two cases to prove the lemma. Consider first the case
  that at the moment $C$ is enqueued, an equivalent chain $D$ was
  previously enqueued.  The point in the execution where $C$ is enqueued
  is clearly not
  in an $\adapt$ call, and both $C$ and $D$ are table-defined.
  Then, at the moment $C$ is dequeued, clearly
  $D \in \CompletedChains$. Thus, because of
  Lemma~\ref{lem:equivalentChainsInCompletedChains} and the remark
  about equivalence classes of $\mathfrak{C}$ above, this implies that
  $C \in \CompletedChains$ when it is dequeued.

  The second case is when $C$ has no equivalent chain which was
  previously enqueued.    To simplify notation we assume $\ell = 4$ and 
  show $\val_3(C) \notin \VarG_3$, but the argument is completely generic.
  From Lemma~\ref{lem:theGreenRoundsAreNew1} we get that before
  the assignment which led to $C$ being enqueued, $\val_{3}(C) = \bot$.
  If $\val_{3}(C) \in \VarG_3$ at the time $C$ is dequeued,
  it must be that $\VarG_3(\val_3(C))$ was set during
  completion of a chain $D$.  This chain $D$ was enqueued before $C$ was
  enqueued, and dequeued after $C$ was enqueued.  Also, at the 
  moment $C$ is dequeued, $\val_3(C) = \val_3(D)$. 
  From the point $C$ is enqueued, at any point until $C$ is dequeued, it is not possible that $C \equiv D$: 
  We assumed that there is no chain in the 
  queue that is equivalent to $C$ when $C$ is enqueued, and at the point $C$ is enqueued both 
  $C$ and $D$ are table-defined. Furthermore, this point 
  in the execution is not during an $\adapt$ call. Therefore,  
  by our initial remark, the equivalence relation between $C$ 
  and $D$ stays constant until the moment $C$ is dequeued.

  Consider the last assignment to a table before $\val_3(C) =
  \val_3(D) \neq \bot$ was true. We first argue that this assignment cannot have been of
  the form $\VarG_i(x_i) := f(i,x_i)$, as otherwise the event
  $\BadlyCollide$ would have happened. To see this, we check the conditions 
  for $\BadlyCollide$ for $C$ and $D$. The chain $D$ is table-defined even before the assignment, since
  it is in the queue. The assignment happens earliest
  right before $C$ is enqueued, in which case $C$ is table-defined after the assignment. 
  If the assignment happens later, $C$ is table-defined even before the assignment.     
  Furthermore, we have already seen that $C \equiv D$ is not possible.    
  Clearly, $\val_3(C) = \bot$ or $\val_3(D) = \bot$ before the assignment, and
  $\val_3(C) = \val_3(D) \neq \bot$ after the assignment. 
  
  The assignment cannot have been of the form $P(\downarrow,x_0,x_{1}) = (x_{14},x_{15})$ 
  or $P(\uparrow,x_{14},x_{15}) = (x_0,x_{1})$, since $\val$ can be evaluated at most one
  step further by Lemma~\ref{lem:basicsAboutGoodExecutions} (a).  
  Finally, the assignment cannot have been in a call to
  $\forceVal$, because of Lemma~\ref{lem:noBadAdapt}(b).

  Thus, $\val_3(C) \notin \VarG_3$ when $C$ is dequeued, and the same
  argument holds for the other cases as well.
\end{proof}

The following lemma is an important intermediate goal.  It states that
the simulator never overwrites a value in $\VarG$ in case $(f,p)$ is
good.

\begin{lemma}\label{lem:noSetFafterF}
  Consider an execution of $\Syst_2(f,p)$ with a good pair $(f,p)$.
  Then, for any call to $\forceVal$ of
  the form $\forceVal(x,\cdot,\ell)$ we have $x \notin \VarG_{\ell}$
  before the call.
\end{lemma}
\begin{proof}
  Assume otherwise, and let $C$ be the first chain during completion
  of which the lemma fails. Since the lemma fails for $C$, $C \notin 
  \CompletedChains$ when it is dequeued. Thus, Lemma~\ref{lem:theGreenRoundsAreNew2} 
  implies that
  $\val_{\ell-1}(C) \notin \VarG_{\ell-1}$ and $\val_{\ell+2}(C)
  \notin \VarG_{\ell+2}$ when $C$ is dequeued, 
  and so by Lemma~\ref{lem:noBadAdapt}(a) we get the result.
\end{proof}

We say that a \emph{distinguisher completes all chains}, if, at the
end of the execution, it emulates a call to $\evalFwd(x_0,x_1,0,14)$ for
all queries to $\P(x_0,x_1)$ or to $(x_0,x_1) =
\P^{-1}(x_{14},x_{15})$ which it made during the execution.

\begin{lemma}\label{lem:atEndFeistelGivesTheRightThing}
  Consider an execution of $\Syst_2(f,p)$ with a good pair $(f,p)$ in
  which the distinguisher completes all chains.  Suppose that during
  the execution $P(\downarrow, x_0,x_1)$ is queried.  Then, at the end
  of the execution it holds that $P(\downarrow, x_0,x_1) =
  \bigl(\val^+_{14}(x_0,x_1,0),\val^+_{15}(x_0,x_1,0)\bigr)$,
  and $P(\uparrow,x_{14},x_{15}) = \bigl(\val^{-}_0(x_{14},x_{15},14),
  \val^{-}_{1}(x_{14},x_{15},14)\bigr)$.
\end{lemma}
\begin{proof}
  If the query $P(\downarrow, x_0,x_1)$ was made by the simulator,
  then this was while it was completing a chain.  Then, right after it
  finished adapting we clearly have the result.  By
  Lemma~\ref{lem:noSetFafterF} no value is ever overwritten.  Since
  the event $\BadP$ does not occur, the conclusion of the lemma must
  also be true at the end of the execution.

  Consider the case that $P(\downarrow, x_0,x_1)$ was a query by the
  distinguisher.  Since it eventually issues the corresponding Feistel
  queries, it must query the corresponding values $x_7$ and $x_8$
  at some point.  Thus, $x_7 \in \VarG_{7}$ and $x_{8} \in
  \VarG_{8}$ at the end of the execution.  One of the two values was
  defined later, and in that moment, $(x_7,x_{8},7)$ was enqueued by
  the simulator.  Thus, it is dequeued at some point.  If it was not
  in $\CompletedChains$ at this point, it is now completed and the
  conclusion of the lemma holds right after this completion.
  Otherwise, it was completed before it was inserted in
  $\CompletedChains$, and the conclusion of the lemma holds after this
  completion.  Again, by Lemma~\ref{lem:noSetFafterF} no value is ever
  overwritten, and again $\BadP$ never occurs, hence the
  conclusion also holds at the end of the execution.
\end{proof}

\begin{lemma}\label{lem:nAdaptEqualsNofPqueries}
  Consider an execution of $\Syst_2(f,p)$ with a good pair $(f,p)$ in
  which the distinguisher completes all chains.  Then, the number of
  calls to $\adapt$ by the simulator equals the number of queries to
  $p(\cdot,\cdot,\cdot)$ made by the two-sided random function.
\end{lemma}
\begin{proof}
  Since the event $\BadP$ does not occur, the number of
  queries to $p(\cdot,\cdot,\cdot)$ equals half the number of entries
  in $\VarP$ at the end of the execution.

  For each call to $\adapt$, there is a corresponding pair of entries
  in $\VarP$: just before $\adapt$ was called, such an entry was read
  either in $\evalFwd$ or $\evalBwd$.  Furthermore, for no other call
  to $\adapt$ the same entry was read, as otherwise a value would have
  to be overwritten, contradicting Lemma~\ref{lem:noSetFafterF}.

  For each query to $p(\cdot,\cdot,\cdot)$, there was a corresponding
  call to $\adapt$:  if the query to $p$ occurred in a call to $P$
  by the simulator, then we consider the call to
  $\adapt$ just following this call (as the simulator only queries $P$
  right before it adapts).  If the query to $p$ occurred
  in a call by the distinguisher, the distinguisher eventually
  queries the corresponding Feistel chain.  At the moment it queries
  $\VarG_{8}(x_8)$, we find the first chain which is equivalent to
  $(x_7,x_8,7)$ at this point and was enqueued.  This chain must have
  been adapted accordingly.
\end{proof}

\subsubsection{Mapping randomness of $\Syst_2$ to randomness of $\Syst_3$}
\label{sec:theMappingTau}

We next define a map $\tau$ which maps a pair of tables $(f,p)$ to a
partial table $h$, where a partial table $h: \{1,\ldots,14\} \times
\{0,1\}^{n} \mapsto \{0,1\}^n \cup \{\bot\}$ either has an actual
entry for a pair $(i,x)$, or a symbol $\bot$ which signals that the
entry is unused.  This map will be such that $\Syst_2(f,p)$ and
$\Syst_3(\tau(f,p))$ have ``exactly the same behaviour''.

\begin{definition}
  The function $h = \tau(f,p)$ is defined as follows: Run a simulation
  of $\Syst_2(f,p)$ in which the distinguisher completes all chains.
  If $f(i,x)$ is read at some point, then $h(i,x) := f(i,x)$.  If
  $f(i,x)$ is never read, but for some $y$ a call $\forceVal(i,x,y)$
  occurs, then $h(i,x) := y$ for the first such call.  If $f(i,x)$ is
  never read and no such call to $\forceVal$ occurs, then $h(i,x) :=
  \bot$.
\end{definition}

\begin{lemma}\label{lem:sameBehaviour}
  Suppose $h$ has a good preimage. Consider any execution of $S_3(h)$ and suppose the distinguisher completes all chains. Then, $\Syst_3(h)$ never queries
  $h$ on an index $(i,x)$ for which $h(i,x) = \bot$.  Furthermore, the
  following two conditions on $(f,p)$ are equivalent:
  \begin{enumerate}[(1)]
  \item The pair $(f,p)$ is good and $\tau(f,p)=h$.
  \item The queries and answers to the two-sided random function
    in $\Syst_2(f,p)$ are exactly the same as the queries and answers
    to the Feistel construction in $\Syst_3(h)$; and $h(i,x) = f(i,x)$
    for any query $(i,x)$ issued to $f$ or $h$ by the simulator.
  \end{enumerate}
\end{lemma}
\begin{proof}
  We first show that $(1)$ implies $(2)$. 
  Thus, because the distinguisher
  is deterministic, we need to show the following:
  \begin{itemize}
  \item When the simulator sets $\VarG_{i}(x_i) := f(i,x_i)$ in $\Syst_2(f,p)$,
    respectively $\VarG_{i}(x_i) := h(i,x_i)$ in $\Syst_3(h)$, the two values are
    the same.
  \item When the simulator queries $\P(x_0,x_1)$ or
    $\P^{-1}(x_{14},x_{15})$ it gets the same answer in
    $\Syst_2(f,p)$ and $\Syst_3(h)$.
  \end{itemize}
  
  The first bullet is obvious, because if the simulator ever sets
  $\VarG_{i}(x_i) := f(i,x_i)$ in $\Syst_2(f,p)$, then~$h$ will be set
  accordingly by definition of $\tau$.

  Thus, we consider a query to $\P(x_0,x_1)$ (queries to $\P^{-1}$
  are handled in the same way).  Recall that we assume that the
  distinguisher completes all chains.  Because of
  Lemma~\ref{lem:atEndFeistelGivesTheRightThing}, the answer of the
  query to $\P$ is exactly what we obtain by evaluating the Feistel
  construction at the end in experiment $\Syst_2$.  But each query in
  the evaluation of the Feistel construction was either set as
  $\VarG_{i}(x_i) := f(i,x_i)$ or in a $\forceVal$ call, and in both
  cases the values of $h$ must agree, since in good executions no
  value is ever overwritten (Lemma~\ref{lem:noSetFafterF}).  Thus, the
  query to $\P$ is answered by the Feistel in the same way.

  We now show that (2) implies (1). Assume now that $(2)$ holds.  
  Let $(f_h,p_h)$ be a good
  preimage of $h$, i.e., a pair satisfying $(1)$.  We know already
  that condition $(2)$ holds for $(f_h,p_h)$, and because we assume
  that it holds for $(f,p)$, we see that in the two executions
  $\Syst_2(f_h,p_h)$ and $\Syst_2(f,p)$ all queries to the two-sided
  random function are the same, and also the entries $f(i,x)$ and $f_h(i,x)$ 
  for values considered match.  This implies that $(f,p)$ must be good.
  Furthermore, this implies $\tau(f,p) = \tau(f_h,p_h)$.

  Finally, we argue that $\Syst_3(h)$ never queries $h$ on an index
  $(i,x)$ for which $h(i,x) = \bot$.  
  Let $(f_h,p_h)$ be a good preimage of $h$.
  Clearly (1) holds for $h$ and $(f_h,p_h)$, which implies (2) as shown above. 
  Thus, it cannot be that a query to $h$ in $\Syst_3(h)$ returns $\bot$, as otherwise 
  the answers in $\Syst_2(f_h, p_h)$ and $\Syst_3(h)$ would differ.
\end{proof}

\begin{lemma}\label{lem:goodHAreUniform}
  Suppose $h$ has a good preimage.  Pick $(f,p)$ uniformly at
  random.  Then,
  \begin{align}
    \Pr[\text{$(f,p)$ is good} \land \tau(f,p) = h]
    = 2^{-n |h|},
  \end{align}
  where $|h|$ is the number of pairs $(i,x)$ for which $h(i,x) \neq \bot$.
\end{lemma}
\begin{proof}
  Let $(f_h,p_h)$ be a good preimage of $h$.  With probability $2^{-n
    |h|}$ all queries in $\Syst_2(f,p)$ are answered exactly as those
  in $\Syst_2(f_h,p_h)$: every query to $f$ is answered the same with
  probability $2^{-n}$, and every query to $p$ with probability
  $2^{-2n}$.  Because of Lemma~\ref{lem:nAdaptEqualsNofPqueries} the
  number $|h|$ of non-nil entries in $h$ is exactly the number of
  queries to $f$ plus twice the number of queries to $p$.
\end{proof}

\begin{replemma}{lem:eqS2andS3}
  The probability that a fixed distinguisher answers $1$ in
  $\Syst_2(f,p)$ for uniform random $(f,p)$ differs at most by
  $\frac{8\cdot 10^{19} \cdot q^{10}}{2^n}$ from the probability that it answers $1$ in
  $\Syst_3(h)$ for uniform random $h$.
\end{replemma}
\begin{proof}
  First, modify the distinguisher such that for each query to
  $\P(x_0,x_1)$ or to $(x_0, x_1) = \P^{-1}(x_{14},x_{15})$ which it
  made during the execution (to either the two-sided random function 
  in $\Syst_2$ or the Feistel construction in $\Syst_3$), it 
  issues the corresponding Feistel queries to $\F$ in the end
  (i.e., it emulates a call to $\evalFwd(x_0,x_1,0,14)$).
  This increases the number of queries of the distinguisher by at most 
  a factor of $14$. Furthermore, any unmodified distinguisher that achieves
  some advantage will achieve the same advantage when it is 
  modified.
  
  Consider now the following distribution over values $h^*$, which
  are either tables for $\Syst_3(h^*)$ which contain no entry $\bot$,
  or special symbols $\bot$.  To pick an element $h^*$, we pick a pair
  $(f,p)$ uniformly at random.  If $(f,p)$ is good, we compute $h :=
  \tau(f,p)$ and set each entry of $h$ with $h(i,x) = \bot$ uniformly
  at random.  The result is $h^*$.  If $(f,p)$ is not good, we set
  $h^* = \bot$. Let $H$ be the random variable that takes values according
  to this distribution.
  
  We now claim that the probability that any fixed table 
  $h^* \neq \bot$ is output is at most $2^{-n|h^*|}$.
  To prove this, we first show that it cannot be that two different values
  $h$ which both have a good preimage can yield the same $h^*$. Towards a contradiction
  assume that $h$ and $h'$ are different and both have a good preimage, and they yield the
  same $h^*$. Let $(f_h, p_h)$ and $(f_{h'}, p_{h'})$ be good preimages of $h$ and $h'$,
  respectively. Then, Lemma~\ref{lem:sameBehaviour} item (2) implies that the queries
  and answers in $\Syst_2(f_h,p_h)$ and $\Syst_3(h)$ are the same. Furthermore, since
  $\Syst_3(h)$ never queries $h$ on an index $(i,x)$ where $h(i,x) = \bot$ 
  (Lemma~\ref{lem:sameBehaviour}), we get that the queries and answers in $\Syst_3(h)$ 
  and $\Syst_3(h^\ast)$ are the same.  Arguing symmetrically for 
  $(f_{h'}, p_{h'})$, we see that the queries and answers in $\Syst_3(h')$ and
  $\Syst_3(h^\ast)$ are the same, and so the queries and answers in $\Syst_2(f_h,p_h)$ and 
  $\Syst_2(f_{h'}, p_{h'})$ must be the same. But by definition of $\tau$, this implies that $h = h'$, 
  a contradiction.
      
  We now calculate the probability of getting a fixed table $h^* \neq \bot$. In the first case, suppose there exists $h$ with a
  good preimage that can lead to $h^*$. Let $\rho$ be the randomness that is used to replace the $\bot$ entries
  in $h$ by random entries. We have
  \begin{align}
  	\Pr_{(f,p),\rho}[H = h^*] &= \Pr_{(f,p),\rho}[(f,p) \text{ is good } \land h=\tau(f,p)  \text{ can lead to } h^* \land \text{ filling with $\rho$ leads to }h^*]. \nonumber	
  \end{align}
  Now, as we have seen, no two different values for $h$ can yield the same $h^*$. Thus, we can assume that $h^* = (h,\rho^*)$, where $h$ is the unique
  table that leads to $h^*$, and $\rho^*$ stands for the entries that occur in $h^*$, but are $\bot$ in $h$. Then, the above probability equals
  \begin{align}
  	  &\Pr_{(f,p),\rho}[(f,p) \text{ is good } \land \tau(f,p) = h \land \rho = \rho^*] \nonumber \\
  	= &\Pr_{(f,p),\rho}[(f,p) \text{ is good } \land \tau(f,p) = h ] \cdot \Pr_{(f,p),\rho}[\rho = \rho^*] \nonumber \\				
  	= &2^{-n|h|} \cdot 2^{-n(|h^*|-|h|)} = 2^{-n|h^*|}, \nonumber
  \end{align}
  where for the second equality we apply Lemma~\ref{lem:goodHAreUniform} and note that $\rho$ is chosen uniformly.
  
  In the second case, there exists no $h$ with a good preimage that can lead to $h^*$. Then we have $\Pr_{(f,p),\rho}[H = h^*] = 0$, and so in both cases
  \begin{align}
  	 \Pr_{(f,p),\rho}[H = h^*] \leq 2^{-n|h^*|} \label{eq:lem34.1}
  \end{align}
  
  This implies that the statistical distance of the distribution over
  $h^*$ which we described to the uniform distribution is exactly the
  probability that $(f,p)$ is not good. For completeness, we give a formal
  argument for this. Consider $H$ as above, and let $U$ be a random variable
  taking uniform random values from $\{0,1\}^{|h^*|}$. We have
  \begin{align}	
    d(U,H) &= \frac12 \sum_{h*}\bigl|\Pr[U=h^*]-\Pr_{(f,p),\rho}[H=h^*]\bigr|\nonumber \\
    			 &= \frac12 \bigl|\underbrace{\Pr[U=\bot]}_{=0} - \underbrace{\Pr_{(f,p),\rho}[H=\bot]}_{=\Pr_{(f,p)}[(f,p) \text{ is not good}] }\bigr| + 
    			 		\frac12 \sum_{h^* \neq \bot} \bigl|\Pr[U=h^*]-\Pr_{(f,p),\rho}[H=h^*]\bigr| \nonumber\\
    			 &= \frac12 \Pr_{(f,p)}((f,p) \text{ is not good}) 
    			 						+ \frac12 \underbrace{\sum_{h^* \neq \bot} \Pr[U=h^*]}_{=1}
    			 						- \frac12 \underbrace{\sum_{h^* \neq \bot}\Pr_{(f,p),\rho}[H=h^*]}_{=1-\Pr_{(f,p)}[(f,p) \text{ is not good}]} \nonumber\\
    			 &= \Pr_{(f,p)}[(f,p) \text{ is not good}]\;,\nonumber
  \end{align}
  where the third equality uses (\ref{eq:lem34.1}).
  
  We proceed to argue that $\Pr_{(f,p)}[(f,p) \text{ is not good}]$ is small. 
  In $\Syst_2(f,p)$, 
  by Lemma~\ref{lem:SimulatorIsEfficient} we have that $|\VarG_i| \leq 6\cdot(14\cdot q)^2$ and $|\VarP| \leq 6\cdot(14\cdot q)^2$, where
  the additional factor of $14$ comes in because the distinguisher completes all chains. 
  By Lemma~\ref{lem:goodPairsAreLikely}
  $\Pr_{(f,p)}[(f,p) \text{ is not good}] \leq 16\,000\cdot\frac{(6\cdot(14\cdot q)^2)^5}{2^n} < \frac{4\cdot 10^{19} \cdot q^{10}}{2^n}$.
  
  By Lemma~\ref{lem:sameBehaviour}, for good $(f,p)$, the behaviour of $\Syst_2(f,p)$ and $\Syst_3(H)$ is identical. Thus, $
  \bigl| \Pr_{(f,p)}[\text{$D$ outputs $1$ in $\Syst_2(f,p)$}] - \Pr_{(f,p)}[\text{$D$ outputs $1$ in $\Syst_3(H)$}] \bigr| \leq \Pr_{(f,p)}[(f,p) \text{ is not good}]$.
  Furthermore, 
  $$\bigl| \Pr_{(f,p)}[\text{$D$ outputs $1$ in $\Syst_2(H)$}] - \Pr[\text{$D$ outputs $1$ in $\Syst_3(U)$}] \bigr| \leq d(H,U) = \Pr_{(f,p)}[(f,p) \text{ is not good}],$$ and therefore
  \begin{align*}
  	  \bigl| \Pr_{(f,p)}[\text{$D$ outputs $1$ in $\Syst_2(f,p)$}] - \Pr[\text{$D$ outputs $1$ in $\Syst_3(U)$}] \bigr| &\leq 2\cdot \Pr_{(f,p)}[(f,p) \text{ is not good}]
\\& < \frac{8\cdot 10^{19} \cdot q^{10}}{2^n} ,
  \end{align*}
  using our bound on the probability that $(f,p)$ is good above.
\end{proof}

\section*{Acknowledgements}
It is a pleasure to thank Ueli Maurer and Yannick Seurin for their insightful feedback.

 Robin K\"unzler was partially
supported by the Swiss National Science Foundation (SNF), project  no.~200021-132508. Stefano Tessaro was partially supported by NSF grant CNS-0716790; part of this work was done while he was a graduate student at ETH Zurich.

\bibliographystyle{alpha}
\bibliography{MABibliography,prf}

\appendix


\section{Detailed Definition of the Simulator of Coron et al.}
\label{appendixSimulatorDefinition}

We proceed to provide the full definition of the simulator in
\cite{CPS08v2}. In particular, for ease of reference, we stick to the
same variable naming used within their work, even though this is
inconsistent with the notation used in the rest of this paper.

As sketched in the simulator overview, $\si$ keeps a {\em history} of
the function values it has defined up to the current point of
execution, which consist of sets $\hist(\URF_i)$ of pairs $(x,
\URF_i(x))$.  With some abuse of notation, and to improve readability,
we denote as $\URF_i := \{x | (x,\URF_i(x)) \in \hist(\URF_i) \}$ the
set of $\URF_i$-queries whose values have been defined. We define the
history $\hist := \{(x, \URF_i(x), i) | (x, \URF_i(x)) \in
\hist(\URF_i) \text{ for some } i \}$.

Also, if the history size gets too big, i.e.~$|\hist(\URF_i)| > \hmax$
for some $\URF_i$ and a value $\hmax$ depending only on the number of
distinguisher queries $q$, the simulator aborts.

\paragraph{Procedures $\Query$ and $\ChainQuery$.}
Upon a query $x$ for $\URF_k$ issued by $\di$, the simulator $\sitwo$
executes the following prodecure $\Query(x, k)$:
\begin{lstlisting}[language=pseudocode,name=cpssim]
procedure $\Query(x,k)$:
   if $x \in \URF_k$ then 
      return $\URF_k(x)$ from $\hist(\URF_k)$ 
   else
      $\URF_k(x) \inu \{0,1\}^n$
      $\ChainQuery(x,k)$
      return $\URF_k(x)$ from $\hist(\URF_k)$
\end{lstlisting}
Procedure $\ChainQuery$ checks if 3-chains as defined in
Fig.~\ref{figChainsV1} w.r.t.~query $x$ occur. Note that this is only
a selection of all possible chains one might consider. Observe that
the chain sets for $\URF_1, \URF_2$ and $\URF_3$ are defined
symmetrically to those of $\URF_4, \URF_5$ and $\URF_6$.
\begin{figure}[ht]
\centering
                \begin{tabular}{|c|l|}
                                \hline Query to & Chain Sets \\
                                \hline $\URF_1$ & $\Chain(-,R,1) = \bigl\{ (S,A) \in (\URF_6,\URF_5) \,|\, \URP^{-1}(S || A \oplus \URF_6(S))\rightp = R \bigr\}$  \\
                                \hline $\URF_2$ & $\Chain(+,X,2) = \bigl\{ (Y,Z) \in (\URF_3,\URF_4) \,|\, X = \URF_3(Y) \oplus Z \bigr\}$  \\
                                                                                          & $\Chain(-,X,2) = \bigl\{ (R,S) \in (\URF_1,\URF^\ast_6(X)) \,|\, \URP(X \oplus \URF_1(R) || R)\leftp = S \bigr\}$  \\
                                \hline $\URF_3$ & $\Chain(+,Y,3) = \bigl\{ (Z,A) \in (\URF_4,\URF_5) \,|\, Y = \URF_4(Z) \oplus A \bigr\}$  \\
                                \hline $\URF_4$ & $\Chain(-,Z,4) = \bigl\{ (Y,X) \in (\URF_3,\URF_2) \,|\, Z = \URF_3(Y) \oplus X \bigr\}$  \\
                                \hline $\URF_5$ & $\Chain(+,A,5) = \bigl\{ (S,R) \in (\URF_6,\URF^\ast_1(A)) \,|\, \URP^{-1}(S || A \oplus \URF_6(S))\rightp = R \bigr\}$ \\
                                                                                                & $\Chain(-,A,5) = \bigl\{ (Z,Y) \in (\URF_4,\URF_3) \,|\, A = \URF_4(Z) \oplus Y \bigr\}$  \\
                                \hline $\URF_6$ & $\Chain(+,S,6) = \bigl\{ (R,X) \in (\URF_1,\URF_2) \,|\, \URP(X \oplus \URF_1(R) || R)\leftp = S \bigr\}$  \\                                                          
                                \hline
                \end{tabular}
                \caption{Chains to be considered by Procedure
                  $\ChainQuery$.}
                \label{figChainsV1}
\end{figure}
In particular, the sets $\URF^\ast_6$ and $\URF^\ast_1$ are defined as
follows for the understood parameters $X$ and $A$:
\begin{align}
                \URF^\ast_6(X) &:= \URF_6 \cup \bigl\{S | \exists (R',X') \in (\URF_1, \URF_2 \setminus \{X\}), \URP(X' \oplus \URF_1(R')||R')\leftp = S\bigr\}  \nonumber \\
                \URF^\ast_1(A) &:= \URF_1 \cup \bigl\{R | \exists (S',A') \in (\URF_6, \URF_5 \setminus \{A\}), \URP^{-1}(S'||A'  \oplus \URF_6(S'))\rightp = R\bigr\}  \nonumber
\end{align}
The chains that are considered additionally when using the sets
$\URF_i^\ast$ instead of $\URF_i$ in the cases of $(\URF_2,-)$ and
$(\URF_5,+)$ chains are called \textit{virtual chains}. These chains
do not consist of history values exclusively. If such a chain occurs,
the simulator first defines the three values that constitute the
$3$-chain that are not yet in the history and then completes it. The
intuition\footnote{This intuition comes from studying the proof in
  \cite{CPS08v1}} why such virtual chains are considered in addition
is as follows: Every time the sets $\Chain(-,R,1)$ for some $R$ and
symmetrically $\Chain(+,S,6)$ for some $S$ are computed, it will not
be possible that more than one such chain is found. Note that it is
not clear at this point if this goal can indeed be
achieved. Intuitively, such a fact might simplify (or even make
possible) the analysis of the recursions of the simulator.

The procedure $\ChainQuery$ now handles the recursions. We still
consider the chains in Fig.~\ref{figChainsV1}. If such 3-chains occur,
$\ChainQuery$ calls a further procedure, called $\CompleteChain$, to
complete these chains, i.e.~it consistently (with $\URP)$ fills in the
remaining values for each 3-chain. Then, $\ChainQuery$ is called
recursively for the values defined during the completions of chains:
\begin{lstlisting}[language=pseudocode,name=cpssim]
procedure $\ChainQuery(x,k)$:
   if $k \in \{ 1,2,5,6\}$ then $\XorQuery_1(x,k)$  
   if $k \in \{ 1,3,4,6\}$ then $\XorQuery_2(x,k)$  
   if $k \in \{ 3,4\}$ then $\XorQuery_3(x,k)$  
   $\mathcal{U} := \emptyset$ 
   if $k \in \{2,3,5,6\}$ then 
      forall $(y,z) \in \Chain(+,x,k)$ do
         $\mathcal{U}:= \mathcal{U} \cup \CompleteChain(+, x, y, z, k)$ 
   if $k \in \{1,2,4,5\}$ then 
      forall $(y,z) \in \Chain(-,x,k)$ do
         $\mathcal{U}:= \mathcal{U} \cup \CompleteChain(-, x, y, z, k)$ 
   forall $(x',k') \in \mathcal{U}$ do
      $\ChainQuery(x',k')$ 
\end{lstlisting}
The first three lines of $\ChainQuery$ make calls to the so-called
$\XorQuery$ procedures: they perform additional $\ChainQuery(x',k')$
executions for values $x'$ other than $x$ that fullfil certain
properties. This ensures that for the values $x'$, we are always sure
that $\ChainQuery(x',k')$ occurs \textit{before} the chains for
$\ChainQuery(x,k)$ are completed. To understand $\sitwo$ one can
ignore, for the moment, the details about the $\XorQuery$ procedures:
We first detail $\CompleteChain$ and only subsequently address the
$\XorQuery$ procedures.

\paragraph{Procedure $\CompleteChain$.}
Procedure $\CompleteChain$ completes a chain $(x, y, z)$, given
direction $d$ and the index $k$ of $\URF_k$ according to the following
table.
\begin{figure}[h]
\centering
                \begin{tabular}{|c|c|c|c|c|c|}
                                \hline Query $x$ to $\URF_k$ & Dir $d$ & $(y, z)$ in History    & Additionally                                                  & Compute & Adapt $(\URF_j, \URF_{j+1})$\\
                                                                                                                                                 &                               &                                                                                      &                                       set $\URF_{i}$  &                         &     \\
                                \hline $\URF_1$                                                  & $-$           & $(\URF_6, \URF_5)$                   & $\URF_4$                                                                      & $S||T$        & $(\URF_2, \URF_3)$ \\
                                \hline $\URF_2$                                                  & $+$           & $(\URF_3, \URF_4)$                   & $\URF_1$                                                                      & $L||R$        & $(\URF_5, \URF_6)$ \\
                                                                                                                                                 & $-$           & $(\URF_1, \URF_6)$                   & $\URF_3$                                                                      & $L||R$        & $(\URF_4, \URF_5)$ \\
                                \hline $\URF_3$                                                  & $+$           & $(\URF_4, \URF_5)$                   & $\URF_6$                                                                      & $S||T$        & $(\URF_1, \URF_2)$ \\
                                \hline $\URF_4$                                                  & $-$           & $(\URF_3, \URF_2)$                   & $\URF_1$                                                                      & $L||R$        & $(\URF_5, \URF_6)$ \\
                                \hline $\URF_5$                                                  & $+$           & $(\URF_6, \URF_1)$                   & $\URF_4$                                                                      & $S||T$        & $(\URF_2, \URF_3)$ \\
                                                                                                                                         & $-$           & $(\URF_4, \URF_3)$                   & $\URF_6$                                                                      & $S||T$        & $(\URF_1, \URF_2)$ \\
                                \hline $\URF_6$                                                  & $+$           & $(\URF_1, \URF_2)$                   & $\URF_3$                                                                      & $L||R$        & $(\URF_4, \URF_5)$ \\
                                \hline
                \end{tabular}
\end{figure} \\ 
In detail this looks as follows:
\begin{lstlisting}[language=pseudocode,name=cpssim]
procedure $\CompleteChain(d,x,y,z,k)$:
   $\mathcal{U} := \emptyset$ 
   if there was a $\CompleteChain$-execution w.r.t.~values $x,y,z$ before then 
      return $\mathcal{U}$ 
   else 
      if $(d,k) = (-,2)$ and $z\notin \URF_6$ then 
         $\URF_6(z) \inu \{0,1\}^n$ 
         $\mathcal{U} := \mathcal{U} \cup \{(z, 6)\}$ 
      if $(d,k) = (+,5)$ and $z\notin \URF_1$ then 
         $\URF_1(z) \inu \{0,1\}^n$ 
         $\mathcal{U} := \mathcal{U} \cup \{(z, 1)\}$ 
      compute $x_i$ ($i$ according to the table, additional 'set')
      if $x_i \notin \URF_i$ then
         $\URF_i(x_i) \inu \{0,1\}^n$
         $\mathcal{U} := \mathcal{U} \cup \{(x_i, i)\}$
      if Compute $L||R$ (according to the table) then 
         compute $L||R$
         compute $S||T := \URP(L||R)$
      if Compute $S||T$ (according to the table) then
         compute $S||T$
         compute $L||R := \URP^{-1}(S||T)$
      now all inputs $(x_1, x_2, \ldots, x_6)$ to $(\URF_1, \URF_2, \ldots, \URF_6)$ of the completion of $(x,y,z)$ are known
      compute $x_0 := L$, $x_7 := T$ 
      if $x_j \in \URF_j$ or $x_{j+1} \in \URF_{j+1}$ then 
         abort 
      else (adapt according to the table)
         $\URF_j(x_j) := x_{j-1} \oplus x_{j+1}$ 
         $\URF_{j+1}(x_{j+1}) := x_{j} \oplus x_{j+2}$ 
         $\mathcal{U} := \mathcal{U} \cup \{(x_j, j),(x_{j+1}, j+1)\}$
      return $\mathcal{U}$ 
\end{lstlisting}

Lines $2$ and $3$ need further explanation. We assume that the
simulator keeps track of the $6$-tuples $(R,X,Y,Z,A,S)$ of values that
were defined in any $\CompleteChain$-execution up to the current point
of execution. In line $2$, $\CompleteChain$ checks whether the values
$x,y,z$ for given $k,d$ are part of such a $6$-tuple of values that
were defined in some earlier $\CompleteChain$-execution. If so, the
empty set is returned. For example, if in $\ChainQuery(R',1)$ we find
$(S',A') \in \Chain(-,R',1)$, then $\CompleteChain(-,R', S', A', 1)$
occurs and line $2$ checks if there is a $6$-tuple $(R,X,Y,Z,A,S)$
from an earlier $\CompleteChain$-execution where $R=R', S=S'$ and
$A=A'$. If so, the chain $(S',A')$ is not completed again and the
empty set is returned in line $3$ immediately.
Note that steps $2$ and $3$ were not included in the simulator
definition in \cite{CPS08v2}. But if these steps are not performed,
the simulator trivially aborts as soon as a recursive $\CompleteChain$
call occurs in $\ChainQuery$. (In the above example of $(S',A')$, if
the values $(R',X,Y,Z,A',S')$ were defined in an earlier
$\CompleteChain$-execution, the simulator would abort at step $29$ of
$\CompleteChain$, since both $Y \in \URF_3$ and $X \in
\URF_2$ already.)  Furthermore, note that lines $5$ to $11$ are used
to define missing function values for virtual chains.

Fig.~\ref{figCCSv2} illustrates how the simulator $\sitwo$ completes $3$-chains. 
\begin{figure}[!t]
                \begin{center}
                        \includegraphics[scale=0.37]{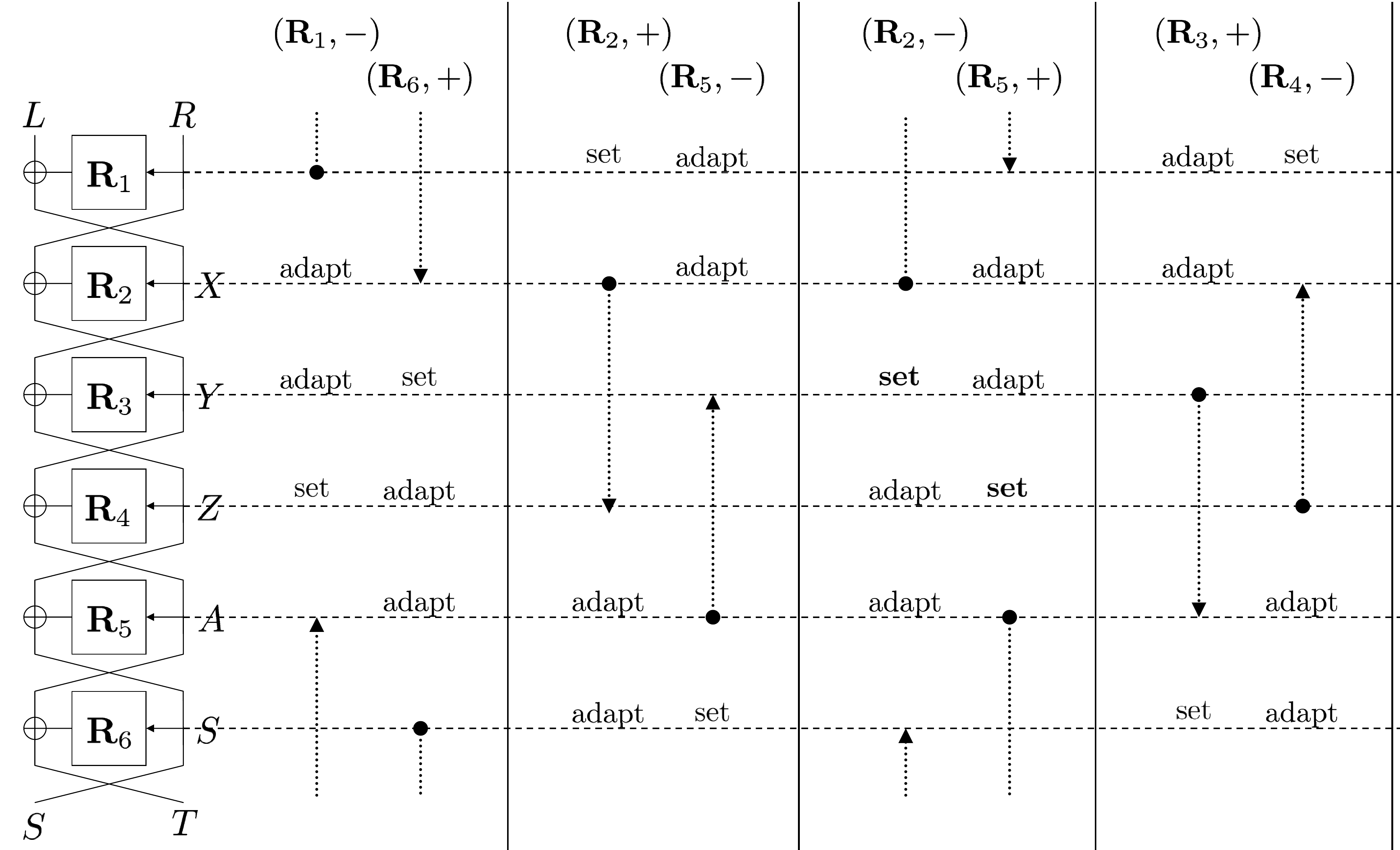}
                \end{center}
                \caption{An illustration of how $\sitwo$ completes $3$-chains}
                \label{figCCSv2}
\end{figure}

\paragraph{The $\XorQuery$ procedures.}
As mentioned earlier, the procedures $\XorQuery_1$, $\XorQuery_2$, and
$\XorQuery_3$ perform additional calls to $\ChainQuery$
\textit{before} the chains for $\ChainQuery(x,k)$ are completed.

The idea behind $\XorQuery_1$ is the following: We consider the
execution of $\ChainQuery(X,2)$ for some $X$ and in this execution a
recursive call to $\ChainQuery(Y,3)$ that occurs for some $Y$. For
$\ChainQuery(Y,3)$, $\XorQuery_1$ should ensure that for any $(\URF_3,
+)$ chain $(Y, Z, A)$, $X:= \URF_3(Y) \oplus Z$ is not in
$\URF_2$. The same property should be ensured symmetrically for
$\ChainQuery(A,5)$. 
\begin{lstlisting}[language=pseudocode,name=cpssim]
procedure $\XorQuery_1(x,k)$: 
   if $k=5$ then   
      $\mathcal{A}' := \{x \oplus R_1 \oplus R_2 \notin \URF_5 | R_1, R_2 \in \URF_1, R_1 \neq R_2 \}$ 
   else if $k=1$ then   
      $\mathcal{A}' := \{A \oplus x \oplus R_2 \notin \URF_5 | A \in \URF_5, R_2 \in \URF_1\}$ 
   if $k=5$ or $k=1$ then 
      forall $A' \in \mathcal{A}'$ do
        if $\exists R' \in \URF_1, \exists S' \in \URF_6 : \URP^{-1}(S'|| \URF_6(S') \oplus A')\rightp = R'$ then 
            $\URF_5(A') \inu \{0,1\}^n$ 
            $\ChainQuery(A',5)$ 
   if $k=2$ then   
      $\mathcal{X}' := \{x \oplus S_1 \oplus S_2 \notin \URF_2 | S_1, S_2 \in \URF_6, S_1 \neq S_2 \}$ 
   else if $k=6$ then   
      $\mathcal{X}' := \{X \oplus x \oplus S_2 \notin \URF_2 | X \in \URF_2, S_2 \in \URF_6\}$  
   if $k=2$ or $k=6$ then 
      forall $X' \in \mathcal{X}'$ do
        if $\exists S' \in \URF_6, \exists R' \in \URF_1 : \URP(\URF_1(R') \oplus X'||R')\leftp = S'$ then 
            $\URF_2(X') \inu \{0,1\}^n$ 
            $\ChainQuery(X',2)$    
\end{lstlisting}
 $\XorQuery_2$ and $\XorQuery_3$ are used as follows: Consider the
execution of $\ChainQuery(Y,3)$ for some $Y$ and in this execution a
recursive call to $\ChainQuery(X,2)$. These two procedures should
ensure that, under certain assumptions, the simulator does not abort
in the next two recursion levels, and after these two levels certain
properties hold. The same holds symmetrically for $\ChainQuery(A,5)$
for some $A$. Again, this is just an intuition and it is not clear at
this point if these goals are achievable with $\XorQuery_2$ and
$\XorQuery_3$.
\begin{lstlisting}[language=pseudocode,name=cpssim]
procedure $\XorQuery_2(x,k)$:
   $\mathcal{M} := \{(L,R,Z,A,S) |  \URP^{-1}(S||A\oplus \URF_6(S)) = L||R,$   
                                   $R \notin \URF_1,$ 
                                   $A \in \URF_5,$ 
                                   $S \in \URF_6,$ 
                                   $Z = \URF_5(A) \oplus S \}$ 
   forall $(L,R,Z,A,S) \in \mathcal{M}$ do
      if $k=6$ and $\exists Z' \in \URF_4 \setminus \{Z\}: \URP(L \oplus Z \oplus Z' ||R)\leftp = x$ or 
                 $k=3$ and $\exists S' \in \URF_6 : \URP(L \oplus x \oplus Z ||R)\leftp = S'$ then 
         $\URF_1(R) \inu \{0,1\}^n$ 
         $\ChainQuery(R,1)$ 
   $\mathcal{M} := \{(S,T,R,X,Y) |  \URP(X \oplus \URF_1(R)|| R) = S||T,$   
                                    $S \notin \URF_5,$ 
                                    $X \in \URF_2,$ 
                                    $R \in \URF_1,$ 
                                    $Y = \URF_2(X) \oplus R \}$ 
   forall $(S,T,R,X,Y) \in \mathcal{M}$ do
      if $k=1$ and $\exists Y' \in \URF_3 \setminus \{Y\}: \URP^{-1}(S||T \oplus Y \oplus Y')\rightp = x$ or 
                  $k=4$ and $\exists R' \in \URF_1, : \URP^{-1}(S||T \oplus x \oplus Y)\rightp = R'$ then 
         $\URF_6(S) \inu \{0,1\}^n$ 
         $\ChainQuery(S,6)$ 
 
 
procedure $\XorQuery_3(x,k)$:
   $\mathcal{R} := \{(Y,R_1,R_2) |  \URP^{-1}(S_1||A_1\oplus \URF_6(S_1)) = L_1||R_1,$   
                                          $\URP^{-1}(S_2||A_2\oplus \URF_6(S_2)) = L_2||R_2,$   
                                          $Y \notin \URF_3,$   
                                          $S_1 \in \URF_6$   
                                          $Z_1 = \URF_5(A_1) \oplus S_1, Z_1 \in \URF_4,$   
                                          $A_1 = \URF_4(Z_1) \oplus Y, A_1 \in \URF_5$   
                                          $S_2 \in \URF_6$   
                                          $Z_2 = \URF_5(A_2) \oplus S_2, Z_2 \in \URF_4,$   
                                          $A_2 = \URF_4(Z_2) \oplus Y, A_2 \in \URF_5 \} $ 
    if $k=3$ and $\exists (Y, R_1, R_2) \in \mathcal{R}: Y = x \oplus R_1 \oplus R_2$ then 
       $\URF_3(Y) \inu \{0,1\}^n$
       $\ChainQuery(Y,3)$    
    $\mathcal{S} := \{(Z,S_1,S_2) |  \URP(\URF_1(R_1) \oplus X_1||R_1) = S_1||T_1,$   
                                          $\URP(\URF_1(R_2) \oplus X_2||R_2) = S_2||T_2,$   
                                          $Z \notin \URF_4,$   
                                          $R_1 \in \URF_1$   
                                          $Y_1 = \URF_2(X_1) \oplus R_1, Y_1 \in \URF_3,$   
                                          $X_1 = \URF_3(Y_1) \oplus Z, X_1 \in \URF_2$   
                                          $R_2 \in \URF_1$   
                                          $Y_2 = \URF_2(X_2) \oplus R_2, Y_2 \in \URF_3,$   
                                          $X_2 = \URF_3(Y_2) \oplus Z, X_2 \in \URF_2 \}$   
   if $k=4$ and $\exists (Z, S_1, S_2) \in \mathcal{S}: Z = x \oplus S_1 \oplus S_2$ then 
       $\URF_4(Z) \inu \{0,1\}^n$
       $\ChainQuery(Z,4)$    
\end{lstlisting}

\paragraph{Illustrations describing the $\XorQuery$ procedures.} 
We provide illustrations to better describe procedures $\XorQuery_1$,
$\XorQuery_2$, and $\XorQuery_3$.

Fig.~\ref{figXorQuery1} illustrates how $\XorQuery_1$ works. In the
figure, the values that are required to be in the history are marked
with boxes, and the value $x$ that $\XorQuery_1$ is called upon is
marked with a circle. We abbreviate $\ChainQuery$ by \textsf{CQ}.

For example, the left upper quarter of Fig.~\ref{figXorQuery1}
describes calls to $\XorQuery_1(A,5)$ for some $A$. Upon such a call,
$\sitwo$ computes the values $A'$ for any pairs $R_1, R_2$ in $\URF_1$
where $R_1 \neq R_2$. Now if some $A'$ is not in $\URF_5$ (in the
figure, there is no box around $A'$), then $\sitwo$ checks if there is
a 3-chain $(A', S', R')$ for some $S' \in \URF_6$ and $R' \in \URF_1$
(in the figure, there are boxes around these values). If such a chain
is found, $\sitwo$ calls $\ChainQuery(A',5)$. In the figure, we write
$\XorQuery_1(A,5) \rightarrow \textsf{CQ}(A',5)$ to say this.
\begin{figure}[ht]
                \begin{center}
                        \includegraphics[scale=0.5]{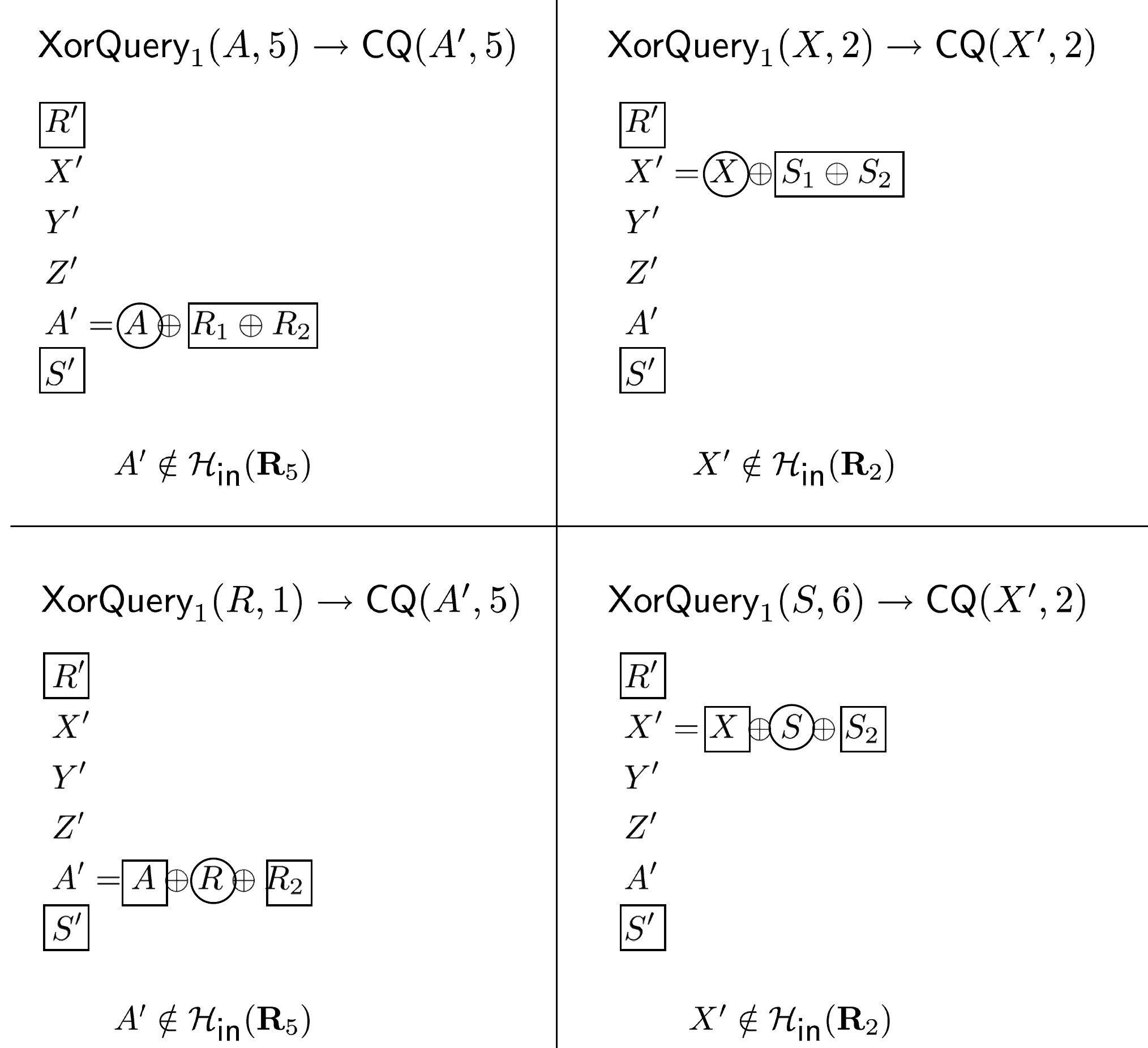}
                \end{center}
                \caption{An illustration for $\XorQuery_1$}
                \label{figXorQuery1}
\end{figure}

Fig.~\ref{figXorQuery2} provides an illustration of $\XorQuery_2$. The notation is the same as in the illustration for $\XorQuery_1$.
\begin{figure}[ht]
                \begin{center}
                        \includegraphics[scale=0.47]{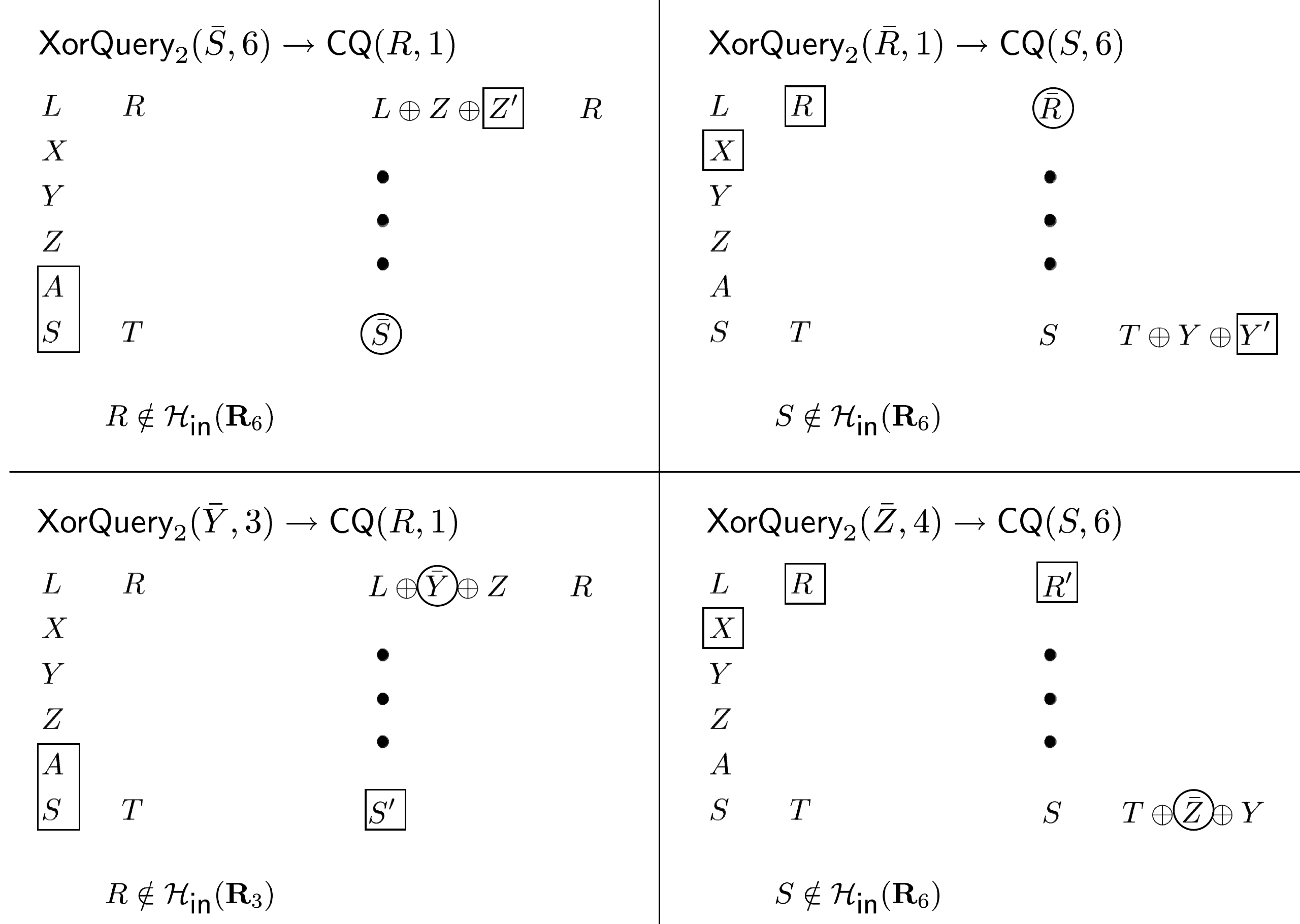}
                \end{center}
                \caption{An illustration for $\XorQuery_2$}
                \label{figXorQuery2}
\end{figure}

For $\XorQuery_3$ we provide Fig.~\ref{figXorQuery3} to faciliate the understanding.
\begin{figure}[ht]
                \begin{center}
                        \includegraphics[scale=0.47]{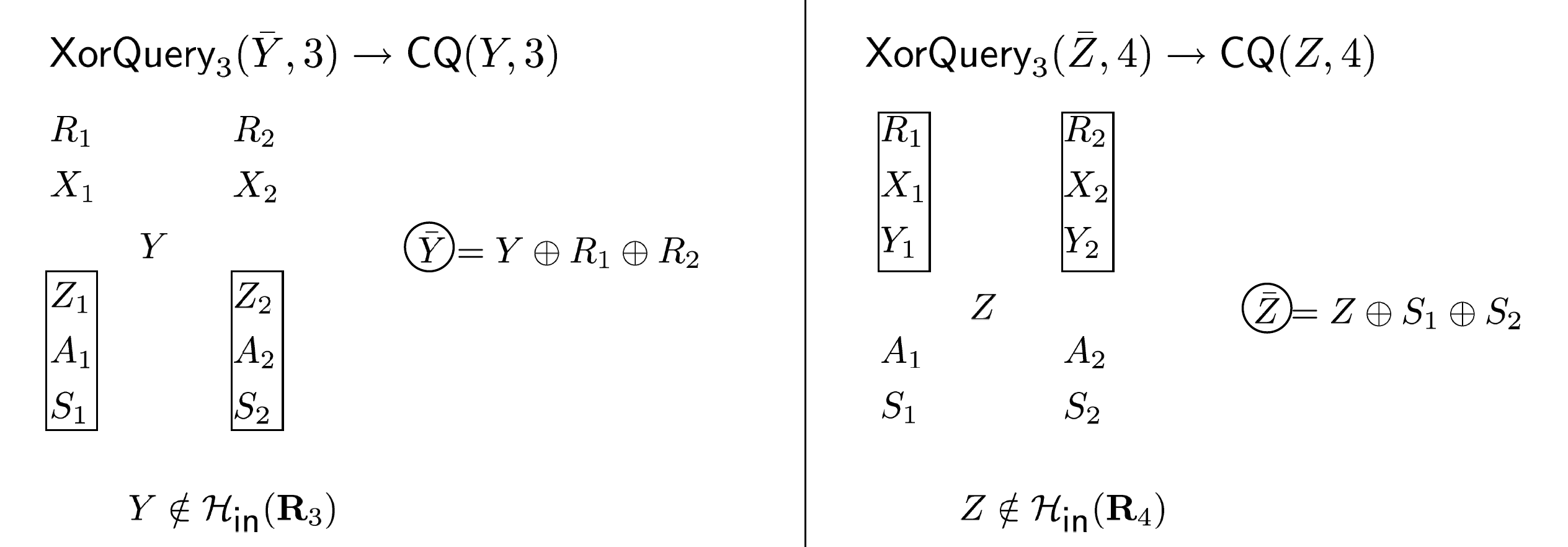}
                \end{center}
                \caption{An illustration for $\XorQuery_3$}
                \label{figXorQuery3}
              \end{figure} \\
              Here is a description of $\XorQuery_3$ for the case of
              $k=3$: Upon query $(x,3)$ where $x = \bar{Y}$ the
              procedure sets $\URF_3(Y) \inu \{0,1\}^n$ and calls
              $\ChainQuery(Y,3)$ for any $Y \notin \URF_3$ if there
              are $3$-chains $(Z_1, A_1, S_1)$ and $(Z_2, A_2, S_2)$
              (as in the figure) in the history and the corresponding
              $R_1$ and $R_2$ are such that $\bar{Y} = Y \oplus R_1
              \oplus R_2$.

\section{Detailed Analysis of the Attack against the Simulator of Coron et al.}

\label{app:analysisattack}

Formally, the distinguisher $\di$ is defined as follows. (The
distinguisher's output bit is irrelevant, since its goal is to let the
simulator abort.)

      \lstdefinelanguage{pseudocode}{numbers=left,numberstyle=\tiny,mathescape,escapechar=*,morekeywords={System,Distinguisher,
          abort,procedure,if,else,do,done,endif,forall,new,assert,end,return,while,then,private,public,
          Simulator, Variables,true,false},flexiblecolumns}

\begin{lstlisting}[language=pseudocode,name=distinguisher]
Distinguisher $\di$
    $X \leftarrow_R \{0,1\}^n$, $R_2 \leftarrow_R \{0,1\}^n$, $R_3 \leftarrow_R \{0,1\}^n$
    $L_2 := \URF_1(R_2) \oplus X$, $L_3 := \URF_1(R_3) \oplus X$
    $S_2\|T_2 := \URP(L_2\|R_2)$, $S_3\|T_3 := \URP(L_3\|R_3)$
    $A_2 := \URF_6(S_2) \oplus T_2$, $A_3 := \URF_6(S_3) \oplus T_3$
    $R_1 := R_2 \oplus A_2 \oplus A_3$
    $L_1 := \URF_1(R_1) \oplus X$
    $S_1\|T_1 := \URP(L_1\|R_1)$
    $A_1 := \URF_6(S_1) \oplus T_1$
    $\bar{A} := A_1 \oplus R_1 \oplus R_2$
    query  $\URF_5(\bar{A})$
\end{lstlisting}

\newcommand{\Bad}{\mathsf{Bad}} 
\newcommand{\Hist}{\widetilde{\mathcal{H}}}

\noindent {\em Implementing $\URP$.} The distinguisher makes $7 < 2^3$
queries to $\URF$ and three permutation queries. Assume without loss
of generality that at most $B := 2^{50}$ queries are made to the
permutation $\URP$ (or its inverse $\URP^{-1}$) by the simulator or by
the distinguisher. Once more than $B$ $\URP$ queries are made, we
assume that the simulator $\sitwo$ aborts. By Lemma 1
in~\cite{CPS08v2}, the probability that abort occurs in the original
experiment (where no limit on the number of $\URP$ queries made by the
simulator is imposed) is at most $2^{55}/2^n = \mathcal{O}(2^{-n})$
lower than in the version where the query number is bounded by
$B$. Note that, while large, $B$ is constant, and although it could be
made significantly smaller for the purposes of this section, we use
the larger value to rely on the analysis of \cite{CPS08v2}.

\newcommand{\aup}{\uparrow} \newcommand{\ado}{\downarrow}

It is convenient to think of $\URP$ as being implemented as follows:
Initially, two lists $\set{L}_{\ado}$ and $\set{L}_{\aup}$ of $B$
uniformly distributed, but distinct, $2n$-bit values are
generated. Then, each time a $\URP$ query is issued with input $x$, we
first check if $\URP(x)$ is defined (in which case we simply return
the previously defined value), or if $\URP^{-1}(y) = x$ (in which case
we return $y$.) Otherwise, we assign to $\URP(x)$ the first value $y$
in $\set{L}_{\ado}$ such that $\URP^{-1}(y)$ is undefined, and let
$\URP^{-1}(y) := x$. Also, $\URP^{-1}$ queries are answered
symmetrically.  It is not hard to verify that this gives rise to a uniform
random permutation as long as at most $B$ queries are made: Each
forward query $\URP(x)$ assigns a value from $\set{L}_{\ado}$ which is
uniformly distributed among all values for which $\URP^{-1}(y)$ is not
defined, and, if $x \in \set{L}_{\aup}$, ensures that $x \in
\set{L}_{\ado}$ cannot be used as an answer to a query to $\URP^{-1}$.

For simplicity, denote as $\set{L}_1$ and $\set{L}_2$ the lists
containing the first and second halves, respectively, of the elements of
$\set{L} := \set{L}_{\ado}\|\set{L}_{\aup}$. (Here, $\|$ denotes list
concatenation.)  \bigskip

\noindent {\em Initialization.} We consider the interaction of $\di$
and $\sitwo$, and show that the latter aborts with overwhelming
probability. Before executing Line 11, it is clear that no additional
$\URF_i(x)$ entry is defined by one of the \XorQuery\ calls, since
only $\XorQuery_1$ and $\XorQuery_2$ can be called when answering
queries to either of $\URF_1$ or $\URF_6$, but the histories of
$\URF_2, \URF_3, \URF_4, \URF_5$ are all empty. Thus when
$\URF_5(\bar{A})$ is called, only the values $R_1, R_2, R_3, S_1, S_2,
S_3$ are in the history of $\sitwo$, and consequently, no $3$-chain
exists so far. Also, the first three elements of $\set{L}$ are
$S_1\|T_1$, $S_2\|T_2$, and $S_3\|T_3$.

The following definitions introduce bad events: Their definition is
tailored at what is needed later, and may appear confusing at first. (We
invite the reader to skip their definition, and come back later.)

\begin{definition}
  The event $\Bad_{\URP}$ occurs if one of the following is true:
  \begin{enumerate}[(i)]
  \item There exists a collision among the first or second halves of
    the elements of $\set{L}$;
  \item $\set{L}_2 \cap \{R_1, R_2, R_3\} \ne \emptyset$, that is, the
    second half of some element of $\set{L}$ is in $\{R_1, R_2, R_3\}$;
  \item There exists $R, R' \in \set{L}_2$ such that $R_1 \oplus R_2 = R \oplus R'$.
  \end{enumerate}\end{definition}

\begin{definition}
  The event $\Bad_1$ occurs if one of the following is true:
\begin{enumerate}[(i)]
\item Two elements among $R_1, R_2, R_3$ collide;
\item $A \oplus \URF_6(S) \in \set{L}_2$ for $(A, S) \in (\URF_5 \cup \{\bar{A}\})
  \times \URF_6 \setminus \{(A_i, S_i)\,|\, i = 1,2,3\}$;
\item There exists $(i, j) \in \{(1,3), (2,3)\}$ and $k \in \{1,2,3\}$
  such that for $A' := \bar{A} \oplus R_i \oplus R_j \notin
  \URF_5$ we have $\URF_6(S_k) \oplus A' \in \set{L}_2$;
\item There exist $i,j,k,\ell \in \{1,2,3\}$, $i \ne j$, such that for $A'' := A_\ell \oplus R_i \oplus R_j
  \notin \URF_5$ we have $\URF_6(S_k) \oplus A'' \in
  \set{L}_2$;
\item $R_i \oplus R_j \oplus A_j \oplus A = 0^n$ for $i \ne j$, $A \in
  \{\bar{A}, A_1, A_2, A_3\}$, and $(i, j, A) \notin \{(i,i,A_i)\,|\,
  i = 1,2,3\} \cup \{(1,2, A_3), (2,1, \bar{A})\}$.
\end{enumerate}
\end{definition}

The following two lemmas upper bound the probability of these events
occurring.

\begin{lemma}
  $\Pr[\Bad_{\URP}] = \mathcal{O}(B^2 \cdot 2^{-n})$.
\end{lemma}
\begin{proof}
  For (i), note that for any two $i, j \in \{1, \ldots, 2B\}$, $i \ne j$, and $h \in \{1,2\}$,
  \begin{displaymath}
    \Pr[V_i|_h = V_j|_h] = 2^n \cdot \frac{2^n \cdot (2^{n} - 1)}{2^{2n} \cdot (2^{2n} - 1)} = \mathcal{O}(2^{-n})
  \end{displaymath}
  and $\Pr[V_i|_h \in \{R_1,R_2,R_3\}] \le 3 \cdot 2^{-n}$. The bound
  follows by the union bound. Finally, an upper bound of the
  probability of (iii) is $B^2 \cdot 2^n \frac{2^n \cdot (2^n -
    1)}{2^{2n} (2^{2n}-1)} = \mathcal{O}(B^2 \cdot 2^{-n})$.
\end{proof}

\begin{lemma}
  $\Pr[\Bad_1] = \mathcal{O}(B \cdot 2^{-n})$.
\end{lemma}
\begin{proof}
  For (i), we observe that the random variable $R_1$ is uniform and
  independent of $R_2$ and $R_3$, since $A_2$ and $A_3$ are
  independent of $R_2$ and $R_3$ due to $\URF_6(S_2)$ and
  $\URF_6(S_3)$ being chosen uniformly and independently. Thus, the
  probability that $R_1 = R_2$, $R_2 = R_3$ or $R_1 = R_3$ is at most
  $3 \cdot 2^{-n}$.

  It is also easy to verify that the value $T' := A \oplus \URF_6(S)$
  is uniformly distributed and independent of $\set{L}$ for $(A, S)
  \in (\URF_5 \cup \{\bar{A}\}) \times \URF_6 \setminus \{(A_i,
  S_i)\,|\, i = 1,2,3\}$, and therefore $T' \in \set{L}_2$ holds with
  probability at most $2B \cdot 2^{-n}$; an upper bound on the
  probability of (ii) follows by the union bound.

  To bound (iii), we use the fact that for all $(i, j) \in \{(1,3),
  (2,3)\}$ and $k \in \{1,2,3\}$ the value $T' := \URF_6(S_k) \oplus
  A'$ equals
  \begin{displaymath}
    \URF_6(S_k) \oplus \URF_6(S_1) \oplus T_1  \oplus R_1 \oplus R_2 \oplus R_i \oplus R_j 
  \end{displaymath}
  and is hence uniformly distributed and independent of
  $\set{L}$. Therefore, (iii) occurs with probability at most $2B
  \cdot 2^{-n}$. Similarly, to bound (iv), we observe that for all $i
  \ne j$, and $k, \ell \in \{1,2,3\}$ the value $T'' := \URF_6(S_k) \oplus
  A''$ equals
  \begin{displaymath}
    \URF_6(S_k) \oplus \URF_6(S_\ell) \oplus T_\ell \oplus R_i \oplus R_j
  \end{displaymath}
  and is therefore uniformly distributed, and independent of
  $\set{L}$.

  In (v) we see that in all cases, substituting $\bar{A}$ with $A_1
  \oplus R_1 \oplus R_2$, $A_i$ with $\URF_6(S_i) \oplus T_i$, 
  and $R_1$ with $R_2 \oplus A_2 \oplus A_3$, we end up in one of the
  follwing two cases: Either the given sum still contains at least one term which is uniformly
  distributed and independent of the other terms, and thus the sum
  equals $0^n$ with probability $2^{-n}$. Or, the resulting equation 
  is $R_2 = R_3$, which does not hold by the choice of these values by $\di$.
	The actual bound follows by
  a union bound over all possible combinations. 
\end{proof}

From now on, the analysis assumes that neither of $\Bad_{\URP}$ and
$\Bad_{1}$ has occurred.  Before we proceed in analyzing the rest of
the execution, we prove the following lemmas, which will be useful to
simplify the analysis below, and rely on the assumptions that the above
events do not occur.

\begin{lemma}
  \label{lem:noxorquery1}
  As long as $\URF_1 = \{R_1, R_2, R_3\}$ and
  $\URF_6 = \{S_1, S_2, S_3\}$, no $\XorQuery_1(A_i, 5)$ (for
  $i = 1, 2, 3$) call results in a recursive $\ChainQuery$ call.
\end{lemma}
\begin{proof}
  For the {\bf if} statement within $\XorQuery_1(A_i, 5)$ call to be
  satisfied, there must exist $A'' = A_i \oplus R_i \oplus R_j$ such
  that $R_i \ne R_j$ and $\URP^{-1}(S'\|\URF(S') \oplus A'')|_2 \in
  \{R_1, R_2, R_3\}$, where $S' \in \{S_1, S_2, S_3\}$. However, since
  $\URF(S') \oplus A'' \notin \set{L}_2$, as this implies $\Bad_{1}$,
  we need to have $\URP^{-1}(S'\|\URF_6(S') \oplus A'') \in
  \set{L}$. But then, since $\URP^{-1}(S'\|\URF(S') \oplus A'')|_2 \in
  \{R_1, R_2, R_3\}$, we also have $\Bad_{\URP}$.
\end{proof}

\begin{lemma}
  \label{lem:no:xorquery2}
  Assume that $\URF_1 = \{R_1, R_2, R_3\}$ and
  $\URF_6 = \{S_1, S_2, S_3\}$. Then, as long as
  $\Bad_{\URP}$ holds, a call to $\XorQuery_2(x, k)$ for $k \in
  \{3,4\}$ and $x \notin \URF_{7 - k}$ does not set any new
  values and does not make any new recursive \ChainQuery\ calls.
\end{lemma}
\begin{proof}
  Let $k = 3$ and consider a tuple $(L, R, Z, A, S) \in \set{M}$: This
  means that querying $\URP^{-1}$ on input $S\|\URF_6(S) \oplus A$
  (with $S = S_i$ for some $i \in \{1, 2,3\}$) returns a pair $L\|R$
  with $R \notin \{R_1, R_2, R_3\}$. This in particular means that $A
  \ne A_i$ (as otherwise $R = R_i$), and also that $L\|R \in \set{L}$,
  as otherwise $S_i\|\URF_6(S_i) \oplus A \in \set{L}$, and (ii) for
  $\Bad_1$ would have occurred.  However, if the {\bf if} statement is
  satisfied, this also means that $\URP(L \oplus x \oplus Z\|R)|_1 =
  S' \in \{S_1, S_2, S_3\}$. But since $\Bad_{\URP}$ has not happened,
  this means that there has been a previous $\URP^{-1}$ query with
  input $S'\|T'$ which has returned $L \oplus x \oplus Z\|R \in
  \set{L}$. Yet, since $L \oplus x \oplus Z \ne L$ (due to $x \ne Z$,
  as otherwise $x \in \URF_4$), this also means that there has been a
  collision on the second half, and $\Bad_{\URP}$ has occurred.

  The case for $k = 4$ is fully symmetric.
\end{proof}

Added complexity in the execution of $\sitwo$ stems from the fact that
it tests for so-called ``virtual chains'', and we want to argue that
they do not play a role in the upcoming analysis of the attack. First,
note that as long as $\URF_2$ contains at most one element
(this is the case for most of the attack), $\URF^\ast_6 =
\URF_6$. Additionally, when calling $\ChainQuery(A, 5)$ for
$A \in \{\bar{A}, A_1, A_2, A_3\}$, in order for $\URF_1 \ne
\URF^\ast_1$ to occur, we need that there exist $(A, S), (A',
S') \in \{\bar{A}, A_1, A_2, A_3\} \times \{S_1, S_2, S_2\}$ such that
$A \ne A'$ and $\URP^{-1}(S\|\URF_6(S) \oplus A)|_2 =
\URP^{-1}(S'\|\URF_6(S') \oplus A')$. It is not hard to verify that
any possible case implies $\Bad_{\URP}$ or $\Bad_1$, and hence we can
safely ignore virtual chains in the following. (We will indeed need to
ignore virtual chains only for as long as the conditions needed for
these arguments hold.)
\bigskip

\noindent{\em First phase of the simulator's execution.}  From now on,
we continue the analysis of the execution under the assumption that
$\Bad_1$, and $\Bad_{\URP}$ have not occurred. Upon querying
$\URF_5(\bar{A})$, the simulator sets $\URF_5(\bar{A}) \inu \{0,1\}^n$ 
and then first executes
$\XorQuery_1(\bar{A},5)$: Note that $A_1 = \bar{A} \oplus R_1 \oplus
R_2 \notin \URF_5$ (since, at this point, the history of $\URF_5$ is
empty), and $A_1$ satisfies the {\bf if} statement (with $S' = S_1$
and $R' = R_1$) so that $\URF_5(A_1) \inu \{0,1\}^n$ is set and
$\ChainQuery(A_1,5)$ is called. Also, no other $\ChainQuery$ call
occurs, as if the condition in the {\bf if} statement is true for some
other value, then since (ii) in the definition of $\Bad_1$ does not
occur, this means that, for some input $x$, $\URP(x)$ has been
assigned a value from $\set{L}$ whose second half is in $\{R_1, R_2,
R_3\}$, which implies $\Bad_{\URP}$.

Moreover, in the subsequent execution of $\ChainQuery(A_1, 5)$ the
procedure $\XorQuery_1$ also does not invoke $\ChainQuery$ by
Lemma~\ref{lem:noxorquery1}. Moreover, $\Chain(+, A_1, 5) = \{(S_1,
R_1)\}$, since $(A_1, S_i, R_j)$ for $(i, j) \ne (1, 1)$ cannot
constitute a chain, as this would either imply (i) in the definition
of $\Bad_1$ or the fact that $\Bad_{\URP}$ occurs. Also, $\Chain(-,
A_1, 5) = \emptyset$, since no $\URF_3$ and $\URF_4$ values have been
defined so far. Therefore, $(S_1, R_1) \in \Chain(+,A_1,5)$ is found
and gets completed by $\CompleteChain$ to the tuple $(R_1, X, Y_1,
Z_1, A_1, S_1)$, by defining
\begin{eqnarray*}
	X := \URF_1(R_1) \oplus L_1, & Z_1 := \URF_5(A_1) \oplus S_1, &  \URF_{4}(Z_1) \leftarrow_R \{0,1\}^{n}, \\
	Y_1 := \URF_4(Z_1) \oplus A_1, & \URF_2(X) := R_1 \oplus Y_1,&  \URF_3(Y_1):= Z_1 \oplus X.
\end{eqnarray*}
We consider the following event defined on these new values.

\begin{definition}
  The event $\Bad_2$ occurs if one of the following holds:
  \begin{enumerate}[(i)]
  \item $Y_1 = Z_1$; 
  \item $Y_1 = Z_1 \oplus R_1 \oplus R$ where $R \in\{R_2, R_3\}$;
  \item $Z_1 \oplus \URF_5(\bar{A}) \in \set{L}_1$;
  \item $Z_1 \oplus \URF_5(\bar{A}) = S_1$.
  \end{enumerate}
\end{definition}
\begin{lemma}
  $\Pr[\Bad_2] = \mathcal{O}(B \cdot 2^{-n}) $
\end{lemma}
\begin{proof}
  For (i) and (ii), since $Y_1$ is uniform and independent of $Z_1$,
  $R_1$, and $R_2$ by the fact that $\URF_5(Y_1)$ is set uniformly,
  the values on both sides are equal with probability $2^{-n}$. 
  Furthermore, both $\URF_5(\bar{A})$ and $\URF_5(A_1)$ are set
  uniformly (since $A_1 \ne \bar{A}$ by $\Bad_1$ not occurring), and
  thus $Z_1 \oplus \URF_5(\bar{A}) = S_1 \oplus \URF_5(A_1) \oplus
  \URF_5(\bar{A})$ is uniform and independent of $\set{L}_1$, and is
  thus in the set with probability at most $2B \cdot 2^{-n}$ by the
  union bound. Similarly, we show that (iv) occurs with probability
  $2^{-n}$ only.
\end{proof}

\noindent {\em Second phase of the simulator's execution.}
Subsequently, the simulator schedules calls, in arbitrary\footnote{In
  particular, we do not want to make any assumption on the order in
  which they are called. Of course, it is easier to provide a proof if
  a certain processing order is assumed, and this would suffice to
  give a strong argument against $\sitwo$. Still, we opt for showing
  the strongest statement.} order, to $\ChainQuery(X, 2)$,
$\ChainQuery(Y_1, 3)$ and $\ChainQuery(Z_1, 3)$.

  \begin{lemma}
    No $\ChainQuery$ invocation preceding the invocation of
    $\ChainQuery(X, 2)$ issues a recursive $\ChainQuery$
    call. Furthermore, $\XorQuery_1(X, 2)$ within $\ChainQuery(X, 2)$
    also does not trigger a $\ChainQuery$ invocation.
  \end{lemma}
  \begin{proof}
    Assume that $\ChainQuery(X, 2)$ has not been invoked yet.  Calls
    to $\XorQuery_2$ cannot invoke $\ChainQuery$ recursively by
    Lemma~\ref{lem:no:xorquery2} and the fact that $Y_1 \ne Z_1$ (using $\Bad_2$). Also
    note that $\XorQuery_3(Y_1, 3)$ cannot produce recursive
    $\ChainQuery$ calls: Note that $A_1 \ne \bar{A}$ (equality implies
    (i) in $\Bad_1$), and therefore any triple $(Y, R, R') \in
    \set{R}$ {\em must} satisfy $R = R'$ and $Y \notin \URF_3$. But
    then, we cannot have $Y_1 = Y \in \URF_3$, and thus the {\bf if}
    statement is never satisfied. Similarly, the fact that
    $\XorQuery_3(Z_1, 4)$ does not invoke \ChainQuery\ follows from
    the fact that both $\URF_3$ and $\URF_2$ only contain one single
    element.

    For both possible \ChainQuery\ calls, it also clear that no
    additional $3$-chains to be completed are found. Namely, within
    $\ChainQuery(Y_1, 3)$, only the chains $(R_1, X, Y_1)$ and $(Y_1,
    Z_1, A_1)$ are possible, and both have been completed. Moreover,
    when running $\ChainQuery(Z_1, 4)$, the $3$-chain $(X, Y_1, Z_1)$
    has also already been completed, whereas no chain $(Z_1, \bar{A},
    S_1)$ exists, as this would yield $\Bad_2$.

    Finally, when $\ChainQuery(X, 2)$ is invoked, we also observe that
    if the {\bf if} statement in the execution of $\XorQuery_1(X, 2)$
    cannot be true: It would imply that there exists $X' \ne X$ such
    that $\URP(X' \oplus \URF_1(R_i)\|R_i)|_1 \in \{S_1, S_2, S_3\}$,
    but since $\URF_2(X') \oplus R_i \ne L_i$, this implies
    $\Bad_{\URP}$ (i).  \end{proof}

  We hence can consider the execution of $\ChainQuery(X, 2)$, as any
  previous $\ChainQuery$ invocation does not affect the history of the
  simulated round functions. First, note that $\Chain(-,X,2) = \{(R_2,
  S_2), (R_3, S_3)\}$ and $\Chain(+, X, 2) = \{(Y_1, Z_1)\}$. The
  latter chain was already completed, therefore both negative chains
  are completed, by defining values
  \begin{eqnarray*}
    Y_2 := \URF_2(X) \oplus R_2, & \URF_3(Y_2) \leftarrow_R \{0,1\}^n,&  Z_2 := X \oplus \URF_3(Y_2),\\
 \URF_4(Z_2) := Y_2 \oplus A_2, &   \URF_5(A_2) := Z_2 \oplus S_2,&  
  \end{eqnarray*}
  as well as
  \begin{eqnarray*}
    Y_3 := \URF_2(X) \oplus R_3, & \URF_3(Y_3) \leftarrow_R \{0,1\}^n,&  Z_3 := X \oplus \URF_3(Y_3),\\
    \URF_4(Z_3) := Y_3 \oplus A_3, &   \URF_5(A_3) := Z_3 \oplus S_3.&  
  \end{eqnarray*}
  In addition, let us introduce the last bad event in this analysis,
  defined on the newly defined values.
 
  \begin{definition}
    The event $\Bad_3$ occurs if one of the following holds:
    \begin{enumerate}[(i)]
    \item $\URF_3 \cap\URF_4 \ne \emptyset $;
    \item $Z \oplus \URF_5(A) \in \set{L}_1$ for $(Z, A) \in
      \URF_4 \times \URF_5 \setminus \{(Z_i, A_i)\,|\, i = 1,2,3\}$.
    \end{enumerate}
  \end{definition}

\begin{lemma}
  $\Pr[\Bad_3] \le \mathcal{O}(2^{-n})$
\end{lemma}
\begin{proof}
  For (i), note that the value of $Z_i$ is independent of the value of
  $Y_1, Y_2, Y_3$ for $i = 2, 3$, and thus equality only occurs with
  negligible probability. (Recall that we already assume that $Y_1 \ne
  Z_1$). Moreover, the case that $Z_1$ equals $Y_2$ or $Y_3$ is would
  already imply $\Bad_2$: This is because we would have, for $i =
  2,3$,
  \begin{displaymath}
    Z_1 = Y_i = \URF_2(X) \oplus R_i = R_1 \oplus Y_1 \oplus R_i.
  \end{displaymath}
  For (ii), we claim that $Z \oplus \URF_5(A)$ is always uniform and
  independent of $\set{L}$ when it is defined. If $Z = Z_1$, then note
  that $\URF_5(\bar{A})$ is set uniformly and independently of $Z_1$
  and $\set{L}_1$, whereas $\URF_5(A_i) = X \oplus \URF_2(Y_i) \oplus
  S_i$ for $i = 2,3$, where $\URF_2(Y_i)$ is set independently and
  uniformly. If $Z = Z_i$ for $i = 2, 3$, then note that $\URF_5(A_1)$
  and $\URF_5(\bar{A})$ are set uniformly, whereas by the above
  $\URF_5(A_{5 - i})$ is also independent and uniform.
\end{proof}

\noindent {\em Final phase of the simulator's execution.}  From now
on, $\ChainQuery(Y_i, 3)$, $\ChainQuery(Z_i, 4)$ and $\ChainQuery(A_i,
5)$ for $i = 2,3$ are called, in any order.\footnote{Once again, we
  dispense with any assumption on the execution order adopted by the
  simulator.}  The crucial point is reached as soon as one of
$\ChainQuery(Y_2, 3)$ or $\ChainQuery(A_3, 5)$ is invoked. However, we
need to show that all calls preceding one of these calls do not start
recursions or set additional function values. First, however, we show
that no other chains occur, other than those we expect.

\begin{lemma}
  \label{lem:nochain2}
  Only the $3$-chains $(Y_i, Z_i, A_i)$ for $i=1,2,3$, $(Y_1, Z_2,
  A_3)$ and $(Y_2, Z_1, \bar{A})$ exist in $\URF_3 \times
  \URF_4 \times \URF_5$. Moreover, only the three $3$-chains 
  $(Z_i, A_i, S_i)$ for $i = 1, 2, 3$ exist in $\URF_4
  \times \URF_5 \times \URF_6$. 
\end{lemma}
\begin{proof}
  Such a chain $(Y_i, Z_j, A)$ implies
  \begin{displaymath}
    \begin{split}
      Y_i \oplus \URF_4(Z_j) \oplus A & = R_i \oplus \URF_2(X)
      \oplus \URF_4(Z_j) \oplus A = R_i \oplus \URF_2(X) \oplus Y_j
      \oplus A_j \oplus A \\
      & = R_i \oplus R_j \oplus A_j \oplus A = 0^n,
    \end{split}
  \end{displaymath}
  and hence $\Bad_1$. The second part of the statement follows from
  $\Bad_3$ not occurring and the fact that $S_1, S_2, S_3 \in
  \set{L}_1$.
\end{proof}

\begin{lemma}
  \label{lem:nochainquery3}
  No $\ChainQuery$ call preceding the invocation of both
  $\ChainQuery(Y_2, 3)$ and $\ChainQuery(A_3, 5)$ provokes a recursive
  $\ChainQuery$ call. Also, no $\XorQuery$ within the execution of
  $\ChainQuery(Y_2, 3)$ and $\ChainQuery(A_3, 5)$ (whichever is
  executed first) provokes a recursive $\ChainQuery$ call.
\end{lemma}
\begin{proof}
  We know, by Lemma~\ref{lem:no:xorquery2} and item (i) in the
  definition of $\Bad_3$ not taking place, that calls to
  $\XorQuery_2(Y_3, 3)$, $\XorQuery_2(Z_i, 4)$ for $i = 1, 2$ do not
  provoke any recursive $\ChainQuery$ calls. Furthermore,
  $\XorQuery_3(Y_3, 3)$ cannot invoke $\ChainQuery$, as by
  Lemma~\ref{lem:nochain2}, the fact that only the three chains $(Z_i,
  A_i, S_i)$ for $i = 1, 2, 3$ exist in $\URF_4 \times \URF_5 \times
  \URF_6$, the set $\set{R}$ must be empty.  Also, $\XorQuery_3(Z_i,
  4)$ cannot invoke $\ChainQuery$ by the fact that $X$ is the only
  element of $\URF_2$.

  In addition, none of these $\ChainQuery$ invokes $\CompleteChain$
  because of Lemma~\ref{lem:nochain2}.

  The fact that $\XorQuery_1(A_2, 5)$ in $\ChainQuery(A_2, 5)$ does
  not make any $\ChainQuery$ calls is implied by
  Lemma~\ref{lem:noxorquery1}, whereas if $\ChainQuery(Y_2, 3)$ is
  invoked first, then $\XorQuery_2(Y_2, 3)$ does not invoke
  $\ChainQuery$ by Lemma~\ref{lem:no:xorquery2}, whereas
  $\XorQuery_3(Y_2, 3)$ does not invoke $\ChainQuery$ because of
  Lemma~\ref{lem:nochain2} as above.
\end{proof}

Finally, we can distinguish two cases:
  \begin{enumerate}[(1)]
  \item $\ChainQuery(Y_2, 3)$ is invoked first. As shown above, no
    $\XorQuery$ call calls $\ChainQuery$ recursively.  Then, $\sitwo$
    now finds and completes the chain $(Z_1, \bar{A}) \in
    \Chain(+,Y_2,3)$. No other chains are found by
    Lemma~\ref{lem:nochain2}.

    The following values are set:
    \begin{eqnarray*}
    	 X' := \URF_3(Y_2) \oplus Z_1, & & S_4 := \URF_5(\bar{A}) \oplus Z_1, \\
       \URF_6(S_4) \inu \{0,1\}^n,&  & L_4\|R_4 := \URP^{-1}(S_4\|\URF_6(S_4) \oplus \bar{A}) \\
      \URF_1(R_4) := L_4 \oplus X', & & \URF_2(X') := Y_2 \oplus R_4. 
    \end{eqnarray*}
    In particular, note that since $S_4 \notin \set{L}_1$ (since $\Bad_3$ (ii) does not occur), then $R_4
    \in \set{L}_2$. But then $R_5 := \URF_2(X') \oplus Y_1 \notin
    \set{L}_2$, as otherwise this would imply that $\Bad_{\URP}$ (iii) occurs, since
    \begin{displaymath}
      R_5 = Y_2 \oplus R_4 \oplus Y_1 = \URF_2(X) \oplus R_2 \oplus R_4 \oplus Y_1 = R_1 \oplus R_2 \oplus R_4.
    \end{displaymath}
   Recall that we assume that $S_5 = \URF_5(A_3) \oplus Z_2 \notin
   \set{L}_1$, as this yields $\Bad_3$ (ii). But then, at some point (the
   latest at the invocation of $\ChainQuery(X', 2)$) the chain going
   through $X'$, $Y_1$, $Z_2$, and $A_3$ needs to be
   completed. However, no completion is possible, because $R_5 \notin
   \set{L}_2$ and $S_5 \notin \set{L}_1$ (by $\Bad_3$ (ii)). In fact, this holds {\em
     regardless} of which strategy the simulator employs to complete
   the chain, and in this concrete case, this is reflected by an
   abort.
  \item $\ChainQuery(A_3, 5)$ is invoked first. The argument is
    symmetric. First, by Lemma~\ref{lem:nochainquery3}, no $\XorQuery$
    calls trigger $\ChainQuery$ invocations, and $\sitwo$ now finds
    and completes the chain $(Y_1, Z_2) \in \Chain(-,A_3,5)$, and no
    other chains are found by Lemma~\ref{lem:nochain2}. This in
    particular means that the following values are set
    \begin{eqnarray*}
    	X' := \URF_3(Y_1) \oplus Z_2, & & S_5 := \URF_5(A_3) \oplus Z_2, \\
      \URF_6(S_5) \inu \{0,1\}^n,&  & L_5\|R_5 := \URP^{-1}(S_5\|\URF_6(S_5) \oplus A_3) \\
      \URF_1(R_4) := L_5 \oplus X', & & \URF_2(X') := Y_1 \oplus R_5. 
    \end{eqnarray*}
    Once again, we have that $R_4 := \URF_2(X') \oplus Y_2 \notin
    \set{L}_2$ because of $\Bad_{\URP}$ not occurring, since $R_5 \in \set{L}_2$, and
    \begin{displaymath}
      R_4 = Y_1 \oplus R_5 \oplus Y_2 = R_1 \oplus \URF_2(X) \oplus R_5 \oplus \URF_2(X) \oplus R_2 = R_1 \oplus R_2 \oplus R_5.
    \end{displaymath}
    Also, we have assumed that $\URF_6(S_4) \oplus \bar{A} \notin
    \set{L}_1$. But now, at some point, the chain going through $X'$,
    $Y_2$, $Z_1$, and $\bar{A}$ has to be completed. However, this is
    not possible since $R_4 \notin \set{L}_2$, and $\URF_6(S_4) \oplus
    \bar{A} \notin \set{L}_1$.
 \end{enumerate}

\section{A Stronger Attack}
\label{app:strongerattack}

We subdivide the distinguisher execution into three phases: chain
preparation, computation of chain values, and consistency check, and
its description is best represented by means of the following tables:
Note that if a value in the column \textit{Queries to $\si$} is named
with the letter $R$ (or $X, Y, Z, A, S$, respectively), it is issued
to $\URF_1$ (or $\URF_2, \URF_3, \URF_4, \URF_5, \URF_6$,
respectively).

\begin{description}
\item[Chain Preparation.] \hfill\ \\[1ex] \begin{tabular}{l|l|l}
                        \textbf{Step} & \textbf{$\di$ computes}                                                                                                                           & \textbf{Queries to $\si$}       \\ \hline
                        1& $X:= X_1=X_2=X_3=X_4$ u.a.r.                                                                                                                                                                 &                                                                                                                       \\
                        2& $R_2, R_3$ arbitrary s.t.~$R_2 \neq R_3$                                                                                                                     & $R_2, R_3$                                                    \\
                        3& $S_2||T_2 := \URP(X \oplus \URF_1(R_2) || R_2)$                                                                              & $S_2$                                                                               \\
                        4& $S_3||T_3 := \URP(X \oplus \URF_1(R_3) || R_3)$                                                                      &       $S_3$                                                                                   \\
                        5& $A_2 := \URF_6(S_2) \oplus T_2, A_3 := \URF_6(S_3) \oplus T_3$               &                                                                                                       \\
                        6& $R_1 := R_2 \oplus A_2 \oplus A_3$                                                                                                                                           &       $R_1$                                                                                           \\
                        7& $S_1||T_1 := \URP(X \oplus \URF_1(R_1) || R_1)$                                                                              &       $S_1$                                                                                   \\
                        8& $A_1 := \URF_6(S_1) \oplus T_1$                                                                                                                                              &                                                                                                         \\
                        9& $A_5 := A_1 \oplus R_1 \oplus R_2$                                                                                                                                           &                                                                                                       \\
                        10& $R_4 := R_3 \oplus A_3 \oplus A_5$                                                                                                                                  &       $R_4$                                                                           \\
                        11& $S_4||T_4 := \URP(X \oplus \URF_1(R_4) || R_4)$                                                                             &       $S_4$                                                                                   \\
                        12& $A_4 := \URF_6(S_4) \oplus T_4$                                                                                                                                             &                                                                                                                                       \\
                        13& $A_8 := A_4 \oplus R_4 \oplus R_3$                                                                                                                                  & $A_8$                                                                                                 \\
\end{tabular}
\item[Computation of Chain Values.] \hfill\ \\[1ex]
\begin{tabular}{l|l|l}
                        \textbf{Step} & \textbf{$\di$ computes}                                                                                                                                 & \textbf{Queries to $\si$}    \\ \hline
                        14&                                                                                                                                                                                                                                             & $X$                                                                                                   \\
                        15&                                                                                                                                                                                                                                                                             & $A_1, A_2, A_3, A_4$                            \\
                        16& $Z_i := \URF_5(A_i) \oplus S_i$     for $i=1,2,3,4$                                                                         & $Z_1, Z_2, Z_3, Z_4$            \\
                        17& $Y_i := \URF_2(X)   \oplus R_i$ for $i=1,2,3,4$                                                                     &       $Y_1, Y_2, Y_3, Y_4$                    \\
                        18& $Y_6:=Y_1, Y_5:=Y_2, Y_8:= Y_3, Y_7:= Y_4$                                                                                                  &                                                                                                               \\
                        19& $Z_5:=Z_1, Z_6:=Z_2, Z_7:= Z_3, Z_8:= Z_4$                                                                                                  &                                                                                                               \\
                        20& $A_6:= \URF_4(Z_6) \oplus Y_6$                                                                                                                                              & $A_6$                                                                                                 \\
                        21& $X_5:=\URF_3(Y_5) \oplus Z_5, X_6:=\URF_3(Y_6) \oplus Z_6$                  & $X_5, X_6$                                                                                    \\
                        22& $R_5:= \URF_2(X_5) \oplus Y_5, R_6:= \URF_2(X_6) \oplus Y_6$                & $R_5, R_6$                                                                            \\
                        23& $S_5||T_5 :=        \URP(X_5 \oplus \URF_1(R_5) || R_5)$                                                    & $S_5$                                                                                         \\
                        24& $S_6||T_6 :=        \URP(X_6 \oplus \URF_1(R_6) || R_6)$                                                            & $S_6$                                                                                         \\
                        25&                                                                                                                                                                                                                                                                             & $A_5$                                                                                 \\
                        26& $X_7:=\URF_3(Y_7) \oplus Z_7, X_8:=\URF_3(Y_8) \oplus Z_8$                  & $X_7,X_8$                                                                                     \\
                        27& $R_7:= \URF_2(X_7) \oplus Y_7, R_8:= \URF_2(X_8) \oplus Y_8$                & $R_7, R_8$                                                                                            \\
                        28& $S_7||T_7 :=        \URP(X_7 \oplus \URF_1(R_7) || R_7)$                                                    & $S_7$                                                                                         \\
                        29& $S_8||T_8 :=        \URP(X_8 \oplus \URF_1(R_8) || R_8)$                                                            & $S_8$                                                                                         \\
                        30& $A_7:= \URF_4(Z_7) \oplus Y_7$                                                                                                                                              & $A_7$                                                                                         \\
\end{tabular}
\item[Consistency Check.] Check if the following equations hold:
\begin{align}                   
        \text{(i) Chain Equalities: } \qquad & \text{for $i=1,2, \ldots, 8$ we have:} \nonumber \\
                               & \URF_1(R_i) = S_i \oplus X_i, \,\,\, \URF_2(X_i) = R_i \oplus Y_i, \,\,\, \URF_3(Y_i) = X_i \oplus Z_i,\nonumber \\
                               & \URF_4(Z_i) = Y_i \oplus A_i, \,\,\, \URF_5(A_i) = Z_i \oplus S_i, \,\,\, \URF_6(S_i) = A_i \oplus T_i \nonumber \\
        \text{(ii) Equalities: } \qquad & X_5 = X_6, \,\,\,  X_7 = X_8, \,\,\, A_7 = A_5, \,\,\,  A_6 = A_3 \nonumber
\end{align}
If all the above equalities hold, output 1, else output 0.
\end{description}

\paragraph{Intuition behind the attack} Figure~\ref{figAttack} provides a picture of the dependencies between chains that have to be defined consistently by the simulator.
\begin{figure}[ht]
  \begin{center}
    \includegraphics[scale=0.5]{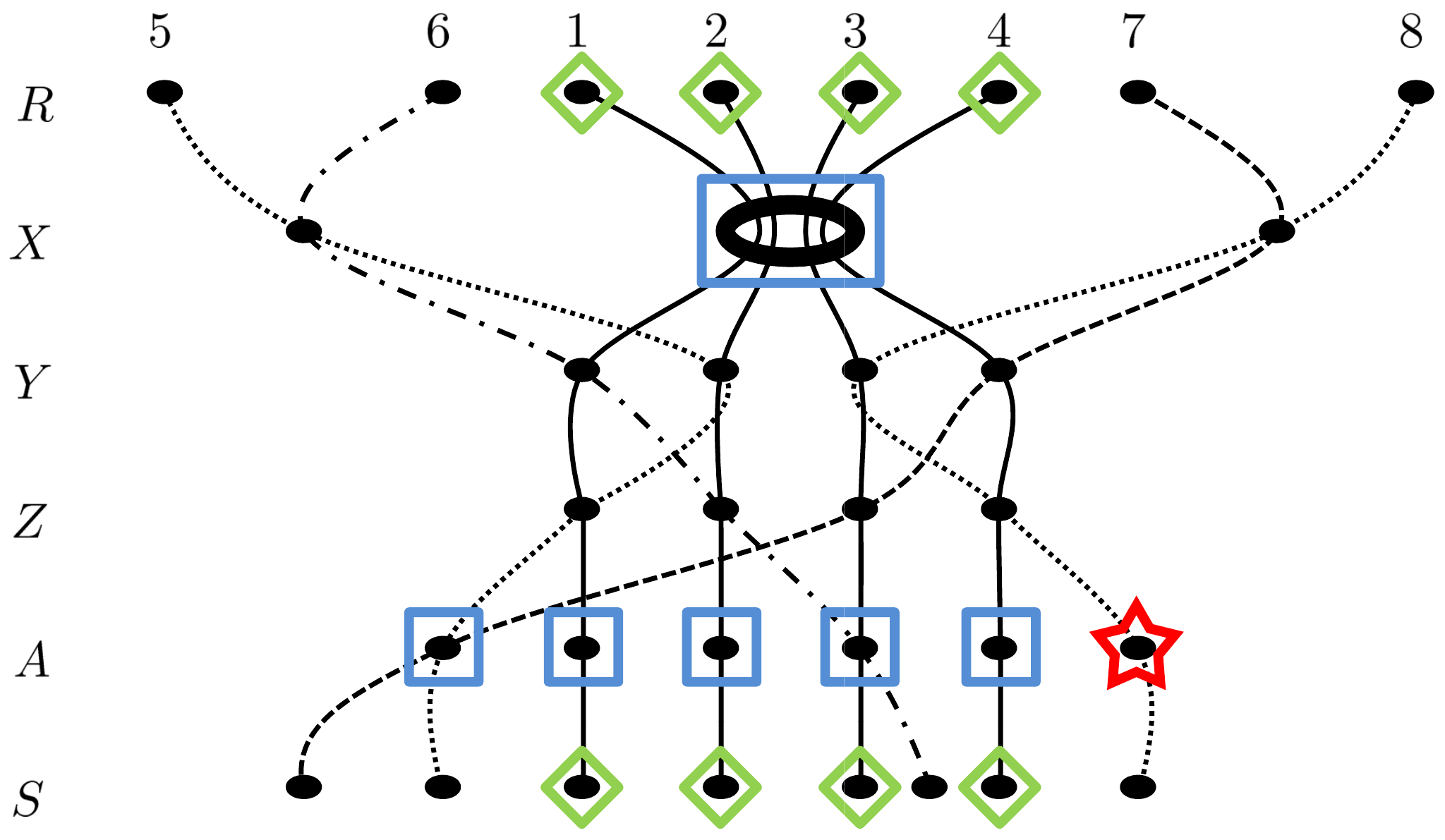}
  \end{center}
  \caption{Illustration of the more general attack.}
  \label{figAttack}
\end{figure}
The intuition behind this attack is as follows: It seems that any simulator that completes chains one after another does not succeed: no matter in which order it completes the chains, it seems to always end up in a situation where some remaining chain cannot be completed consistently.
        
\end{document}